\documentclass[a4paper, UKenglish, cleveref, autoref, thm-restate]{lipics-v2021}
\pdfoutput=1 

\hideLIPIcs  

\title{Integer Linear-Exponential Programming in~\np\\ by Quantifier Elimination
}

\titlerunning{Integer Linear-Exponential Programming in NP by Quantifier Elimination}

\author{Dmitry Chistikov\footnote{During the work on this paper, DC was a visitor to the Max Planck Institute for Software Systems (MPI-SWS), Kaiserslautern and Saarbr\"ucken, Germany, a visiting fellow at St~Catherine's College and a visitor to the Department of Computer Science at the University of Oxford, United Kingdom.}}{Centre for Discrete Mathematics and its Applications (DIMAP) \&\newline{}
Department of Computer Science, University of Warwick, Coventry, UK}{d.chistikov@warwick.ac.uk}{0000-0001-9055-918X}{Supported in part by the Engineering and Physical Sciences Research Council [EP/X03027X/1].}

\author{Alessio Mansutti}{IMDEA Software Institute, Madrid, Spain}{alessio.mansutti@imdea.org}{0000-0002-1104-7299}{funded by the Madrid Regional Government (César Nombela grant 2023-T1/COM-29001), and by MCIN/AEI/10.13039/501100011033/FEDER, EU
(grant  PID2022-138072OB-I00).}

\author{Mikhail R.~Starchak}{St.~Petersburg State University, Russia}{m.starchak@spbu.ru}{0000-0002-2288-9483}{Supported by the Russian Science Foundation, project 23-71-01041.}

\authorrunning{D.~Chistikov, A.~Mansutti, and M.~Starchak} 

\Copyright{Dmitry Chistikov, Alessio Mansutti, Mikhail R.~Starchak} 



\relatedversion{} 

\nolinenumbers 

\EventEditors{John Q. Open and Joan R. Access}
\EventNoEds{2}
\EventLongTitle{42nd Conference on Very Important Topics (CVIT 2016)}
\EventShortTitle{CVIT 2016}
\EventAcronym{CVIT}
\EventYear{2016}
\EventDate{December 24--27, 2016}
\EventLocation{Little Whinging, United Kingdom}
\EventLogo{}
\SeriesVolume{42}
\ArticleNo{23}

\usepackage[utf8]{inputenc}
\usepackage[T1]{fontenc}
\usepackage{xspace}

\usepackage{amsmath,amssymb,amsfonts,amsthm,amscd}
\usepackage{mathtools}
\usepackage{hyperref}
\usepackage{stmaryrd}
\usepackage{etoolbox}
\usepackage{float}
\usepackage{bm}

\usepackage{tikz}
\usetikzlibrary{decorations.pathreplacing,calc,arrows}
\newcommand{\tikzmark}[1]{\tikz[overlay,remember picture] \node (#1) {};}

\newcommand*{\AddNote}[4]{%
    \begin{tikzpicture}[overlay, remember picture]
        \draw [decoration={brace,amplitude=0.4em},decorate,ultra thick,gray]
          ($(#3)!(#2.south)!($(#3)-(0,1)$)$) --  
          ($(#3)!(#1.south)!($(#3)-(0,1)$)$)  
                node [align=center, text width=2.5cm, pos=0.5, right=-1.9cm] {\rotatebox{90}{\textsc{#4}}};
    \end{tikzpicture}
}%

\usepackage{tikzsymbols}

\usetikzlibrary{arrows}
\usetikzlibrary{matrix}
\usetikzlibrary{fit}
\usetikzlibrary{automata}
\usetikzlibrary{positioning}
\usetikzlibrary{calc}
\usetikzlibrary{chains}
\usetikzlibrary{shapes.geometric}
\usetikzlibrary{decorations}
\usetikzlibrary{backgrounds,scopes} 

\definecolor{newblue}{rgb}{0.1,0.1,0.7}

\usepackage{colortbl}

\usepackage{algorithm}
\usepackage[noend]{algpseudocode}

\algnewcommand\algorithmiccontext{\textbf{Context:}}
\algnewcommand\Context{\item[\algorithmiccontext]}
\algnewcommand\algorithmicbranchoutput{\textbf{Branch Output:}}
\algnewcommand\BranchOutput{\item[\algorithmicbranchoutput]}
\algnewcommand\algorithmicndbranchoutput{\textbf{Output of each branch ($\beta$):}}
\algnewcommand\NDBranchOutput{\item[\algorithmicndbranchoutput]}
\algnewcommand\algorithmicglobaloutput{\textbf{Global output:}}
\algnewcommand\GlobalOutput{\item[\algorithmicglobaloutput]}
\algnewcommand\algorithmicglobalspec{\textbf{Ensuring:}}
\algnewcommand\GlobalSpec{\item[\algorithmicglobalspec]}
\algnewcommand\algorithmicswitch{\textbf{switch}}
\algnewcommand\algorithmiccase{\textbf{case}}
\algnewcommand\algorithmicforeach{\textbf{foreach}}
\algnewcommand\algorithmicnondet{\textbf{nondet}}
\algnewcommand\algorithmicor{\textbf{or}}
\algnewcommand\algorithmicassert{\textbf{assert}}
\algnewcommand\algorithmiclet{\textbf{let}}
\algnewcommand\Assert[1]{\State \algorithmicassert(#1)}
\algdef{SE}[SWITCH]{Switch}{EndSwitch}[1]{\algorithmicswitch\ #1\ \algorithmicdo}{\algorithmicend\ \algorithmicswitch}%
\algtext*{EndSwitch}%
\algdef{SE}[CASE]{Case}{EndCase}[1]{\algorithmiccase\ #1\ \textbf{:}}{\algorithmicend\ \algorithmiccase}%
\algtext*{EndCase}%
\algdef{S}[FOR]{ForEach}[1]{\algorithmicforeach\ #1\ \algorithmicdo}%
\algdef{SE}[NONDET]{Nondet}{EndNondet}{\algorithmicnondet\ }{\algorithmicend\ \algorithmicnondet}%
\algtext*{EndNondet}%
\algdef{SE}[Or]{Or}{EndOr}{\algorithmicor\ }{\algorithmicend\ \algorithmicor}%
\algtext*{EndOr}%
\algdef{SE}[LET]{Let}{EndLet}{\algorithmiclet\ }{\algorithmicend\ \algorithmiclet}%
\algtext*{EndLet}%

\AtBeginEnvironment{algorithm}{\linespread{1.05}}

\newcounter{stepnum}
\newenvironment{specification}[1][htb]{%
  \captionsetup[algorithm]{name=Specification}
  \renewcommand{\thealgorithm}{\arabic{stepnum}}
  \begin{algorithm}[#1]%
  }{\end{algorithm}
    \stepcounter{stepnum}
    \renewcommand{\thealgorithm}{\arabic{algorithm}}
  }

\usepackage{todonotes}
\colorlet{dmitry}{green!50!black!80}
\colorlet{alessio}{blue!70!black!80}
\definecolor{mikhail}{HTML}{9F2D20}



\newcommand{\negspace}{}

\newcommand{\newextmathcommand}[2]{%
    \newcommand{#1}{\ensuremath{#2}\xspace}
}

\newcommand{\renewextmathcommand}[2]{%
    \renewcommand{#1}{\ensuremath{#2}\xspace}
}

\makeatletter
\newcommand{\labeltext}[2]{%
  #1%
  \@bsphack%
  \csname phantomsection\endcsname 
  \def\@currentlabel{#1}{\label{#2}}%
  \@esphack%
}
\makeatother

\makeatletter
\newcommand{\customlabel}[2]{%
  \@bsphack%
  \csname phantomsection\endcsname 
  \def\@currentlabel{#1}{\label{#2}}%
  \@esphack%
}
\makeatother

\makeatletter
\newenvironment{ddescription}%
               {\list{}{\leftmargin=0pt
                        \labelwidth\z@ \itemindent-\leftmargin
                        }}%
               {\endlist}
\makeatother

\theoremstyle{definition}
\newtheorem*{myextpa}{Extension to existential Presburger arithmetic}

\newextmathcommand{\N}{\mathbb{N}}
\newextmathcommand{\Np}{\Nat_+}
\newextmathcommand{\Z}{\mathbb{Z}}
\newextmathcommand{\Q}{\mathbb{Q}}
\newextmathcommand{\R}{\mathbb{R}}

\newextmathcommand{\ptime}{\textup{\textsc{PTime}}\xspace}
\newextmathcommand{\np}{\textup{\textsc{NP}}\xspace}
\newextmathcommand{\pspace}{\textup{\textsc{PSpace}}\xspace}
\newextmathcommand{\nexptime}{\textup{\textsc{NExpTime}}\xspace}
\newextmathcommand{\expspace}{\textup{\textsc{ExpSpace}}\xspace}
\newextmathcommand{\threeexptime}{\textup{\textsc{3ExpTime}}\xspace}
\newextmathcommand{\tower}{\textup{\textsc{Tower}}\xspace}

\renewextmathcommand{\phi}{\varphi}
\renewcommand{\vec}{\bm}

\newextmathcommand{\lcm}{{\rm lcm}}
\newextmathcommand{\totient}{\Phi}

\newcommand{\abs}[1]{\ensuremath{\left|#1\right|}\xspace}
\newcommand{\ceil}[1]{\ensuremath{\left\lceil#1\right\rceil}\xspace}
\newcommand{\floor}[1]{\ensuremath{\left\lfloor#1\right\rfloor}\xspace}
\newcommand{\frpart}[1]{\ensuremath{\left\{#1\right\}}\xspace}

\DeclareFontEncoding{LS2}{}{\noaccents@}
  \DeclareFontSubstitution{LS2}{stix}{m}{n}
  \DeclareSymbolFont{stix@largesymbols}{LS2}{stixex}{m}{n}
  \SetSymbolFont{stix@largesymbols}{bold}{LS2}{stixex}{b}{n}
  \DeclareMathDelimiter{\lBrace}{\mathopen} {stix@largesymbols}{"E8}%
                                            {stix@largesymbols}{"0E}
  \DeclareMathDelimiter{\rBrace}{\mathclose}{stix@largesymbols}{"E9}%
                                            {stix@largesymbols}{"0F}

\newcommand{\defeq}{\coloneqq}
\newcommand{\eqdef}{\defeq}

\newcommand{\sub}[2]{\ensuremath{[#1\,/\,#2]}\xspace}
\newcommand{\vigsub}[2]{\ensuremath{{[\mkern-1mu[#1\mathbin{/}#2]\mkern-1mu]}}\xspace}
\newcommand{\varsofsub}{\mathrm{vars}}

\newcommand{\norm}[1]{\ensuremath{\lVert{#1}\rVert}\xspace}
\newcommand{\onenorm}[1]{\ensuremath{\lVert{#1}\rVert_{1}}\xspace}
\newextmathcommand{\card}{\#}

\newcommand{\linnorm}[1]{\ensuremath{\lVert{#1}\rVert_{\!\mathfrak{L}}}\xspace}

\newextmathcommand{\V}{V_2}

\newextmathcommand{\fterms}{\textit{terms}}
\newextmathcommand{\fmod}{\textit{mod}}
\newextmathcommand{\divides}{\mathrel{|}}
\newextmathcommand{\fdiv}{\textit{div}}
\newextmathcommand{\lst}{\textit{lst}}

\newextmathcommand{\lead}{\ell}
\newextmathcommand{\prevlead}{p}
\newcommand{\adjj}[1]{\mathrm{adj}(#1)}

\definecolor{light-gray}{gray}{0.95}
\newcolumntype{g}{>{\columncolor{light-gray}}r}

\newcommand{\rad}{\text{rad}}
\newcommand{\srad}{\text{S}}
\newcommand{\israd}{\mathbb{S}}

\newcommand{\tocheck}[1]{#1}

\newcommand{\myif}{\textbf{if}\xspace}

\newcommand{\myelse}{\textbf{else}\xspace}
\newcommand{\mycontinue}{\textbf{continue}\xspace}
\newcommand{\myguess}{\textbf{guess}\xspace}

\newextmathcommand{\GaussianQE}{\textsc{GaussQE}}
\newextmathcommand{\ExtGaussianQE}{\textsc{ExtGaussQE}}
\newextmathcommand{\Master}{\textsc{LinExpSat}}
\newextmathcommand{\ElimMaxVar}{\textsc{ElimMaxVar}}
\newextmathcommand{\SolvePrimitive}{\textsc{SolvePrimitive}}

\newcommand{\auxi}{\vec u}
\newcommand{\quot}{\vec q}
\newcommand{\rema}{\vec r}
\newcommand{\xquot}{x'}
\newcommand{\xrema}{z'}

\begin{document}

\maketitle

\begin{abstract}
  This paper provides an \np procedure that decides whether a \emph{linear-exponential system} of constraints has an integer solution.
  Linear-exponential systems extend standard integer linear programs with exponential terms $2^x$ and remainder terms ${(x \bmod 2^y)}$.
  Our result implies that the existential theory of the structure 
  $(\N,0,1,+,2^{(\cdot)},\V(\cdot,\cdot),\leq)$ has an \np-complete satisfiability problem,
  thus improving upon a recent \expspace upper bound.
  This theory extends the existential fragment of Presburger arithmetic
  with the exponentiation function $x \mapsto 2^x$ and the binary predicate 
  $\V(x,y)$ that is true whenever $y \geq 1$ is the largest power of $2$ dividing $x$.

  Our procedure for solving linear-exponential systems
  uses the method of quantifier elimination.
  As a by-product, we modify the classical Gaussian variable elimination into a non-deterministic 
  polynomial-time procedure for integer linear programming (or: existential Presburger arithmetic).
\end{abstract}

\section{Introduction}
\label{sec:introduction}

Integer (linear) programming is the problem of deciding whether a system of
linear inequalities has a solution over the integers ($\Z$).
It is a textbook fact that this problem is \np-complete; however, proof of
membership in \np is not trivial.
It is established~\cite{BT76,P81}
by showing that, if a given system has a solution over $\Z$, then it
also has a \emph{small} solution. 
The latter means that the bit size of all components can be bounded from
above by a polynomial in the bit size of the system.
Integer programming is an important language that can encode many combinatorial problems
and constraints from multiple application domains; see, e.g.,~\cite{50yIP,Schrijver99}.

In this paper we consider more general systems of constraints,
which may contain not only linear inequalities (as in integer programming)
but also constraints of the form $y = 2^x$ (exponentiation base~$2$)
and $z = (x \bmod 2^y)$ (remainder modulo powers of~$2$).
Equivalently, and embedding both new operations into a uniform syntax,
we look at a conjunction of inequalities of the form
\begin{equation}
\label{eq:exp-lin-inequality}
\sum\nolimits_{i=1}^n \Big(a_i \cdot x_i + b_i \cdot 2^{x_i} + \sum\nolimits_{j=1}^n c_{i,j} \cdot (x_i \bmod 2^{x_j})\Big) + d \le 0\,,
\end{equation}
referred to as an \emph{(integer) linear-exponential system}.
In fact, the linear-exponential systems that we consider can also feature equalities $=$ and strict inequalities $<$\,.%
{\interfootnotelinepenalty=10000
\footnote{While equalities are considered for convenience only (they can be encoded with a pair of inequalities~$\leq$), the addition of~$<$ is of interest. Indeed, differently from standard integer programming, one cannot define~$<$ in terms of~$\leq$,
since~$2^y$ is not an integer for $y < 0$.
Observe that~$(x \bmod 2^y) = 0$ when $y < 0$, 
because over the reals~$(a \bmod m) = a - m \floor{\frac{a}{m}}$,
where $\floor{.}$ is the floor function.
}}

Observe that a linear-exponential system of the form $x_1 = 1 \land \bigwedge_{i=1}^n (x_{i+1} = 2^{x_i})$ states that $x_{n+1}$ is the tower of 2s of height~$n$. This number is huge, and makes proving an analogue of the small solution property described above a hopeless task in our setting.
This obstacle was recently shown avoidable~\cite{DraghiciHM24}, however, and
an exponential-space procedure for linear-exponential programs was found,
relying on automata-theoretic methods.
Our main result is that, in fact, the problem belongs to \np.

\begin{theorem}
  \label{thm:buchi-semenov-in-np-z}
  Deciding whether a linear-exponential system over $\Z$ has a
  solution is in \np.
\end{theorem}
We highlight that the choice of the base~$2$ for the exponentials
is for the convenience of exposition:
our result holds for any positive integer base
given in binary as part of the~input. 



As an example showcasing the power of integer linear-exponential systems,
consider computation of discrete logarithm base~$2$:
given non-negative integers~$m,r \in \N$, producing
an $x \in \N$ such that $2^x - r$ is divisible by~$m$.
As sketched in~\cite{Haase20},
this problem is reducible to checking feasibility
(existence of solutions)
of at most $\log m$ linear-exponential systems in two variables,
by a binary search for a suitable exponent~$x$.
Hence, improving Theorem~\ref{thm:buchi-semenov-in-np-z} from \np to \ptime
for the case of linear-exponential systems with a \emph{fixed} number of variables 
would require a major breakthrough in number theory.
In contrast, under this restriction, feasibility of standard integer linear programs can be determined in \ptime~\cite{Lenstra83}.

For the authors of this paper,
the main motivation for looking at linear-exponential
systems stems from logic.
Consider the first-order theory of the structure
${(\N,0,1,+,2^{(\cdot)},\V(\cdot,\cdot),\leq)}$,
which we refer to as the \emph{B\"uchi--Semenov arithmetic}.
In this structure, the signature~$(0,1,+,\leq)$ of linear arithmetic is extended with the function 
symbol $2^{(\cdot)}$, 
interpreted as the function $x \mapsto 2^{x}$, 
and the binary predicate symbol $\V$, interpreted as
$\{(x,y) \in \N \times \N : {y} \text{ is the largest}$ 
$\text{power of $2$ that divides $x$}\}$. 
Importantly,
the predicate $\V$ can be replaced in this definition
with the function ${x \bmod 2^y}$, because the two are mutually expressible:
\begin{align*}
  \V(x,y) 
  & \ \iff \exists v  \,\big( 2\cdot y = 2^v \land 2 \cdot (x \bmod 2^v) = 2^v\big),
  \\
  (x \bmod 2^y) = z
  & \ \iff  z \leq 2^y-1 \land \big(x = z \lor \exists u \, ( \V(x-z,2^u) \land 2^y \leq 2^u)\big).
\end{align*}
Above, the subtraction symbol can be expressed in 
the theory in the obvious way (perhaps with the help of an auxiliary existential quantifier for expressing $x-z$). 


B\"uchi--Semenov
arithmetic subsumes logical theories known as B\"uchi arithmetic
and
Semenov arithmetic; see Section~\ref{sec:related}.
%
As a consequence of \Cref{thm:buchi-semenov-in-np-z}, we
show: 

\begin{theorem}
  \label{thm:buchi-semenov-np}
  The satisfiability problem of existential B\"uchi--Semenov
  arithmetic is in \np.
\end{theorem}

Theorems~\ref{thm:buchi-semenov-in-np-z} and~\ref{thm:buchi-semenov-np} improve upon
several results in the literature.
The most recent such result is the exponential-space procedure~\cite{DraghiciHM24} already mentioned above.
In~2023, two elementary decision procedures were developed concurrently and independently for
integer linear programs with exponentiation constraints ($y = 2^x$),
or equivalently for the existential fragment of Semenov arithmetic:
they run in non-deterministic exponential time~\cite{BenediktCM23}
and in triply exponential time~\cite{ZhanEtAl23}, respectively.
Finally, another result subsumed by~\Cref{thm:buchi-semenov-np} is
the membership in \np for the existential fragment of B\"uchi arithmetic~\cite{GuepinHW19}.

Theorem~\ref{thm:buchi-semenov-in-np-z} is established by designing
a non-deterministic polynomial-time decision procedure, which,
unlike those in papers~\cite{DraghiciHM24,GuepinHW19}
but similarly to~\cite{BenediktCM23,ZhanEtAl23},
avoids automata-theoretic methods and instead relies on
\emph{quantifier elimination}.
This is a powerful method (see, e.g.,~\cite{Cooper72} as well as Section~\ref{sec:related}) that
can be seen as a bridge between logic and integer programming.
Presburger~\cite{Pre29} used it to show decidability of linear integer arithmetic
(and Tarski for real arithmetic with addition and multiplication).
For systems of linear equations, quantifier elimination is essentially
Gaussian elimination.
As a little stepping stone, which was in fact one of the springboards for our paper, we extend
the \ptime integer Gaussian elimination procedure by Bareiss~\cite{Bareiss68,Weispfenning97}
into an \np procedure for solving systems of inequalities
over~$\Z$ (thus re-proving membership of integer linear programming in \np).

\subparagraph*{A look ahead.}
The following Section~\ref{sec:related} recalls some relevant related work
on logical theories of arithmetic.
At the end of the paper (Section~\ref{sec:conclusion}) this material is complemented
by a discussion of future research directions, along with several more
key references.

The \np procedure for integer programming is
given as \Cref{algo:gaussianqe} in Section~\ref{sec:gaussian-elimination}.
In this extended abstract, we do not provide a proof of correctness or analysis
of the running time, but instead compare the algorithm
with the classic Gauss--Jordan variable elimination and
with Bareiss' method for systems of equations (that is, equalities).
Necessary definitions and background information are provided
in the Preliminaries (Section~\ref{sec:preliminaries}).

Our core result is an \np procedure for solving linear-exponential systems over~$\N$.
Its pseudocode is split into \Cref{algo:master,algo:elimmaxvar,algo:linearize}.
These are presented in the same imperative style with non-deterministic branching
as \Cref{algo:gaussianqe}, and in fact
\Cref{algo:elimmaxvar} relies on~\Cref{algo:gaussianqe}.
\Cref{sec:procedure} provides a high-level overview of all three algorithms together.
To this end, we introduce several new auxiliary concepts: quotient systems and quotient terms,
delayed substitution, and primitive linear-exponential systems. 
After this, technical details of~\Cref{algo:master,algo:elimmaxvar,algo:linearize} are given.
\Cref{sec:procedure-details} sketches key ideas behind the correctness argument, 
and the text within this section is thus to be read alongside the pseudocode of Algorithms.
An overview of the analysis of the worst-case running time is presented in \Cref{sec:complexity}.
The basic definitions are again those from Preliminaries, and
of particular relevance are the subtleties of the action of term substitutions.

Building on the core procedure,
in \Cref{sec:wrappers}
we show how to solve in \np linear-exponential systems not only over \N but also over \Z
(\Cref{thm:buchi-semenov-in-np-z})
and how to decide B\"uchi--Semenov arithmetic in \np (\Cref{thm:buchi-semenov-np}). 
The modifications to this procedure that enable proving both results for
a different integer base $b > 2$ for the exponentials are given in \Cref{app:inter-def}.

\section{Arithmetic theories of B\"uchi, Semenov, and Presburger}
\label{sec:related}
In this section, we review results on arithmetic theories that are the most relevant to~our~study. 

\emph{B\"uchi arithmetic} is the first-order theory of the structure ${(\N,0,1,+,\V(\cdot,\cdot),\leq)}$.
By the celebrated
  B\"uchi--Bruy\`ere theorem~\cite{BruyereHMV94,Buchi60},
  a set $S \subseteq \N^d$ is definable in
  ${(\N,0,1,+,\V(\cdot,\cdot),\leq)}$ if and only if the
  representation of $S$ as a language over the alphabet 
  $\{0,1\}^d$ is recognisable by a deterministic finite
  automaton (DFA). The theorem is effective, and implies that
  the satisfiability problem for B\"uchi arithmetic is in
  \tower (in fact, \tower-complete)~\cite{Schmitz16,StockmeyerM73}.
  The situation is different for the existential
  fragment of this theory. The satisfiability problem is 
  now~\np-complete~\cite{GuepinHW19}, but 
  existential formulae are less expressive~\cite{HaaseR21}. 
  In particular, this fragment fails to capture the binary language $\{10 , 01\}^*$. 
  Decision procedures for
  B\"uchi arithmetic have been successfully implemented and
  used to automatically prove many results in
  combinatorics on words and number theory~\cite{Shallit22}.

\emph{Semenov arithmetic} is the first-order theory of the structure $(\N,0,1,+,2^{(\cdot)},\leq)$.
Its decidability follows from
  the classical work of Semenov on sparse
  predicates~\cite{Semenov80,Semenov84}, and an explicit
  decision procedure was given by~Cherlin and
  Point~\cite{CherlinP86,Point00}. Similarly to B\"uchi
  arithmetic, Semenov arithmetic is \tower-complete~\cite{ComptonH90,Point07}; however, its existential fragment has
  only been known to be in \nexptime~\cite{BenediktCM23}.
  The paper~\cite{ZhanEtAl23} provides
  applications of this fragment to solving
  systems of string constraints with string-to-integer
  conversion functions. 
  
\emph{B\"uchi--Semenov arithmetic} is a natural combination of these two theories.
Differently from B\"uchi and Semenov arithmetics,
  the $\exists^*\forall^*$-fragment of this logic
  is undecidable~\cite{CherlinP86}. 
  In view of this, the recent result
  showing that the satisfiability problem of existential
  B\"uchi--Semenov arithmetic is in~\expspace~\cite{DraghiciHM24} is surprising. 
  The proof technique, moreover, establishes the membership in~\expspace of the
  extension of existential \mbox{B\"uchi--Semenov} arithmetic with
  arbitrary regular predicates given on input as DFAs. 
  Since this extension can express the
  intersection non-emptiness problem for DFAs, its satisfiability problem is~\pspace-hard~\cite{Kozen77}. 
  The decision procedure of~\cite{DraghiciHM24} was applied
  to give an algorithm for solving real-world instances of word equations with length
  constraints. 

Both first-order theories of the structures $(\N,0,1,+,\leq)$ and $(\Z,0,1,+,\leq)$ are usually referred to as \emph{Presburger arithmetic},
because the decision problems for these theories are logspace inter-reducible, 
meaning that each structure can be interpreted in the other.
The procedures that we propose in this paper
  build upon a version of the quantifier-elimination procedure for 
  the first-order theory of the structure $(\Z,0,1,+,\leq)$. 
  Standard procedures for this theory~\cite{Cooper72,Oppen73,Weispfenning90} 
  are known to be suboptimal when applied to the existential fragment:
  throughout these procedures, the bit size of the numbers in the formulae grow exponentially faster than in, e.g.,~geometric procedures for the theory~\cite{ChistikovHM22}.
  A remedy to this well-known issue was proposed by Weispfenning~\cite[Corollary~4.3]{Weispfenning97}.     
  We~develop~his~observation~in~\Cref{sec:gaussian-elimination}.

\section{Preliminaries}
\label{sec:preliminaries}

We usually write $a,b,c,\dots$ for integers, $x,y,z,\dots$ for integer variables, and $\vec a,\vec b,\vec c,\dots$ and $\vec x,\vec y,\vec z,\dots$ for vectors of those.
By $\vec x \setminus y$ we denote the vector obtained by removing the variable $y$ from $\vec x$.
We denote linear-exponential systems and logical formulae with the letters $\phi,\chi,\psi,\dots$, and write $\phi(\vec x)$ when the (free) variables of $\phi$ are among $\vec x$.

For $a \in \R$, we write $\abs{a}$, $\ceil{a}$, and $\log a$ for the \emph{absolute value}, \emph{ceiling},
and (if $a > 0$) the \emph{binary logarithm} of $a$.
All numbers encountered by our algorithm are encoded in binary;
note that $n \in \N$ can be represented using $\ceil{\log(n+1)}$ bits. For $n,m \in \Z$, denote $[n,m] \coloneqq \{n,n+1,\dots,m\}$.
The set~$\N$ of non-negative integers contains~$0$.

\subparagraph*{Terms.}
\addcontentsline{toc}{subsection}{\emph{Definition:} Linear-exponential terms}
As in~\Cref{eq:exp-lin-inequality}, 
a \emph{(linear-exponential) term} is an expression of the form 
\begin{equation}
  \label{eq:a-term}
  \sum\nolimits_{i=1}^n \big(a_i \cdot x_i + b_i \cdot 2^{x_i} + \sum\nolimits_{j=1}^n c_{i,j} \cdot (x_i \bmod 2^{x_j})\big) + d,
\end{equation}
where $a_i,b_i,c_{i,j} \in \Z$ are the \emph{coefficients} of the term and $d \in \Z$ is its \emph{constant}. 
If all $b_i$~and~$c_{i,j}$ are zero then the term is said to be \emph{linear}.
We denote terms by the letters~$\rho,\sigma,\tau,\dots$, and write $\tau(\vec x)$ if all variables of the term~$\tau$ are in $\vec x$. 
For a term $\tau$ in~\Cref{eq:a-term},
its $1$-norm is $\onenorm{\tau} \coloneqq \sum_{i=1}^n (\abs{a_i} + \abs{b_i} + \sum_{j=1}^n \abs{c_{i,j}}) + \abs{d}$.

We use the words `system' and `conjunction' of constraints interchangeably.
While equalities and inequalities of a linear-exponential system are always of the form $\tau = 0$, $\tau \leq 0$, and $\tau < 0$, 
for the convenience of exposition we often rearrange left- and right-hand sides and write, e.g., $\tau_1 \leq \tau_2$.
In our procedures, linear-exponential systems 
may contain equalities, inequalities, and also \emph{divisibility constraints} $d \divides \tau$, 
where $\tau$ is a term as in~\Cref{eq:a-term}, $d \in \Z$ is non-zero, and $\divides$ is the \emph{divisibility predicate},
$\{(d,n) \in \Z \times \Z : \text{$n = kd$ for some $k \in \Z$} \}$. 
We write $\fmod(\phi)$ for the (positive) least common multiple of all divisors $d$ appearing in divisibility constraints 
$d \divides \tau$ of a system~$\phi$.
For purely syntactic reasons, it is sometimes convenient to see a divisibility constraint $d \divides \tau_1 - \tau_2$ 
as a \emph{congruence} $\tau_1 \equiv_d \tau_2$, where $d \geq 1$ with no loss of generality.
We use the term \emph{divisibility constraint} also for these congruences.

\negspace
\subparagraph*{Substitutions.}
\addcontentsline{toc}{subsection}{\emph{Definitions:} Substitutions}
Our procedure uses several special kinds of substitutions. Consider a linear-exponential system $\phi$, a term $\tau$, two variables $x$ and $y$, and $a \in \Z \setminus \{0\}$.
\begin{itemize}
\item We write $\phi\sub{\tau}{x}$ for the system obtained from $\phi$ by replacing every \emph{linear occurrence} of~$x$ outside modulo operators with $\tau$. To clarify,~this substitution only modifies the ``$a_i \cdot x_i$'' parts of the term in~\Cref{eq:a-term}, but not the ``$c_{i,j} \cdot (x_i \bmod 2^{x_j})$'' parts.
\item We write $\phi\sub{\tau}{x \bmod 2^y}$ and $\phi\sub{\tau}{2^x}$ for the system obtained from $\phi$ by replacing with $\tau$ every occurrence of $(x \bmod 2^y)$ and $2^x$, respectively.
\item We write $\phi\vigsub{\frac{\tau}{a}}{x}$ for the \emph{vigorous substitution} of $\frac{\tau}{a}$ for $x$. This substitution works as follows. \textbf{1:} Multiply every equality and inequality by $a$, flipping the signs of inequalities if $a < 0$; this step also applies to inequalities in which $x$ does not occur.
\textbf{2:} Multiply both sides of divisibility constraints in which $x$ occurs by $a$, i.e., $d \divides \tau$ becomes $a \cdot d \divides a \cdot \tau$.
\textbf{3:} Replace with $\tau$ every linear occurrence of $a \cdot x$ outside modulo operators.
Note that, thanks to step~1, each coefficient of $x$ in the system can be factorised as $a \cdot b$ for some~${b \in \Z}$.
\end{itemize}
We sometimes see substitutions~$\sub{\tau}{\tau'}$ as first-class citizens: functions mapping systems to systems. When adopting this perspective, $\phi\sub{\tau}{\tau'}$ is seen as a function application.
 

\section{Solving systems of linear inequalities over $\Z$}
\label{sec:gaussian-elimination}

In this section we present \Cref{algo:gaussianqe} (\GaussianQE),
a non-deterministic polynomial time quantifier elimination (QE) procedure for solving systems of linear inequalities over $\Z$, or in other words for
integer programming.
A constraint (equality, inequality, or divisibility) is \emph{linear}
if it only contains linear terms, as defined in Section~\ref{sec:preliminaries}.
Our algorithm assumes that all inequalities are non-strict ($\tau \le 0$).

We already mentioned in Section~\ref{sec:introduction} that
$\textsc{Integer Programming} \in \np$ is a standard result.
Intuitively, the range of each variable is infinite,
which necessitates a proof that a suitable (and \emph{small}) range
suffices; see, e.g.,~\cite{BT76,P81,GS78}. 
Methods developed in these references, however, do not enjoy
the flexibility of quantifier elimination: e.g.,
they either do not preserve formula equivalence
or are not actually removing quantifiers.

{
\begin{algorithm}[t]
  \caption{\GaussianQE: Gauss--Jordan elimination for integer programming.}
  \label{algo:gaussianqe}
  \begin{algorithmic}[1]
    \Require $\vec x:{}$ sequence of variables; \ $\phi(\vec x, \vec z) :{}$ system of linear constraints without $<$.
    \NDBranchOutput
      system $\psi_\beta(\vec z)$ of linear constraints. 
    \GlobalSpec
      $\bigvee_{\beta} \psi_\beta$ is equivalent to $\exists \vec x \,\phi$.
    \medskip
    \State replace each inequality $\tau \le 0$ in $\phi$ with $\tau + y = 0$, where $y$ is a fresh slack variable \label{gauss:introduce-slack}
    \State $\lead \gets 1$; \quad
           $s \gets ()$ \label{gauss:bareissfactors}
    \Comment{$s$ is an empty sequence of substitutions}
    \ForEach{$x$ in $\vec x$} \label{gauss:mainloop}
       \If{no equality of $\phi$ contains $x$} \mycontinue \EndIf \label{gauss:triviality-check}
       \State \myguess $a x + \tau = 0 \ \text{(with $a \ne 0$)} \gets \text{equality in $\phi$ that contains $x$}$ \label{gauss:guess-equation}
       \State $\prevlead \gets \lead$; \quad $\lead \gets a$
       \Comment{previous and current lead coefficients}
       \label{gauss:new-factor}
       \If{$\tau$ contains a slack variable $y$ not assigned by $s$}\label{gauss:branch}
	  \State \myguess $v \gets \text{integer in $[0, |a| \cdot \fmod(\phi) - 1]$}$ \label{gauss:guess-slack}
	  \State append $\sub{v}{y}$ to $s$ \label{gauss:append-seq}
       \EndIf
       \State $\phi \gets \phi \vigsub{\frac{-\tau}{a}}{x}$ \label{gauss:vig}\label{gauss:vigorous}
       \State divide each constraint in $\phi$ by \prevlead \label{gauss:divide}
       \Comment{in divisibility constraints, both sides are affected}
       \State $\phi \gets \phi \land (a \divides \tau)$ \label{gauss:restore}
    \EndFor \label{gauss:endmainloop}
    \ForEach{equality $\eta = 0$ of $\phi$ that contains some slack variable $y$ not assigned by $s$} \label{gauss:unslackloop}
       \State replace $\eta = 0$ with $\eta\sub{0}{y} \le 0$ \myif the coefficient at~$y$ is positive \myelse
                                 with $\eta\sub{0}{y} \ge 0$ \label{gauss:drop-slack}
    \EndFor
    \State apply substitutions of $s$ to $\phi$ \label{gauss:apply-seq}
    \ForEach{$x$ in $\vec x$ that occurs in $\phi$} \label{gauss:loop-rem}
       \State \myguess $r \gets \text{integer in $[0, \fmod(\phi)-1]$}$ \label{gauss:guess-rem}
       \State $\phi \gets \phi \sub{r}{x}$ \label{gauss:subst-rem}
    \EndFor
    \State \textbf{return} $\phi$
  \end{algorithmic}
\end{algorithm}%
}%

\begin{restatable}{theorem}{ThmGaussianQENP}
  \label{thm:gaussianQE-in-np}
  \Cref{algo:gaussianqe} (\GaussianQE) runs in non-deterministic polynomial time
  and,
  given a linear system $\phi(\vec x, \vec z)$ and variables $\vec x$,
  produces in each non-deterministic branch~$\beta$ a~linear system $\psi_\beta(\vec z)$
  such that
  $\bigvee_{\beta} \psi_\beta$ is equivalent to $\exists \vec x \,\phi$.
\end{restatable}
\addcontentsline{toc}{subsection}{Theorem~\ref{thm:gaussianQE-in-np} and Algorithm~\ref{algo:gaussianqe} (\GaussianQE)}

%
%
\GaussianQE is based on an observation by Weispfenning,
who drew a parallel between a weak form of QE and Gaussian variable elimination~\cite{Weispfenning97}.
Based on this observation and relying on an insight by Bareiss~\cite{Bareiss68} (to be discussed below),
Weispfenning sketched
a non-deterministic procedure for deciding \textit{closed} existential formulae of 
Presburger arithmetic in polynomial time.
Although the idea of weak QE~\cite{Weispfenning97} has since been developed further~\cite{LasarukS07},
the \np observation has apparently remained not well known.


Due to space constraints, we omit the proof of Theorem~\ref{thm:gaussianQE-in-np}, and explain instead only the key ideas.
We first consider the specification of~\GaussianQE,
in particular non-deterministic branching.
We then recall
the main underlying mechanism:
Gaussian variable elimination
(thus retracing and expanding Weispfenning's observation).
After that, we discuss extension of this mechanism to tackle inequalities over \Z.


\subparagraph*{Input, output, and non-determinism.}
\addcontentsline{toc}{subsection}{Input, output, and non-determinism}
The input to \GaussianQE is a system~$\phi$ of linear constraints,
as well as a sequence $\vec x$ of variables to eliminate.
The algorithm makes non-deterministic guesses
in lines~\ref{gauss:guess-equation}, \ref{gauss:guess-slack}, and
\ref{gauss:guess-rem}.
Output of each branch (of the non-deterministic execution) is specified
at the top: it is a system $\psi_\beta$ of linear constraints, in which
all variables $x$ in $\vec x$ have been eliminated.
For any specific non-deterministic branch, call it $\beta$,
the output system $\psi_\beta$ may not necessarily be equivalent to $\exists \vec x\, \phi$,
but the disjunction of all outputs across all branches must be:
$\bigvee_\beta \psi_\beta$ has the same set of satisfying assignments as $\exists \vec x\, \phi$.%
\footnote{%
 Formally, an assignment is a map $\nu$ from (free) variables to $\Z$. It satisfies
 a constraint if replacing each $z$ in the domain of $\nu$ with $\nu(z)$ makes
 the constraint a true numerical assertion.
}

The number of non-deterministic branches (individual paths through the execution tree)
is usually exponential, but each of them runs in polynomial time.
(This is true for all algorithms presented in this paper.)
If all variables of the input system $\phi$ are included in~$\vec x$, then
each branch returns a conjunction of numerical assertions that evaluates 
to true or false.

\negspace
\subparagraph*{Gaussian elimination and Bareiss' method.} 
\addcontentsline{toc}{subsection}{Gaussian elimination and Bareiss' method}
Consider a system $\phi$ of linear equations (i.e., equalities) over fields, e.g., $\R$ or $\Q$, 
and let $\vec x$ be a vector of variables that we wish to eliminate from $\phi$.
We recall the Gauss--Jordan variable elimination algorithm,
proceeding as follows:
{\makeatletter%
\def\ALG@step%
   {%
   \addtocounter{ALG@line}{1}%
   \addtocounter{ALG@rem}{1}%
   \ifthenelse{\equal{\arabic{ALG@rem}}{\ALG@numberfreq}}%
      {\setcounter{ALG@rem}{0}\alglinenumber{0\arabic{ALG@line}}}
      {}%
   }%
\makeatletter
\begin{algorithmic}[1]
  \State $\lead \gets 1$ \customlabel{01}{ptime-gauss:coefficient}
  \ForEach{$x$ in $\vec x$} \customlabel{02}{ptime-gauss:mainloop}
     \If{no equality of $\phi$ contains $x$} \mycontinue \EndIf
     \State \textbf{let} $a x + \tau = 0 \ \text{(with $a \ne 0$)} \gets \text{an arbitrary equality in $\phi$ that contains $x$}$ \customlabel{04}{ptime-gauss:guess-equation}
     \State $\prevlead \gets \lead$; \quad $\lead \gets a$
     \customlabel{05}{ptime-gauss:update-coeffs}
     \State $\phi \gets \phi \vigsub{\frac{-\tau}{a}}{x}$ \customlabel{06}{ptime-gauss:substitution}
     \State divide each constraint in $\phi$ by \prevlead \customlabel{07}{ptime-gauss:divide}
  \EndFor \label{ptime-gauss:endmainloop}
  \State \textbf{return} $\phi$
\end{algorithmic}%
}%
By removing from this code all lines involving $\prevlead$ and $\lead$ (lines~\ref{ptime-gauss:coefficient},
\ref{ptime-gauss:update-coeffs} and~\ref{ptime-gauss:divide}), we obtain a naive version
of the procedure: an equation is picked in line~\ref{ptime-gauss:guess-equation} and used to remove one of its occurring variables in line~\ref{ptime-gauss:substitution}.
Indeed,
applying a vigorous substitution $\vigsub{\frac{-\tau}{a}}{x}$ to an equality
$b x + \sigma = 0$ is equivalent to first multiplying this equality by the lead
coefficient~$a$
and then subtracting $b \cdot (a x + \tau) = 0$.
The result is $- b \tau + a \sigma = 0$, and $x$ is eliminated. 

An insightful observation due to Bareiss~\cite{Bareiss68} is that, after multiple iterations,
coefficients accumulate non-trivial common factors.
Lines~\ref{ptime-gauss:coefficient}, \ref{ptime-gauss:update-coeffs}, and \ref{ptime-gauss:divide} take advantage of this.
Indeed, line~\ref{ptime-gauss:divide}
divides every equation by such a common factor.
Importantly, if all numbers in the input system~$\phi$ are integers, then
the division is without remainder.
To show this, Bareiss uses a linear-algebraic argument
based on an application of 
the Desnanot--Jacobi identity (or, more generally, Sylvester's identity) for determinants%
~\cite{Bareiss68,Dodgson1867,KarapiperiRR}.
Over \Q, this makes it possible to perform
Gaussian elimination
(its `fraction-free one-step' version)
in~\ptime.
(This is not the only
 polynomial-time method; cf.~\cite[Section~3.3]{Schrijver99}.)

Gaussian elimination for systems of equations can be extended to solving over~\Z,
by introducing divisibility constraints: line~\ref{ptime-gauss:substitution} becomes
$\phi \gets \phi \vigsub{\frac{-\tau}{a}}{x} \land (a \divides \tau)$.
However, while the running time of the procedure remains polynomial,
its effect becomes more modest:
the procedure reduces a system of linear equations over \Z to an equivalent
system of equations featuring variables not in~$\vec x$ and
multivariate linear congruences that may still contain variables 
from~$\vec x$. To completely eliminate~$\vec x$, further computation is
required. For our purposes, non-deterministic guessing is a good enough solution
to this problem; see the final \textbf{foreach} loop in
lines~\ref{gauss:loop-rem}--\ref{gauss:subst-rem} of~\GaussianQE.

\negspace
\subparagraph*{From equalities to inequalities.}
\addcontentsline{toc}{subsection}{From equalities to inequalities}
\GaussianQE extends Bareiss' method
to systems of inequalities over~$\Z$.
As above, the method allows us to control the (otherwise exponential) growth
of the bit size of numbers.
Gaussian elimination is, of course, still at the heart of the algorithm
(see lines~\ref{gauss:bareissfactors}--\ref{gauss:new-factor}, \ref{gauss:vigorous}, and \ref{gauss:divide}),
and we now discuss two modifications:
\begin{itemize}
\item\label{we-slack}
Line~\ref{gauss:introduce-slack} introduces \emph{slack variables} ranging over~$\N$.
These are
internal to the procedure and are removed at the end (lines~\ref{gauss:unslackloop}--\ref{gauss:apply-seq}).
\item\label{we-guess}
In line~\ref{gauss:guess-equation}
the equality $a x + \tau = 0$ is selected non-deterministically.
\end{itemize}
The latter modification is required 
for the correctness (more precisely: completeness) of~\GaussianQE. Geometrically,
for a satisfiable system of inequalities over $\Z$
consider the convex polyhedron of all solutions over~\R first.
At least one of solutions over~$\Z$ must lie in or near a facet of this polyhedron.
Line~\ref{gauss:guess-equation} of~\Cref{algo:gaussianqe} attempts to guess this facet.
The amount of slack guessed in line~\ref{gauss:guess-slack} corresponds to
the distance from the facet.
Observe that if $a x + \tau = 0$ corresponds to an equality of the original system $\phi$, 
then every solution of $\phi$ needs to satisfy $a x + \tau = 0$ exactly, 
and so there is no slack (lines~\ref{gauss:guess-slack}--\ref{gauss:append-seq} are not taken).

The values chosen for the slack variables in line~\ref{gauss:guess-slack}
have, in fact, a counterpart
in the standard decision procedures for Presburger arithmetic.
When the latter pick a term $\rho$ to substitute,
the substitutions in fact introduce $\rho + k$ for $k$ ranging in some $[0, \ell]$, where $\ell$ depends on $\fmod(\phi)$.
The amount of slack considered in~\GaussianQE corresponds to
these values of $k$.
(Because of this parallel, making the range of guesses in line~\ref{gauss:guess-slack} symmetric, i.e.,
$|v| \le |a| \cdot \fmod(\phi) - 1$,
extends our procedure to the entire existential Presburger arithmetic.)

\section{Solving linear-exponential systems over $\N$: an overview}
\label{sec:procedure}
\addtocontents{toc}{This section introduces our main decision procedure.}

In this section we give an overview of our 
non-deterministic procedure
to solve linear-exponential systems over~$\N$.
The procedure is split into
\Cref{algo:master,algo:elimmaxvar,algo:linearize}. 
A more technical analysis of these algorithms 
is given later in~\Cref{sec:procedure-details}.

Whenever non-deterministic
\Cref{algo:gaussianqe,algo:master,algo:elimmaxvar,algo:linearize}
call one another,
the return value is always just the output of a single branch,
rather than (say) the disjunction over all branches.

\negspace
\subparagraph*{\Cref{algo:master} (\Master).}
\addcontentsline{toc}{subsection}{Introduction to Algorithm \ref{algo:master} (\Master)}
  This
  is the main procedure. It takes as input a
  linear-exponential system~$\phi$ without divisibility
  constraints and decides whether~$\phi$ has a solution
  over~$\N$. The procedure relies on first
  (non-deterministically) fixing a linear ordering~$\theta$ on the exponential terms $2^x$ occurring in~$\phi$ (line~\ref{algo:master:guess-order}). For technical convenience, this ordering contains a term $2^{x_0}$, with $x_0$ fresh variable, and sets~$2^{x_0} = 1$. Variables are iteratively eliminated starting from
  the one corresponding to the leading exponential term in~$\theta$ (i.e., the biggest one), until reaching~$x_0$
  (lines~\ref{algo:master:outer-loop}--\ref{algo:master:end-loop}).
  The elimination of each variable is performed by first
  rewriting the system
  (in lines~\ref{algo:master:inner-loop}--\ref{algo:master:end-inner-loop})
  into a form admissible for \Cref{algo:elimmaxvar}
  discussed below.
  This rewriting introduces new variables,
  which will never occur in exponentials throughout
  the entire procedure
  and are later eliminated when the procedure reaches~$x_0$. 
  Overall, the termination of the procedure
  is ensured by the decreasing number of
  exponentiated variables.
  After \Master rewrites~$\phi$,
  it calls \Cref{algo:elimmaxvar} to eliminate the currently biggest variable.

\negspace
\subparagraph*{\Cref{algo:elimmaxvar} (\ElimMaxVar).}
\addcontentsline{toc}{subsection}{Introduction to Algorithm \ref{algo:elimmaxvar} (\ElimMaxVar)}
  This
  procedure takes as input an ordering~$\theta$, a
  \emph{quotient system induced by~$\theta$}
  and a~\emph{delayed
  substitution}. Let us introduce these notions. 
  \negspace
  \begin{ddescription}
  \item[Quotient systems.]
  \addcontentsline{toc}{subsubsection}{\emph{Definition:} Quotient systems and quotient terms}
  Let $\theta(\vec x)$ be the ordering ${2^{x_n} \geq
  2^{x_{n-1}} \geq \dots \geq 2^{x_0} = 1}$, where $n \geq
  1$. A~\emph{quotient system induced by}~$\theta$ 
  is a system~$\phi(\vec x, \vec x', \vec z')$
  of equalities, inequalities, and divisibility constraints $\tau \sim
  0$, where ${\sim} \in \{{<},{\leq},=,\equiv_d : d \geq 1\}$ and
  $\tau$ is an \emph{quotient term (induced by $\theta$)}, that is,
  a term of the form
  \begin{equation*}
    a \cdot 2^{x_n} + f(\vec x') \cdot 2^{x_{n-1}} + b \cdot x_{n-1} + \tau'(x_{0},\dots,x_{n-2}, \vec z')\,,
  \end{equation*}
  where $a,b \in \Z$, $f(\vec x')$ is a linear term, and
  $\tau'$ is a linear-exponential term in which the variables
  from $\vec z'$ do not occur exponentiated. 
  Furthermore, for every variable $z'$ in $\vec z'$, the quotient system $\phi$ features the inequalities $0 \leq z' < 2^{x_{n-1}}$.
  The variables
  in~$\vec x$, $\vec x'$ and $\vec z'$ form three disjoint
  sets, which we call the \emph{exponentiated
  variables}, the \emph{quotient variables} and
  the~\emph{remainder variables} of the system~$\phi$, respectively. We
  also refer to the term $b \cdot x_{n-1} +
  \tau'(x_{0},\dots,x_{n-2}, \vec z')$ as the \emph{least
  significant part} of the quotient term $\tau$. 

  Importantly, quotient terms are \textbf{not} linear-exponential terms.
  \end{ddescription}
  Here is an example of a quotient system induced by $2^{x_{3}} \geq 2^{x_{2}} \geq 2^{x_1} \geq 2^{x_0} = 1$, 
  and having quotient variables $\vec x' = (x_1',x_2')$ and remainder variables $\vec z' = (z_1',z_2')$ 
  \begin{align*} 
    -2^{x_3} +(2 \cdot x_1' - x_2' -1) \cdot {}&2^{x_2} + \big\{-2 \cdot x_2 + 2^{x_1}-(z_1' \bmod 2^{x_1})\big\} \leq 0,
    &&0 \leq z_1' < 2^{x_2},\\ 
    x_1' \cdot {}&2^{x_2} + \big\{x_1+z_2'-5\big\} = 0,
    &&0 \leq z_2' < 2^{x_2}.
  \end{align*}
  The curly brackets highlight the least significant parts of two terms of the system, the other parts being $\pm z_1'$ and $\pm z_2'$ stemming from the inequalities 
  on the right.
  \negspace
  \begin{ddescription}
  \item[Delayed substitution.]
\addcontentsline{toc}{subsubsection}{\emph{Definition:} Delayed substitution}
    This is a substitution of
    the form $\sub{x' \cdot 2^{x_{n-1}} + z'}{x_n}$, where~$2^{x_n}$ is the leading exponential term of $\theta$. Our
    procedure delays the application of this substitution
    until $x_n$ occurs linearly in the system~$\phi$. One can
    think of this substitution as an equality $(x_n = x' \cdot
    2^{x_{n-1}} + z')$ in~$\phi$ that must not be manipulated for the time being.
  \end{ddescription}%
  Back to~\ElimMaxVar, given a
  quotient system~$\phi(\vec x,\vec x', \vec z')$
  induced by~$\theta$ and the delayed substitution
  $\sub{x' \cdot 2^{x_{n-1}} + z'}{x_n}$, the goals of this
  procedure are to \labeltext{(i)}{elim:i1} eliminate the
  quotient variables~$\vec x' \setminus x'$;
  \labeltext{(ii)}{elim:i2} eliminate all occurrences of the
  leading exponential term~$2^{x_n}$ of $\theta$ and apply 
  the delayed substitution to eliminate the variable $x_n$;
  \labeltext{(iii)}{elim:i3} finally, remove $x'$. 
  Upon exit, \ElimMaxVar gives back
  to~\Master a (non-quotient)
  linear-exponential system where $x_n$ has been
  eliminated; i.e., a system with one fewer exponentiated
  variable.

  For steps~\ref{elim:i1} and~\ref{elim:i3}, the procedure relies on the
  \Cref{algo:gaussianqe} (\GaussianQE) for eliminating variables in systems of inequalities,
  from \Cref{sec:gaussian-elimination}. This is where flexibility of QE is important: 
  in line~\ref{line:elimmaxvar-gauss-1} some variables are eliminated and some are not.
  Step~\ref{elim:i2} is instead implemented
  by~\Cref{algo:linearize}. 

\negspace
\subparagraph*{\Cref{algo:linearize} (\SolvePrimitive).}
\addcontentsline{toc}{subsection}{Introduction to Algorithm \ref{algo:linearize} (\SolvePrimitive)}
  The goal of this procedure is to rewrite a system of constraints
  where $x_n$ occurs exponentiated
  with another system where all constraints are linear.
  The specification of the procedure restricts the output
  further.
  At its core, \SolvePrimitive
  tailors
  Semenov's proof of
  the decidability of the first-order theory of the
  structure~$(\N,0,1,+,2^{(\cdot)},\leq)$~\cite{Semenov84} to
  a small syntactic fragment, which we now define.
  \negspace
  \begin{ddescription}
    \item[Primitive linear-exponential systems.] 
    \addcontentsline{toc}{subsubsection}{\emph{Definition:} Primitive linear-exponential systems}
    Let $u,v$ be two variables. A linear-exponential system
    is said to be \emph{$(u,v)$-primitive} whenever all its
    (in)equalities and divisibility constraints are of the
    form $a \cdot 2^{u} + b \cdot v + c \sim 0$, with~$a,b,c
    \in \Z$ and~${\sim} \in \{{<},{\leq},=,\equiv_d : d \geq 1\}$. 
  \end{ddescription}
  The input to \SolvePrimitive is
  a $(u,v)$-primitive linear-exponential system.
  This procedure
  removes all occurrences of $2^u$ in favour of linear
  constraints, working under the assumption that~$u \geq
  v$. This condition is ensured when \ElimMaxVar
  invokes \SolvePrimitive. The variable
  $u$ of the primitive system in the input corresponds to
  the term~$x_n - x_{n-1}$, and the variable $v$ stands for the
  variable $x'$ in the delayed
  substitution~$\sub{x' \cdot 2^{x_{n-1}} + z'}{x_n}$.
  \ElimMaxVar ensures that $x_{n} - x_{n-1} \geq x'$. 

\section{Algorithms~\ref{algo:master}--\ref{algo:linearize}: a walkthrough}
\label{sec:procedure-details}
\addtocontents{toc}{This section explains the intuition behind the procedure (why it is is correct).}

Having outlined the interplay
between~\Cref{algo:master,algo:elimmaxvar,algo:linearize},
we move to their technical description, 
and present the key 
ideas required to establish the correctness of 
our procedure for solving linear-exponential systems over $\N$.

{
\begin{algorithm}[t]
  \caption{\Master: 
  A procedure to decide linear-exponential systems over~$\N$.}
  \label{algo:master}
    \begin{algorithmic}[1]
      \Require
        $\varphi(x_1,\ldots,x_n):{}$ linear-exponential system (without divisibility constraints).
      \Ensure
        True ($\top$) if $\varphi$ has a solution over $\N$, and otherwise false ($\bot$).
      \medskip
      \State \textbf{let} $x_0$ be a fresh variable 
      \label{algo:master:addxzero}
      \Comment{placeholder for $0$}
      \State\label{algo:master:guess-order}%
      \myguess $\theta \gets{}$ordering of the form $t_1 \geq t_2 \geq \dots \geq t_{n} \geq 2^{x_0} = 1$, where $(t_1,\dots,t_{n})$ is a 
      \Statex\vspace{-2pt} \hphantom{\myguess $\theta \gets{}$}permutation of the terms $2^{x_1},\dots,2^{x_n}$
      \While{$\theta$ is not the ordering $(2^{x_0}=1)$}\label{algo:master:outer-loop}
        \State\label{algo:master:extrX}%
          $2^x\gets$ leading exponential term of $\theta$
        \Comment{in the $i$-th iteration, $2^x$ is $t_i$}
        \State\label{algo:master:extrY}%
          $2^y \gets$ successor of $2^x$ in $\theta$
          \Comment{and $2^y$ is $t_{i+1}$}
        \State\label{algo:master:large-mod-sub}%
          $\varphi\gets\varphi\sub{w}{(w \bmod 2^x):
          \text{$w$ is a variable}}$
        \State $\vec z\gets$ all variables $z$ in $\varphi$ such that $z$ is $x$ or $z$ does not appear in $\theta$
          \label{algo:master:def-z}
          \ForEach{$z$ in $\vec z$}
            \Comment{form a quotient system induced by~$\theta$}
            \label{algo:master:inner-loop}
            \State \textbf{let} $x'$ and $z'$ be two fresh variables
            \State\label{algo:master:imposezlessy}%
            $\varphi\gets\varphi\land(0\leq z'< 2^y)$
            \State\label{algo:master:elimmod}%
              $\varphi\gets\varphi\sub{z'}{(z \bmod 2^y)}$
            \State\label{algo:master:simpmod}%
              $\varphi\gets\varphi\sub{(z'\bmod 2^w)}{(z \bmod 2^w):
              \text{$w$ is such that $\theta$ implies $2^w \leq 2^y$}}$
            \State $\varphi\gets\varphi\sub{(x'\cdot 2^y + z')}{z}$
            \Comment{replaces only the linear occurrences of $z$}
            \label{algo:master:linear-substitution}
            \If{$z$ is $x$} 
              $(x_0',z_0') \gets (x',z')$
             \Comment{for delayed substitution, see next line} 
            \EndIf\label{algo:master:end-inner-loop}
          \EndFor
        \State $\varphi\gets$\Call{ElimMaxVar}{$\theta,\varphi,\sub{x_0'\cdot 2^y + z_0'}{x}$}\label{algo:master:call-elimmaxvar}
        \State\label{algo:master:end-loop}%
          remove $2^x$ from $\theta$
      \EndWhile
      \State\label{algo:master:return}%
        \textbf{return }$\varphi(\vec 0)$
        \Comment{evaluates to $\top$ or $\bot$}
    \end{algorithmic}
  \end{algorithm}%
}%

\subsection{Algorithm~\ref{algo:master}: the main loop} 
Let
$\phi(x_1,\dots,x_n)$ be an input linear-exponential system
(with no divisibility constraints). As explained in the
summary above,~\Master starts by guessing an
ordering $\theta(x_0,\dots,x_n)$ of the form~${t_1 \geq t_2 \geq \dots
\geq t_{n} \geq 2^{x_0} = 1}$, where $(t_1,\dots,t_{n})$ is
a permutation of the terms $2^{x_1},\dots,2^{x_n}$, and
$x_0$ is a fresh variable used as a placeholder for $0$.
Note that if $\phi$ is satisfiable (over $\N$), then
$\theta$ can be guessed so that $\phi \land \theta$ is
satisfiable; and conversely no such $\theta$ exists if
$\phi$ is unsatisfiable.
For the sake of convenience, we assume in this section
that the ordering $\theta(x_0,\dots,x_n)$ guessed by the procedure
is
${2^{x_n} \geq 2^{x_{n-1}} \geq
\dots \geq 2^{x_1} \geq 2^{x_0} = 1}$. 

The \textbf{while} loop starting in
line~\ref{algo:master:outer-loop} manipulates $\phi$ and
$\theta$, non-deterministically obtaining at the end of the $i$th iteration a
system $\phi_i(\vec x, \vec z)$ and an ordering
$\theta_i(\vec x)$, where $\vec x = (x_0,\dots,x_{n-i})$ and
$\vec z$ is a vector of $i$~fresh variables. 
The non-deterministic guesses performed by~\Master are such that the following properties~\eqref{inv:it1}--\eqref{inv:it3} are loop invariants across all branches, whereas~\eqref{inv:it4} 
is an invariant for at least one branch (below, $i \in [0,n]$ and $(\phi_0,\theta_0) \coloneqq (\phi,\theta)$):%
\vspace{4pt}
\begin{description}
  \item[\labeltext{I1}{inv:it1}.] All variables that occur
  exponentiated in $\phi_i$ are among $x_0,\dots,x_{n-i}$.
  \item[\labeltext{I2}{inv:it2}.] $\theta_i$ is the ordering $
  {2^{x_{n-i}} \geq 2^{x_{n-i-1}} \geq \dots \geq 2^{x_1}
  \geq 2^{x_0} = 1}$.
  \item[\labeltext{I3}{inv:it3}.] All variables $z$ in $\vec z$ 
  are such that $z < 2^{x_{n-i}}$ is an inequality in $\phi_i$.
  \item[\labeltext{I4}{inv:it4}.] $\phi_i \land \theta_i$ is
  equisatisfiable with $\phi \land \theta$ over $\N$.
\end{description}
More precisely, writing $\bigvee_{\beta} \psi_\beta$ for the disjunction of all the formulae $\phi_i \land \theta_i$ obtained across all non-deterministic branches, we have that $\bigvee_{\beta} \psi_\beta$ and $\phi \land \theta$ are equisatisfiable. Therefore, whenever $\phi \land \theta$
is satisfiable,~\eqref{inv:it4} holds for at least one branch. If $\phi \land \theta$ is instead unsatisfiable, then~\eqref{inv:it4} holds instead for all branches.

The invariant above is clearly true for $\phi_0$ and $\theta_0$,
with $\vec z$ being the empty set of variables.
Item~\eqref{inv:it2} implies that, after $n$~iterations, $\theta_n$ is $2^{x_0} = 1$,
which causes the \textbf{while} loop
to exit. Given $\theta_n$,
properties~\eqref{inv:it1} and~\eqref{inv:it3} force
the values of $x_0$ and of all variables in $\vec z$ to be
zero, thus making $\phi \land \theta$
equisatisfiable with $\phi_n(\vec 0)$ in at least one branch of the algorithm, by~\eqref{inv:it4}. 
In summary, this will enable us to conclude that the
procedure is correct.

Let us now look at the body of the \textbf{while} loop. 
Its objective is simple: manipulate the current system, say $\phi_i$, so that it becomes a quotient system induced by~$\theta_i$,
and then call~\Cref{algo:elimmaxvar} (\ElimMaxVar).
For these systems, note that~$2^x$ and $2^y$ in
lines~\ref{algo:master:extrX}--\ref{algo:master:extrY} 
correspond to $2^{x_{n-i}}$ and~$2^{x_{n-i-1}}$, respectively. 
Behind the notion of quotient system there are two goals.
One of them is
to make sure that $2^x$ and $2^y$ are not involved in modulo operations.
(We will discuss the second goal in~\Cref{subsec:elimmaxvar}.)
The \textbf{while} loop achieves this goal as follows:
\begin{itemize}
\item
Since $2^x$ is greater than every variable in $\phi_i$, every $(w \bmod 2^x)$ can be replaced with $w$. 
\item
For $2^y$ instead, we ``divide'' every variable $z$ that might be larger than it.
Observe that~$z$ is either $x$ or from the vector~$\vec z$  in~\eqref{inv:it3} of the invariant.
The procedure replaces every linear occurrence of~$z$ with $x' \cdot 2^y + z'$, where $x'$ and $z'$ are fresh variables and $z'$ is a residue modulo $2^y$, that is, $0 \leq z' < 2^y$.
\end{itemize}
The above-mentioned replacement simplifies all modulo operators where $z$ appears: $(z \bmod 2^y)$ becomes $z'$, and every $(z \bmod 2^w)$ such that $\theta_i$ 
entails $2^w \leq 2^y$ becomes $(z' \bmod 2^w)$.
We obtain in this way a quotient system induced by $\theta_i$,
and pass it to \ElimMaxVar.

Whilst the goal we just discussed is successfully achieved,
we have not in fact eliminated the variable~$x$ completely.
Recall that, according to our definition of substitution,
occurrences of~$2^x$ in the system~$\phi$ are unaffected by
the application of $\sub{x' \cdot 2^y + z'}{x}$
in line~\ref{algo:master:linear-substitution} 
of \Master.
Because of this, the procedure keeps this substitution as a
\emph{delayed substitution} for future use, to be applied
(by~\ElimMaxVar) when $x$ will finally occur only
linearly.


\begin{algorithm}
  \tikzmark{left-margin}
  \caption{\ElimMaxVar: Variable elimination for quotient systems.}
  \label{algo:elimmaxvar}
  \begin{algorithmic}[1]
    \Require
      {\setlength{\tabcolsep}{0pt}
      \begin{tabular}[t]{rl}
        $\theta(\vec x):{}$ &\ ordering of exponentiated variables;\\
        $\varphi(\vec x,\vec x', \vec z'):{}$ &\ quotient system induced by~$\theta$, with $\vec x$ exponentiated,\\[-2pt] 
        &\ $\vec x'$ quotient, and $\vec z'$ remainder variables;\\
        $[x' \cdot 2^y + z'/x]:{}$ &\ delayed substitution for $\phi$.  
      \end{tabular}}
    \NDBranchOutput $\psi_\beta(\vec x \setminus x,\vec z'):{}$ linear-exponential system such that 
      for every $z$ in $\vec z'$, 
      $z$ does not occur in exponentials and $0 \leq z < 2^y$ occurs in $\psi_\beta$. 
    \GlobalSpec
      $(\exists x\,\theta)\land\bigvee_{\beta}\psi_\beta$ is equivalent to 
      $\exists x\exists\vec x' (\theta \land\varphi\land x=x'\cdot 2^y+z')$ 
      over~$\N$.%
    \medskip 
    \State\textbf{let} $u$ be a fresh variable
    \label{line:elimmaxvar-u}
    \Comment{$u$ is an alias for $2^{x-y}$}
    \State $\gamma\gets\top;\ \psi\gets\top$
    \label{line:elimmaxvar-initialize}
    \tikzmark{quotient-1-elim-begins}
    \State $\Delta\gets\varnothing$
    \label{line:elimmaxvar-delta}
    \Comment{map from linear-exponential terms to $\Z$}\;
    \ForEach{ $(\tau \sim 0)$ in $\phi$, where ${\sim} \in \bigl\{=,<,\leq,\equiv_d : d \geq 1\bigr\}$}
    \label{line:elimmaxvar-foreach}
      \State \textbf{let} $\tau$ be $(a \cdot 2^x + f(\vec x') \cdot 2^y + \rho)$, where $\rho$ is the least significant part of $\tau$
      \label{line:elimmaxvar-tau}
      \If{$a=0$ and $f(\vec x')$ is an integer}
        $\psi\gets\psi\land(\tau\sim0)$ 
        \label{line:elimmaxvar-psi-trivial}
      \ElsIf{the symbol $\sim$ belongs to $\{{=},{<},{\leq}\}$}
        \label{line:elimmaxvar-second-case}
        \If{$\Delta(\rho)$ is undefined}
        \label{line:elimmaxvar-delta-undefined}
          \State \myguess $r\gets$ integer in $[-\onenorm{\rho},\onenorm{\rho}]$
          \label{line:elimmaxvar-guess-rho}
          \State $\psi\gets\psi\land((r-1)\cdot2^y<\rho)\land(\rho\leq r\cdot2^y)$
          \label{line:elimmaxvar-psi-inequality}
          \State update $\Delta$ : add the key--value pair $(\rho, r)$
          \label{line:elimmaxvar-delta-update}
        \EndIf
        \State $r \gets \Delta(\rho)$
        \label{line:elimmaxvar-delta-from}
        \If{the symbol $\sim$ is $<$}
        \label{line:elimmaxvar-if-less}
           \State \myguess ${\sim^{\prime}} \gets$ sign in $\{=,<\}$;
           \label{line:elimmaxvar-guess-sign} 
           \quad $\psi\gets \psi\land (\rho \sim^{\prime} r\cdot 2^y)$; 
           \label{line:elimmaxvar-strict-add-to-psi} 
           \quad ${\sim} \gets\,\leq$ 
           \label{line:elimmaxvar-strict-ineq}
           \If{the symbol $\sim^{\prime}$ is $=$} $r\gets r+1$
           \EndIf
           \label{line:elimmaxvar-update-r}
        \EndIf
        \State $\gamma\gets\gamma\land(a \cdot u + f(\vec x')+r\sim0)$
        \label{line:elimmaxvar-gamma-inequality}
        \If{the symbol $\sim$ is $=$}
          $\psi\gets\psi\land(r\cdot2^y=\rho)$
          \label{line:elimmaxvar-psi-equality}
        \EndIf
      \Else
        \label{line:elimmaxvar-third-case}
        \Comment{$\sim$ is $\equiv_d$ for some $d \in \N$}
        \State \myguess $r\gets$ integer in $[1,\emph{mod}(\varphi)]$
        \label{line:elimmaxvar-guess-mod}
        \State $\gamma\gets\gamma\land(a \cdot u + f(\vec x') -r\sim0)$
        \label{line:elimmaxvar-gamma-divisibility}
        \State $\psi\gets\psi\land(r\cdot2^y+\rho\sim0)$
        \label{line:elimmaxvar-psi-divisibility}
      \EndIf	
    \EndFor
    \State $\gamma\gets$ \GaussianQE$(\vec x' \setminus x',\gamma \land \vec x' \geq \vec 0)$
    \tikzmark{quotient-1-elim-ends}
    \label{line:elimmaxvar-gauss-1}
    \State $\gamma \gets \gamma[2^u\,/\,u]$
    \label{line:elimmaxvar-2-u}
    \Comment{$u$ now is an alias for $x-y$}
    \State $(\chi,\gamma)\gets$ \textsc{SolvePrimitive}$(u,x',\gamma)$
    \label{line:elimmaxvar-solve-primitive}
    \State $\chi\gets\chi[x-y\,/\,u][x'\cdot2^y+z'\,/\,x]$
    \label{line:elimmaxvar-delayed-subs}
    \Comment{apply delayed substitution: $x$ is eliminated}
    \tikzmark{x1-elim-begins}
    \If{$\chi$ is $(-x'\cdot2^y-z'+y+c = 0)$ for some $c\in\N$}
    \label{line:elimmaxvar-x1-start}
      \State \myguess $b\gets{}$integer in $[0,c]$
      \label{line:elimmaxvar-guess-c}
      \State $\gamma\gets\gamma\land(x' = b)$
      \label{line:elimmaxvar-gamma-equality-x1}
      \State $\psi\gets\psi\land(b\cdot2^y=-z'+y+c)$
      \label{line:elimmaxvar-psi-equality-x1}
    \Else
      \State \textbf{let} $\chi$ be $(-x'\cdot2^y-z'+y+c \leq 0) \land (d \divides x' \cdot 2^y + z' - y - r)$, with $d,r\in\N$, $c \geq 3$
      \label{line:elimmaxvar-x1-let}
      \State \myguess $(b,g) \gets{}$pair of integers in $[0,c] \times [1,d]$
      \label{line:elimmaxvar-guess-times}
      \State $\gamma\gets\gamma\land(x'\geq b)\land(d\mid x'-g)$
      \label{line:elimmaxvar-gamma-inequality-x1}
      \State $\psi\gets\psi\land((b-1)\cdot2^y<-z'+y+c)\land(-z'+y+c\leq b\cdot2^y)\land(d\mid g\cdot2^y+z'-y-r)$
      \label{line:elimmaxvar-psi-inequality-x1}
    \EndIf
    \Assert{\Call{\GaussianQE}{$x',\gamma$} is equivalent to $\top$}
    \Comment{upon failure, \Cref{algo:master} returns~$\bot$}
    \tikzmark{x1-elim-ends}
    \label{line:elimmaxvar-gauss-2}
    \State \textbf{return} $\psi$%
    \vspace{-11pt}
  \end{algorithmic}
  \AddNote{quotient-1-elim-begins}{quotient-1-elim-ends}{left-margin}{Elimination of $\vec x'\setminus x'$}%
  \AddNote{x1-elim-begins}{x1-elim-ends}{left-margin}{Elimination of $x'$}%
\end{algorithm}

\vspace{-3pt}
\subsection{Algorithm~\ref{algo:elimmaxvar}: elimination of leading variable and quotient variables}
\label{subsec:elimmaxvar}

Let $\phi(\vec x, \vec x', \vec z')$ be a quotient system induced by an ordering~$\theta(\vec x)$, with $\vec x$ exponentiated, $\vec x'$ quotient and $\vec z'$ remainder variables, and consider a delayed substitution $\sub{x' \cdot 2^y + z'}{x}$. 
\ElimMaxVar removes $\vec x'$ and~$x$, 
obtaining a linear-exponential system~$\psi$
that adheres to the loop invariant of~\Master. 
This is done by following the three steps described in 
the summary of the procedure, which we now expand.

\negspace
\subparagraph*{Step~\labeltext{(i)}{elim:i1-copy}: lines~\ref{line:elimmaxvar-delta}--\ref{line:elimmaxvar-gauss-1}.}
This step aims at calling~\Cref{algo:gaussianqe} (\GaussianQE) to eliminate all variables in $\vec x' \setminus x'$. 
There is, however, an obstacle: 
these variables are multiplied by~$2^y$. 
Here is where the second goal behind the notion of quotient system comes into play: 
making sure that least significant parts of quotient terms can be bounded in terms of $2^y$.
To see what we mean by this and why it is helpful, consider below an inequality~$\tau \leq 0$ 
from~$\phi$, where $\tau = a \cdot 2^{x} + f(\vec x') \cdot 2^{y} + \rho(\vec x \setminus x, \vec z')$ and $\rho$ is the least significant part of~$\tau$.

Since $\phi$ is a quotient system induced by $\theta$,
all variables and exponential terms $2^w$ appearing in $\rho$ are bounded by $2^y$, and thus every solution of $\phi \land \theta$ must also satisfy~$\abs{\rho} \leq \onenorm{\rho} \cdot 2^y$. 
More precisely, the value of~$\rho$ must lie in the interval~$[(r-1)\cdot2^y+1,r \cdot 2^y]$ for some~$r \in [-\onenorm{\rho},\onenorm{\rho}]$. 
The procedure guesses one such value~$r$ 
(line~\ref{line:elimmaxvar-guess-rho}). 
The inequality $\tau \leq 0$ can be rewritten as
\begin{equation}
  \label{equation:split-inequality}
  \big(a \cdot 2^x + f(\vec x') \cdot 2^y + r \cdot 2^y \leq 0\big)
   \ \land\ \big((r-1) \cdot 2^y < \rho \leq r \cdot 2^y\big).
\end{equation}
Fundamentally, $\tau \leq 0$ has been split into a ``left part'' and a ``right part'', shown with big brackets around. The ``right part'' $(r-1) \cdot 2^y < \rho \leq r \cdot 2^y$ is made of two linear-exponential inequalities 
featuring none of the variables we want to eliminate ($\vec x'$ and $x$). 
Following the same principle, the procedure produces similar splits for all strict inequalities, equalities, and divisibility constraints of $\phi$.
In the pseudocode, the ``left parts'' of the system are stored in the formula~$\gamma$, and the ``right parts''
are stored in the formula $\psi$.

Let us focus on a ``left part'' $a \cdot 2^x + f(\vec x') \cdot 2^y + r \cdot 2^y \leq 0$ in $\gamma$.
Since $\theta$ implies $2^x \geq 2^y$, we 
can factor out $2^y$ from this constraint, obtaining the inequality $a \cdot 2^{x-y} + f(\vec x') + r \leq 0$. 
There we have it: the variables $\vec x' \setminus x'$ occur now linearly in~$\gamma$ and can be eliminated thanks to~\GaussianQE.
For performing this elimination, the presence of $2^{x-y}$ is unproblematic. In fact, the procedure uses a placeholder variable~$u$ for $2^{x-y}$ (line~\ref{line:elimmaxvar-u}), so that
$\gamma$ is in fact a linear system with, e.g., inequalities $a \cdot u + f(\vec x') + r \leq 0$. 
Observe that inequalities $\vec x' \geq \vec 0$ are added to $\gamma$ in line~\ref{line:elimmaxvar-gauss-1}, since~\GaussianQE works over $\Z$ instead of $\N$.
This concludes Step~\ref{elim:i1-copy}.

Before moving on to Step~\ref{elim:i2-copy}, we justify the use of the map~$\Delta$ from line~\ref{line:elimmaxvar-delta}.
If the procedure were to apply~\Cref{equation:split-inequality} and 
replace every 
inequality $\tau \leq 0$ with three inequalities, 
then multiple calls to~\ElimMaxVar would produce a system with exponentially many constraints.
A solution to this problem is to guess $r \in [-\onenorm{\rho},\onenorm{\rho}]$ only once, 
and use it in all the ``left parts'' stemming from inequalities in~$\phi$ having $\rho$ as their least significant part. 
The ``right part'' $(r-1) \cdot 2^y < \rho \leq r \cdot 2^y$ is added to $\psi$ only once.
The map~$\Delta$ implements this memoisation, avoiding the aforementioned 
exponential blow-up. 
Indeed, the number of least significant parts 
grows very slowly throughout~\Master, 
as we will see in~\Cref{sec:complexity}.

\subparagraph*{Step~\labeltext{(ii)}{elim:i2-copy}: lines~\ref{line:elimmaxvar-2-u}--\ref{line:elimmaxvar-delayed-subs}.}
The goal of this step is to eliminate all occurrences of the term $2^{x-y}$. For convenience, the procedure first reassigns $u$ to now be a placeholder for $x-y$ (line~\ref{line:elimmaxvar-2-u}). 
Because of this reassignment, the system $\gamma$ returned by~\GaussianQE 
at the end of Step~\ref{elim:i1-copy} is a $(u,x')$-primitive linear-exponential system.

The procedure calls~\Cref{algo:linearize} (\SolvePrimitive), 
which constructs from $\gamma$ a pair of systems $(\chi_{\beta}(u),\gamma_\beta(x'))$, which is assigned to $(\chi,\gamma)$. 
Both are linear systems, and thus 
all occurrences of $2^{x-y}$ (rather, $2^u$) have been removed.
At last, all promised substitutions can be realised (line~\ref{line:elimmaxvar-delayed-subs}):
$u$ is replaced with $x-y$, and the delayed substitution replaces $x$ with $x' \cdot 2^y + z'$.
This eliminates~$x$. The only variable that is yet to be removed is~$x'$ (Step~\ref{elim:i3-copy}).

It is useful to recall at this stage that~\SolvePrimitive 
is only correct under the assumption that $u \geq x' \geq 0$.
This assumption is guaranteed by the definition of $\theta$,
the delayed substitution, and the fact that $u$ is a placeholder for $x-y$ (and we are working over~$\N$).
Indeed, if $x' = 0$, then the inequality $2^x \geq 2^y$ in $\theta$ ensures $u = x-y \geq 0 = x'$. 
If $x' \geq 1$, 
\begin{align*}  
  u=x-y&=x'\cdot2^y+z'-y
  &\text{delayed substitution}
  \\
  &\geq x'\cdot(y+1)+z'-y
  &2^y\geq y+1 \text{, for every }y\in\N 
  \\
  &= y\cdot(x'-1)+x'+z'\geq x'.
  &\text{since } x' \geq 1.
  \\[-20pt]
\end{align*}

\subparagraph*{Step~\labeltext{(iii)}{elim:i3-copy}: lines~\ref{line:elimmaxvar-x1-start}--\ref{line:elimmaxvar-gauss-2}.}
This step deals with eliminating the variable~$x'$ 
from the formula $\gamma(x') \land \chi(x',z',y) \land \psi(\vec x \setminus x,\vec z') $, 
where $\psi$ contains the ``right parts'' of $\phi$ computed during Step~\ref{elim:i1-copy}. 
The strategy to eliminate~$x'$ follows closely what was done to eliminate the other quotient variables from $\vec x'$ during Step~\ref{elim:i1-copy}: 
the algorithm first splits the formula $\chi(x',z',y)$ into a ``left part'', which is added to $\gamma$ and features the variable $x'$, 
and a ``right part'', which is added to $\psi$ and features the variables $z'$ and $y$. 
It then eliminates $x'$ by calling~\GaussianQE on $\gamma$. 
To perform the split into ``left part'' and ``right part'', 
observe that $\chi$ is a system 
of the form either $-x' \cdot 2^y-z'+y+c = 0$ or $(-x'\cdot2^y-z'+y+c \leq 0) \land (d \divides x' \cdot 2^y + z' - y - r)$ (see the spec of~\SolvePrimitive). 
By arguments similar to the ones used for~$\rho$ in Step~\ref{elim:i1-copy}, $-z'+y+c$ can be bounded in terms of $2^y$.
(Notice, e.g., the similarities between the inequalities in line~\ref{line:elimmaxvar-psi-inequality-x1} 
and the ones in line~\ref{line:elimmaxvar-psi-inequality}.)
After the elimination of $x'$, if~\GaussianQE does not yield an unsatisfiable formula, 
\ElimMaxVar returns the system $\psi$ to~\Master.

Before moving on to the description of~\SolvePrimitive, let us clarify the semantics of the~\textbf{assert} statement occurring in~line~\ref{line:elimmaxvar-gauss-2}. 
It is a standard semantics from programming languages.
If an assertion~$b$ evaluates to true at runtime, \textbf{assert}($b$) does nothing. If $b$ evaluates to false instead, the execution aborts and the main procedure (\Master) returns $\bot$. This semantics allows for assertions to query \np problems, as done in~line~\ref{line:elimmaxvar-gauss-2} (and in~line~\ref{line:linearize-assert} of~\SolvePrimitive), without undermining the membership in \np of~\Master.

\vspace{2em}

{
\begin{algorithm}[t]
  \caption{\SolvePrimitive: A procedure to decompose and linearise primitive systems.}\label{algo:linearize}
  \begin{algorithmic}[1]
    \Require
      $u,v:{}$ two varaibles; \ 
      $\varphi:{}$ $(u,v)$-primitive linear-exponential system.
    \NDBranchOutput a pair of linear systems $(\chi_\beta(u),\gamma_\beta(v))$ such that
      $\chi_\beta(u)$ is either of the form $(u=a)$ or of the form $(u \geq b) \land (d \mid u-r)$, where $a,d,r\in\N$ and $b \geq 3$.
    \GlobalSpec
      $(u \geq v \geq 0)$ entails that
      $\bigvee_{\beta}(\chi_\beta \land \gamma_\beta)$ is equivalent to~$\phi$.
    \medskip
    \State \textbf{let} $\varphi$ be $(\chi\land\psi)$, where $\chi$ is the conjunction of all (in)equalities from $\phi$ containing $2^{u}$
    \label{line:linearize-decompose}
    \State $(d,n)\gets$ pair of non-negative integers such that $\fmod(\phi)=d\cdot2^n$ and $d$ is odd
    \label{line:linearize-mod-factor}
    \State $C \gets \max\big\{n, 3 + 2 \cdot\bigl\lceil\log(\frac{|b|+|c|+1}{|a|})\bigr\rceil : \text{$(a\cdot 2^{u}+b\cdot v+c\sim0)$ in $\chi$, where ${\sim} \in \{=,<,\leq\}$} \big\}$
    \label{line:linearize-max-constant}
    \State \myguess $c\gets{}$element of $[0,C-1] \cup \{\star\}$
    \label{line:linearize-guess-c}
    \Comment{$\star$ signals $u \geq C$}
    \label{line:linearize-guess-const}
    \If{$c$ is not $\star$}
      \State $\chi\gets(u=c)$
      \label{line:linearize-gamma-equality}
      \State $\gamma\gets\varphi[2^c\,/\,2^{u}]$
      \label{line:linearize-chi-equality}
    \Else
      \Comment{assuming~$u \geq C$, (in)equalities in $\chi$ simplify to $\top$ or $\bot$}
      \label{line:linearize-else}
      \Assert{$\chi$ has no equality, and in all its inequalities $2^u$ has a negative coefficient}
      \label{line:assert-not-bottom}
      \State \myguess $r\gets{}$integer in $[0,d-1]$
      \label{line:linearize-guess-div}
      \Comment{remainder of $2^{u-n}$ modulo $d$ when $u\geq C \geq n$}
      \Assert{$d \divides 2^u - 2^n \cdot r$ is satisfiable}
      \label{line:linearize-assert}
      \State $r'\gets$ discrete logarithm of $2^n \cdot r$ base $2$, modulo $d$
      \label{line:linearize-discrete-log}
      \State $d'\gets$ multiplicative order of $2$ modulo $d$
      \label{line:linearize-mult-ord}
      \State $\chi\gets(u \geq C)\land(d'\mid u-r')$
      \label{line:linearize-chi-case-2}
      \State $\gamma \gets \psi[2^n \cdot r /2^u]$
      \Comment{$2^n \cdot r$ is a remainder of $2^u$ modulo $\fmod(\psi)=d\cdot2^n$}
      \label{line:linearize-chi-div}
    \EndIf
    \State \textbf{return }$(\chi,\gamma)$
  \end{algorithmic}
\end{algorithm}%
}%
\vspace{-3pt}
\subsection{Algorithm~\ref{algo:linearize}: from primitive systems to linear systems}
Consider an input $(u,v)$-primitive linear-exponential system~$\phi$, 
and further assume we are searching for solutions over~$\N$ where $u \geq v$. The goal of~\SolvePrimitive is to decompose~$\phi$ (in the sense of monadic decomposition~\cite{HagueLRW20,Libkin03}) into two \emph{linear} systems: a system $\chi$ only featuring the variable $u$, and a system $\gamma$ only featuring~$v$. 

To decompose~$\phi$, 
the key parameter to understand is the threshold~$C$ for the variable $u$ (line~\ref{line:linearize-max-constant}). 
This positive integer depends on two quantities, one for ``linearising'' the divisibility constraints, and one for ``linearising'' the equalities and inequalities of $\phi$. 
Below we first discuss the latter quantity. 
Throughout the discussion, we assume $u \geq C$, as otherwise the 
procedure simply replaces~$u$ with 
a value in $[0,C-1]$ (see lines~\ref{line:linearize-gamma-equality} and~\ref{line:linearize-chi-equality}).

Consider an inequality $a \cdot 2^u + b \cdot v + c \leq 0$. 
Regardless of the values of $u$ and $v$, as long as $\abs{a \cdot 2^u} > \abs{b \cdot v + c}$ holds, the truth of this inequality will solely depend on the sign of the coefficient~$a$.
Since we are assuming $u \geq v$ and $u \geq C \geq 1$, $\abs{a \cdot 2^u} > \abs{b \cdot v + c}$ is implied by~$\abs{a} \cdot 2^u > (\abs{b} + \abs{c}) \cdot u$. In turn, this inequality is implied by $u \geq C$,
because both sides of the inequalities are monotone functions, $\abs{a} \cdot 2^u$ grows faster than $(\abs{b} + \abs{c}) \cdot u$, and,
given $C' \coloneqq 3+2 \cdot \ceil{\log(\frac{\abs{b}+\abs{c}+1}{\abs{a}})}$ (which is at most $C$),
we have
\begin{align*}
  |a|\cdot 2^{C'} 
  &\geq |a|\cdot 2^3\cdot\left(\frac{|b|+|c|+1}{|a|}\right)^2 >\,\big(|b|+|c|\big)\cdot2^{\bigl\lceil\log(\frac{|b|+|c|+1}{|a|})\bigr\rceil+2}
  >\,\big(|b|+|c|\big)\cdot C'\,,
\end{align*}
where to prove the last inequalities 
one uses the fact that $2^{x+1} > 2 \cdot x + 1$ for every $x \geq 0$.
Hence, when~$u \geq C$, every inequality in $\phi$
simplifies to either $\top$ or $\bot$, and this is also true
for strict inequalities. The Boolean value~$\top$ arises
when $a$ is negative. The Boolean~$\bot$ arises when $a$ is
positive, or when instead of an inequality we consider an
equality.  

It remains
to handle the divisibility constraints, again under the assumption $u \geq C$.
This is where the second part of the definition of $C$
plays a role.
Because $u \geq C \geq n$ (see the definition of~$(d,n)$ in~line~\ref{line:linearize-mod-factor}), we can guess $r \in [0,d-1]$ such that 
$\fmod(\phi) \divides 2^u - 2^n \cdot r$ (line~\ref{line:linearize-guess-div}). 
This constraint is equivalent to $d \divides 2^{u-n} - r$ and, since $2^n$ and $d$ are coprime, it is also equivalent to $d \divides 2^u - 2^n \cdot r$. It might be an unsatisfiable constraint: the procedure checks for this eventuality in line~\ref{line:linearize-assert}, 
by solving a discrete logarithm problem (which can be done in~\np, see~\cite{Joux14}). 
Suppose a solution is found, say~$r'$ (as in line~\ref{line:linearize-discrete-log}).
We can then represent the set of solutions of $d \divides 2^u - 2^n \cdot r$ as an arithmetic progression: it suffices to compute the multiplicative order of $2$ modulo $d$, i.e., the smallest positive integer~$d'$ such that $d \divides 2^{d'} - 1$. This is again a discrete logarithm problem, but differently from the previous case $d'$ always exists since $d$ and $2$ are coprime.
The set of solutions of~$d \divides 2^u - 2^n \cdot r$ 
is given by~$\{r' + \lambda \cdot d' : \lambda \in \Z\}$,
that is,~$\fmod(\phi) \divides 2^u - 2^n \cdot r$ 
is equivalent to $d' \divides u - r'$. 
The procedure then returns~$\chi(u) \coloneqq (u \geq C \land d' \divides u - r')$ and $\gamma(v) \coloneqq \psi\sub{2^n \cdot r}{2^u}$ (see lines~\ref{line:linearize-chi-case-2} and~\ref{line:linearize-chi-div}), where $\psi$ (defined in line~\ref{line:linearize-decompose}) is the system obtained from $\phi$ by removing all equalities and inequalities featuring $2^u$.

\vspace{3pt}

Elaborating the arguments sketched in this section,
we can prove that \Cref{algo:master,algo:elimmaxvar,algo:linearize} comply with their specifications. 

\begin{restatable}{proposition}{PropMasterCorrect}
  \label{prop:master-correct}
  \Cref{algo:master} (\Master) is a correct procedure for deciding the satisfiability of linear-exponential systems over~$\N$.
\end{restatable}

\section{Complexity analysis}
\label{sec:complexity}
We analyse the procedure introduced in~\Cref{sec:procedure,sec:procedure-details} and show that it runs in non-deterministic polynomial time. This establishes~\Cref{thm:buchi-semenov-in-np-z} restricted to~$\N$.

\begin{restatable}{proposition}{PropMasterInNp}
  \label{prop:master-in-np}
  \Cref{algo:master} (\Master) runs in non-deterministic polynomial time.
\end{restatable}

To simplify the analysis required to establish~\Cref{prop:master-in-np}, we assume that~\mbox{\Cref{algo:master,algo:elimmaxvar,algo:linearize}} store the divisibility constraints $d \divides \tau$ of a system~$\phi$
in a way such that the coefficients and the constant of~$\tau$ are always reduced modulo~$\fmod(\phi)$. 
For example, if $\fmod(\phi) = 5$,
the divisibility $5 \divides (7 \cdot x + 6 \cdot 2^x - 1)$ is stored as $5 \divides (2 \cdot x + 2^x + 4)$.
Any divisibility 
can be updated in polynomial time to satisfy this requirement, so 
there is no loss of generality.
Observe that~\Cref{algo:gaussianqe} (\GaussianQE) is an exception to this rule, as the divisibility constraints it introduces in line~\ref{gauss:restore} must respect some structural properties throughout its execution. 
Thus, line~\ref{line:elimmaxvar-2-u} of~\Cref{algo:elimmaxvar} (\ElimMaxVar) implicitly reduces the output of \GaussianQE modulo $m = \fmod(\phi)$ as appropriate.
Since \GaussianQE runs in non-deterministic polynomial time, the reduction takes polynomial time too.

As is often the case for arithmetic theories, the complexity analysis of our algorithms requires tracking several parameters of linear-exponential systems. 
Below, we assume an ordering ${\theta(\vec x) = (2^{x_n} \geq \dots \geq 2^{x_0} = 1)}$ 
and let $\phi$ be either a linear-exponential system or a quotient system induced by~$\theta$. Here are the parameters we track:
\vspace{3pt}
\begin{itemize}
  \item The least common multiple of all divisors~$\fmod(\phi)$, defined as in~\Cref{sec:preliminaries}.
  \item The number of equalities, inequalities and divisibility constraints in $\phi$, denoted by~$\card \phi$.
  (Similarly, given a set~$T$, we write $\card T$ for its cardinality.)
  \item The $1$-norm~$\onenorm{\phi} \coloneqq \max\{\onenorm{\tau} : \text{$\tau$ is a term appearing in an (in)equality of $\phi$}\}$. 
  For linear-exponential terms,~$\onenorm{\tau}$ is defined in~\Cref{sec:preliminaries}. For quotient terms~$\tau$ induced by~$\theta$, the
  $1$-norm~$\onenorm{\tau}$ is defined as the sum of the absolute values of all the coefficients and constants appearing in~$\tau$.
  The definition of~$\onenorm{\phi}$ excludes integers appearing in divisibility constraints since, as explained above, those are already bounded by~$\fmod(\phi)$.
  \item The \emph{linear norm} $\linnorm{\phi} \coloneqq \max\{\linnorm{\tau} :  \text{$\tau$ is a term appearing in an (in)equality of $\phi$} \}$. 
  For~a~linear-exponential term~$\tau = {\sum\nolimits_{i=1}^n \big(a_i \cdot x_i + b_i \cdot 2^{x_i} + \sum\nolimits_{j=1}^n c_{i,j} \cdot (x_i \bmod 2^{x_j})\big) + d}$, 
  we define $\linnorm{\tau} \coloneqq \max\{ \abs{a_i}, \abs{c_{i,j}} : i,j \in [1,n] \}$, 
  that is, the maximum of all coefficients of $x_i$ and $(x_i \bmod 2^{x_j})$, in absolute value. 
  For a quotient term induced by~$\theta$,
  of the form
  $\tau = a \cdot 2^{x_n} + (c_1 \cdot x_1' + \dots + c_m \cdot x_m' + d) \cdot 2^{x_{n-1}} + b \cdot x_{n-1} + \rho(x_{0},\dots,x_{n-2}, \vec z')$,
  we define $\linnorm{\tau} \coloneqq \max\big(\abs{b},\linnorm{\rho},\max\{\abs{c_i} : i \in [1,m]\}\big)$, thus also taking into account
  the coefficients of the quotient variables~$x_1',\dots,x_m'$.
  \item The set of the least significant terms $\lst(\phi,\theta)$ defined as $\big\{ \pm \rho : {} \rho$ is the least significant part of a term $\tau$ appearing in an (in)equality $\tau \sim 0$ of $\phi$, 
  with respect to $\theta\,\big\}$.
  We have already defined the notion of the least significant part for a quotient term induced by~$\theta$ in~\Cref{sec:procedure}. For a (non-quotient) linear-exponential system $\phi$, the least significant part of a term $a \cdot 2^{x_n} + b \cdot x_n + \tau'(x_1,\dots,x_{n-1}, \vec z)$ 
  is the term $b \cdot x_n + \tau'$.
\end{itemize}

Two observations are in order. First, the 
bit size of a system $\phi(x_1,\dots,x_n)$ (i.e., the number of bits required to write down~$\phi$) 
is in $O(\card\phi \cdot n^2 \cdot \log(\max(\onenorm{\phi},\fmod(\phi),2)))$. 
Second, together with the number of variables in the input, our parameters are enough to bound all guesses in the procedure.
For instance, the value of~$c \neq \star$ guessed in line~\ref{line:linearize-guess-const} of~\Cref{algo:linearize} ~(\SolvePrimitive) can be bounded as~$O(\log(\max(\fmod(\gamma),\onenorm{\chi})))$.

The analysis of the whole procedure is rather involved.
Perhaps a good overall picture of 
this analysis is given by the evolution of the parameters at each iteration of the main \textbf{while} loop of~\Master, described in~\Cref{lemma:bounds-one-loop} below. 
This loop iterates at most $n$ times, with $n$ being the number of variables in the input system.
Below, $\totient$ stands for Euler's totient function,
arising naturally because of the computation of multiplicative orders in~\SolvePrimitive.

\begin{restatable}{lemma}{LemmaBoundsOneLoop}
  \label{lemma:bounds-one-loop}
  Consider the execution of~{\rm{\Master}}
  on an input~$\phi(x_1,\dots,x_n)$, with $n \geq 1$. For $i \in [0,n]$, let $(\phi_i,\theta_i)$ be the pair of a system and ordering
  obtained after the $i$th iteration of the \textbf{while} loop of~line~\ref{algo:master:outer-loop}, 
  where $\phi_0 = \phi$ and $\theta_0$ is the ordering guessed in line~\ref{algo:master:guess-order}. 
  Then, for every $i \in [0,n-1]$, $\phi_{i+1}$ has at most $n+1$ variables, and for every $\ell,s,a,c,d \geq 1$, 
  \begin{equation*}
    \text{if}\quad
    \begin{cases}
      \card \lst(\phi_i,\theta_i) &\leq \ell\\
      \card \phi_i &\leq s\\
      \linnorm{\phi_i} &\leq a\\ 
      \onenorm{\phi_i} &\leq c\\
      \fmod(\phi_i) &\hspace{3pt}\divides\hspace{2pt} d
    \end{cases}
    \quad\text{then}
    \quad 
    \begin{cases}
      \card \lst(\phi_{i+1},\theta_{i+1}) &\leq 
      \ell + 2(i+2)\\
      \card \phi_{i+1} &\leq s + 6(i+2) + 2\cdot \ell\\
      \linnorm{\phi_{i+1}} &\leq 3 \cdot a\\ 
      \onenorm{\phi_{i+1}} &\leq \tocheck{2^5(i+3)^2(c+2)} + 4 \cdot \log(d)\\
      \fmod(\phi_{i+1}) &\hspace{3pt}\divides\hspace{2pt} \lcm(d,\totient(\alpha_{i} \cdot d))
    \end{cases}
  \end{equation*}
  for some $\alpha_{i} \in [1,\tocheck{(3 \cdot a + 2)^{(i+3)^2}}]$. 
  The $(i+1)$st iteration of the \textbf{while} loop of~line~\ref{algo:master:outer-loop} 
  runs in non-deterministic polynomial time in the bit size of $\phi_i$.
\end{restatable}

We iterate the bounds in~\Cref{lemma:bounds-one-loop}
to show that, for every $i \in [0,n]$, 
the bit size of $\phi_i$ is polynomial in the bit size of the initial system $\phi$.
A challenge is to bound $\fmod(\phi_i)$, which requires studying iterations of the map $x \mapsto \lcm(x,\totient(\alpha \cdot x))$, where $\alpha$ is some positive integer. We show the following lemma:

\begin{restatable}{lemma}{LemmaTotientBound}
  \label{lemma:totient-bound}
  Let $\alpha \geq 1$ be in~$\N$. Consider the integer sequence $b_0,b_1,\dots$ given by
  the recurrence $b_0 \coloneqq 1$ 
  and $b_{i+1} \coloneqq \lcm(b_i, \totient(\alpha \cdot b_i))$.
  For every $i \in \N$, $b_i \leq \alpha^{2 \cdot i^2}$.
\end{restatable}

Given \Cref{lemma:bounds-one-loop},
one can show $\alpha_{j} \leq (\linnorm{\phi} + 2)^{\tocheck{O(j^3)}}$ for every $j \in [0,n-1]$. 
Then, since $\fmod(\phi_0) = 1$, for a given $i \in [0,n-1]$ we apply~\Cref{lemma:totient-bound}
with $\alpha = \lcm(\alpha_0,\dots,\alpha_{i})$ 
to derive~$\fmod(\phi_{i+1}) \leq (\linnorm{\phi} + 2)^{\tocheck{O(i^6)}}$.
Once a polynomial bound for the bit size of every $\phi_i$ is established,~\Cref{prop:master-in-np} follows immediately
from the last statement of~\Cref{lemma:bounds-one-loop}.

\section{Proofs of Theorem~\ref{thm:buchi-semenov-in-np-z} and Theorem~\ref{thm:buchi-semenov-np}}
\label{sec:wrappers}

In this section, we discuss how to
reduce the task of solving linear-exponential systems over~$\Z$ to the non-negative case,
thus establishing~\Cref{thm:buchi-semenov-in-np-z}.
We also prove~\Cref{thm:buchi-semenov-np}.%

\begin{proof}[Solving linear-exponential systems over $\Z$ (proof of~\Cref{thm:buchi-semenov-in-np-z}).]
Let~$\phi(x_1,\dots,x_n)$ be a linear-exponential system~$\phi(x_1,\dots,x_n)$ (without divisibility constraints).
We can non-deterministi\-cally guess which variables will, in an integer solution~$\vec u \in \Z^n$ of $\phi$, 
assume a non-positive value.
Let $I \subseteq [1,n]$ be the set of indices corresponding to these variables. 
Given $i \in I$, all occurrences of $(x \bmod 2^{x_i})$ in $\phi$ can be replaced with~$0$, 
by definition of the modulo operator.
We can then replace each linear and exponentiated occurrence of $x_i$ with $-x_i$. 
Let $\chi(\vec x)$ be the system obtained from $\phi$ after these replacements.

The absolute value of all entries of $\vec u$ is a solution for $\chi$ over~$\N$. However, $\chi$ might feature terms of the form $2^{-x_i}$ for some $i \in I$ 
and thus not be a linear-exponential system.
We show how to remove such terms. Consider an inequality of the form $\tau \leq \sigma$, where the term $\tau$ contains no $2^{-x}$ and $\sigma \coloneqq \sum_{i \in I} a_i \cdot 2^{-x_i}$ with some $a_i$ non-zero.
Since each $x_i$ is a non-negative integer, we have $\abs{\sum_{i \in I} a_i \cdot 2^{-x_i}} \leq \sum_{i \in I} \abs{a_i} \eqqcolon B$. 
Therefore, in order to satisfy $\tau \leq \sigma$, any solution $\vec v$ of $\chi$ must be such that $\tau(\vec v) \leq B$. 
We can then non-deterministically add to $\chi$ either $\tau < -B$ 
or $\tau = g$, for some $g \in [-B,B]$.
\begin{description} 
  \item[Case $\tau < -B$.] The inequality $\tau \leq \sigma$ is entailed by $\tau < -B$ and can thus be eliminated.
  \item[Case {$\tau = g$ for some $g \in [-B,B]$.}] We replace $\tau \leq \sigma$ with $g \leq \sigma$, and multiply both sides of this inequality by $2^{{\Sigma_{i \in I} x_i}}$. The resulting inequality is rewritten as $g \cdot 2^{z} \leq \sum_{i \in I} a_i \cdot 2^{z_i}$, where $z$ and all $z_i$ are fresh variables (over~$\N$) that are subject to the equalities $z = \sum_{i \in I} x_i$ and $z_i = \sum_{j \in I \setminus \{i\}} x_j$. We add these equalities to~$\chi$.
\end{description}
In the above cases we have removed from~$\chi$ the inequality $\tau \leq \sigma$ in favour of inequalities and equalities only featuring linear-exponential terms.
Strict inequalities~$\tau < \sigma$ can be handled analogously; and for equalities $\tau = \sigma$ one can separately consider $\tau \leq \sigma$ and~${-\tau \leq -\sigma}$. 
The fresh variables $z$ and $z_i$ can be introduced once and reused for all inequalities.

Repeating the process above for each equality and inequality yields (in non-deterministic polynomial time) a linear-exponential system~$\psi$ that is satisfiable over~$\N$ if and only if the input system $\phi$ is satisfiable over~$\Z$. 
The satisfiability of $\psi$ is then checked by calling~\Master. 
Hence, correctness and \np membership follow by~\Cref{prop:master-correct,prop:master-in-np}, respectively.
\end{proof}

\begin{proof}[Deciding existential B\"uchi--Semenov arithmetic (proof of~\Cref{thm:buchi-semenov-np}).]
Let $\phi$ be a formula in the existential theory of the structure~${(\N,0,1,+,2^{(\cdot)},\V(\cdot,\cdot),\leq)}$ (i.e., B\"uchi--Semenov arithmetic).
By De Morgan's laws,~we can bring $\phi$ to negation normal form. Negated literals can then be replaced by positive formulae: 
$\lnot \V(\tau,\sigma)$ becomes $\V(\tau,z) \land \lnot(z = \sigma)$ where $z$ is a fresh variable, 
$\lnot(\tau = \sigma)$ becomes $(\tau < \sigma) \lor (\sigma < \tau)$, and $\lnot(\tau \leq \sigma)$ 
becomes $\sigma < \tau$. 
Next, occurrences of $\V(\cdot,\cdot)$ and $2^{(\cdot)}$ featuring arguments other than variables can be 
``flattened'' by introducing extra (non-negative integer) variables: e.g., an occurrence of $2^{\tau}$ can be replaced with $2^{z}$, 
where $z$ is fresh, subject to conjoining to the formula~$\phi$ the constraint $z = \tau$.
Lastly, recall that $\V(x,y)$ can be rephrased in terms of the modulo operator via a linear-exponential system 
$2\cdot y = 2^v \land 2 \cdot (x \bmod 2^v) = 2^v$, where $v$ is a fresh variable.

After the above transformation, we obtain a formula~$\psi$ of size polynomial with respect to the original one. 
This formula is a positive Boolean combination of linear-exponential systems.
A non-deterministic polynomial-time algorithm deciding~$\psi$ first (non-deterministically) rewrites 
each disjunction $\phi_1 \lor \phi_2$ occurring in $\psi$ into either $\phi_1$ or $\phi_2$. 
After this step, each non-deterministic branch contains a linear-exponential system. 
The algorithm then calls~\Master. Correctness and \np membership then follow by~\Cref{prop:master-correct,prop:master-in-np}.
\end{proof}

\section{Future directions}
\label{sec:conclusion}

We have presented a quantifier elimination procedure that decides in non-deterministic polynomial time 
whether a linear-exponential system has a solution over~$\Z$. 
As a by-product, this result shows that satisfiability for existential B\"uchi--Semenov arithmetic belongs to \np.
We now discuss further directions that, in view of our result, may be worth pursuing. 

As mentioned in~\Cref{sec:related}, the $\exists^*\forall^*$-fragment of B\"uchi--Semenov arithmetic is undecidable. 
Between the existential and the $\exists^*\forall^*$-fragments lies, in a certain sense, the optimisation problem: minimising or maximising a variable subject to a formula. 
It would be interesting to study whether the natural optimisation problem for linear-exponential systems lies within an optimisation counterpart of the class~\np. 

With motivation from verification questions,
problems involving integer exponentiation
have recently been approached
with satisfiability modulo theories (SMT) solvers~\cite{Frohn24}.
The algorithms developed in our paper may be useful to further the research in this direction.

Our work considers exponentiation with a single base. 
In a recent paper~\cite{Hieronymi022}, Hieronymi and Schulz prove the first--order theory of $(\N,0,1,+,2^{\N},3^{\N},\leq)$ undecidable, 
where~$k^{\N}$ is the predicate for the powers of $k$. 
Therefore, the first-order theories of the structures $(\N,0,1,+,V_2,V_3,\leq)$ and $(\N,0,1,+,2^{(\cdot)},3^{(\cdot)},\leq)$, 
which capture $2^{\N}$ and $3^{\N}$, are undecidable.
Decidability for the existential fragments of all the theories in this paragraph is open.

Lastly, it is unclear whether there are interesting relaxed versions of linear-exponential systems, i.e., over~$\R$ instead of $\Z$.
Observe that, in the existential theory of the structure~$(\R,0,1,+,2^{(\cdot)},\leq)$, the formula $x = 2^{y'+z'} \land y = 2^{y'} \land z = 2^{z'}$  
defines the graph of the multiplication function~$x = y \cdot z$ for positive reals. 
This ``relaxation'' seems then only to be decidable subject to (a slightly weaker version of) Schanuel's conjecture~\cite{MacWilkie96}.
To have an unconditional result one may consider systems where only one variable occurs exponentiated. 
These are, in a sense, a relaxed version of $(u,v)$-primitive systems. 
Under this restriction, unconditional decidability was previously proved by Weispfenning~\cite{Weispfenning00}.


\bibliography{bibliography}


\appendix
\section{Theorem \ref{thm:buchi-semenov-in-np-z} holds for any positive integer base given in binary}
\label{app:inter-def}

\Cref{algo:master} 
and~\Cref{algo:elimmaxvar} are agnostic with regard to the choice of the base~$k \geq 2$. 
They do not inspect $k$ and see exponential terms $k^x$ as purely syntactic objects. 
Their logic does not need to be updated to accommodate a different base.
To add support for a base $k$ given in input to these two algorithms, 
it suffices to replace in the pseudocode every $2$ with $k$.

\Cref{algo:linearize} is different, as it uses properties of exponentiation. 
In that algorithm, line~\ref{line:linearize-mod-factor} must be updated as follows.
The pair $(d,n)$ is redefined to be such that $d$ is the largest integer coprime with $k$ dividing $\fmod(\gamma)$,
and $k^n$ is the smallest power of $k$ divisible by $\frac{\fmod(\gamma)}{d}$. 
For example, in the case when $k = 6$ and $\fmod(\gamma) = 60$, we obtain $d=5$ 
and $n=2$, because $36$ is the smallest power of $6$ divisible by $\frac{60}{5}=12$. 
It is clear that $n\leq\lceil\log(\fmod(\gamma))\rceil$, and the pair $(d,n)$ can be computed 
in deterministic polynomial time.

Apart from this update, 
it suffices to replace every occurrence of 
$2^n \cdot r$ with $\frac{\fmod(\phi)}{d} \cdot r$, and every remaining occurrence of
$2$ with $k$ 
(except for the constant $2$ appearing in the expression $3 + 2 \cdot\bigl\lceil\log_k(\frac{|b|+|c|+1}{|a|})\bigr\rceil$).
This means that the discrete logarithm problems of lines~\ref{line:linearize-assert}--\ref{line:linearize-mult-ord} 
must be solved with respect to~$k$ instead of $2$ (but this can still be done in non-deterministic polynomial time).
No other change is necessary.


\newpage




\section{Proofs from Section~\ref{sec:gaussian-elimination}: solving systems of linear inequalities over $\Z$}
\label{app:gaussian-elimination}

This appendix provides a proof of~\Cref{thm:gaussianQE-in-np}. 
We first introduce (in~Section~\ref{app:gaussian-elimination:intro}) a matrix representation 
of the input system of constraints and also fix some other notation needed in the subsequent 
sections. The correctness proof is split into three parts. 
Key properties of the matrices in the variable elimination process are gathered in \Cref{l:fundamental} 
(fundamental lemma) in Section~\ref{app:gaussian-elimination:fundamental}. 
Based on these properties, we next prove that the divisions 
in line~\ref{gauss:divide} of \Cref{algo:gaussianqe} (\GaussianQE) are without remainder, in Section~\ref{app:gaussian-elimination:divisibilities}.
In Section~\ref{app:gaussian-elimination:correctness}
we prove that steps of the algorithm keep the system of constraints equivalent
to the input system.
We then provide the complexity analysis of the algorithm in Section~\ref{app:gaussian-elimination:complexity}, and
\Cref{thm:gaussianQE-in-np} will be a direct consequence of the main statements proved in this appendix.

\vspace{-0.5em}
\subsection{Introduction to Algorithm \ref{algo:gaussianqe} (\GaussianQE) and its analysis}
\label{app:gaussian-elimination:intro}

Throughout this section, we refer to \Cref{algo:gaussianqe} (\GaussianQE) on page~\pageref{algo:gaussianqe}.
The input to \GaussianQE is assumed to be a system (conjunction)~$\phi$
of equalities ($\tau = 0$), non-strict linear inequalities ($\tau \le 0$), and divisibility constraints ($d \divides \rho$).
Strict inequalities can be handled by the addition of $+1$ to the left-hand sides.
Variables of the system~$\phi$ are partitioned
into $\vec x$, to be eliminated, and $\vec z$, to remain in the output.

\subparagraph*{Slack variables (line~\ref{gauss:introduce-slack}).}

In addition to variables $\vec x$ and $\vec z$, 
the algorithm also uses \emph{slack} variables, $\vec y$, which are auxiliary.
The intended domain of slack variables is $\N$. These are internal to the procedure and are eliminated
at the end.

Slack variables are not picked in the header (line~\ref{gauss:mainloop}) of the first \textbf{foreach} loop (below, we refer to this loop as the \emph{main \textbf{foreach} loop}).
Instead, a slack variable $y$ gets eliminated when the substitution $\sub{v}{y}$ is applied to $\phi$.
This substitution is set up in line~\ref{gauss:append-seq},
right after the constraint from which it arises is used to eliminate some variable from~$\vec x$.
Thus, each iteration of the loop eliminates one variable from~$\vec x$ and,
if the chosen constraint stems from an inequality of the input formula, also
one slack variable.

\emph{Lazy addition.}
\GaussianQE performs the substitution $\sub{v}{y}$ lazily:
the choice of the value~$v$ for the variable~$y$ is recorded, but no replacement is carried out in the constraints
until the second \textbf{foreach} loop,
so that variable~$y$ continues to be used.
Thus, the addition of original constant terms in the constraints and the integers arising from the replacement of~$y$ by~$v$
is delayed. 
This laziness turns out convenient when proving
the correctness of the algorithm.
It is, however, easy to see that the \emph{non-lazy} (eager) version of the algorithm
in which each substitution $\sub{v}{y}$ is applied straight away (at line~\ref{gauss:append-seq})
can also be proved correct as long as the lazy version is correct.

\subparagraph*{Matrix representation of systems of constraints.}

A system of constraints (equalities, inequalities, and divisibilities) can be written
in a matrix.
Formally, the \emph{matrix associated to a system $\gamma$} has
rows that correspond to constraints and columns that correspond to variables and constant terms
in~$\gamma$.
An entry in this matrix is the coefficient of the variable (or the constant term)
in the constraint.
We assume that there are no inequalities in~$\gamma$, as they have been replaced
with equalities following the introduction of slack variables in line~\ref{gauss:introduce-slack}
as described above.
Each divisibility constraint $d \divides \rho$ is represented as
the equality $d u + \rho = 0$, where $u$ is a fresh dummy variable (different for each constraint).
%
We denote $\vec u$ the vector of all such dummy variables.

Representation of divisibility constraints in the matrix does not influence in any
way the execution of~\GaussianQE, which manipulates such
constraints as described in the pseudocode.
In other words, the associated matrix is a concept used in the analysis only.

To sum up,
after the introduction of slack variables at line~\ref{gauss:introduce-slack}
and at each iteration of the main \textbf{foreach} loop
of \GaussianQE,
the system~$\phi$
will have the form
\begin{equation}
 \left[ \begin{array}{c|c|c|c} A & \vec c & I' & D \end{array} \right] \cdot
 \left[ \begin{array}{c|c|c|c|c} \vec x^\top & \vec z^\top & -1 & \vec y^\top & \vec u^\top \end{array} \right]^{\!\top} = \vec 0,
\end{equation}
where:
\begin{itemize}
\item $\vec x$ is the vector of input variables to be eliminated;
\item $\vec z$ is the vector of all other input variables;
\item $\vec y$ is the vector of slack variables (internal to the procedure);
\item $\vec u$ is the vector of dummy variables for the analysis of divisibility constraints; and
\item
the factor
$B \defeq \left[ \begin{array}{c|c|c|c} A & \vec c & I' & D \end{array} \right]$
is the \emph{matrix associated to the system~$\phi$:}
\begin{itemize}
\item $A$ (the main part) has as many columns as there are variables in $\vec x$ and $\vec z$ combined,
\item $\vec c$ (the constant-term block) consists of one column only,
\item $I'$ (the slack part) is the identity matrix interspersed with some zero rows, and
\item $D$ (the dummy part) is the diagonal matrix interspersed with some zero rows.
\end{itemize}
\end{itemize}

%


\subparagraph*{Main \textbf{foreach} loop of the algorithm (lines~\ref{gauss:mainloop}--\ref{gauss:restore}).}

We call an iteration of the main \textbf{foreach} loop \emph{nontrivial}
if, at line~\ref{gauss:triviality-check}, some equality of the system~$\phi$
contains the variable~$x$.
For $k \ge 0$, we denote by $\phi_k$ the system of constraints~$\phi$ immediately after $k$~nontrivial iterations of the main \textbf{foreach} loop.
For example, $\phi_0$ is the system obtained
after the introduction of slack variables at line~\ref{gauss:introduce-slack}.


In the non-deterministic execution of \GaussianQE,
each branch computes
the sequence of systems $\phi_0, \phi_1, \ldots, \phi_k, \ldots$\,,
and we denote by $B_0, B_1, \ldots, B_k, \ldots$ the matrices associated with them.
Thus, each non-deterministic branch can be depicted using a commutative diagram:
\begin{equation}
\label{eq:systems-and-matrices}
\begin{CD}
\phi @>>> \phi_0 @>>> \phi_1 @>>> \cdots @>>> \phi_k @>>> \cdots \\
@.        @VVV        @VVV        @.          @VVV        @.     \\
\mbox{}@.    B_0 @>>>    B_1 @>>> \cdots @>>>    B_k @>>> \cdots
\end{CD}
\end{equation}
Each horizontal arrow in the diagram involves a non-deterministic choice
of an equation (line~\ref{gauss:guess-equation}), as well as possibly
a non-deterministic choice of the amount of slack (line~\ref{gauss:guess-slack}).
Naturally, systems $\phi_k$ and matrices $B_k$ across different branches
will, in general, be different (depending on these guesses).

\begin{remark}
  Practically, constraints that originated as equalities (and thus, variables that appear in such constraints) 
  should probably be handled first.
  Theoretically, the order in which variables are chosen does not matter.
  However, if a chosen variable appears in a constraint that originated as equality,
  then in line~\ref{gauss:guess-equation} we may restrict the guessing to such equalities only.
  This restriction of choice eliminates
  all non-deterministic branching (guessing) in this iteration of the main \textbf{foreach} loop.
\end{remark}

\subparagraph*{Assumptions for the proof.}


In our analysis it will be convenient
to make the following two assumptions:
\begin{quote}
\begin{enumerate}
\renewcommand{\theenumi}{(A\arabic{enumi})}
\renewcommand{\labelenumi}{\theenumi}
\item\label{a-rows}
the main \textbf{foreach} loop picks (in line~\ref{gauss:guess-equation}) rows of the matrix in the natural order, i.e., from top to bottom; and
\item\label{a-cols}
the variables are eliminated
(handled in lines~\ref{gauss:guess-equation}--\ref{gauss:restore})
in the natural order of the columns of the matrix.
\end{enumerate}
\end{quote}
Both assumptions are made
only for the sake of convenience of notation,
with no loss of generality.

\begin{figure}[t]
  \centering
  \begin{tikzpicture}

    \node (before) at (0,0) {
      \begin{tikzpicture}
        \draw[fill=blue!20,draw=blue!40] (-5.35,-0.02) rectangle (-0.8,1.3);
        \draw[fill=orange!40,draw=orange] (-2.65,0.07) rectangle (-0.91,1.23);

        \draw[fill=gray!20,draw=black!40] (-5.35,-1.33) rectangle (1.1,-0.08);

        \draw[fill=orange!40,draw=orange!40] (1.73,-1.12) rectangle (3,-0.2);

        \node at (0,0) {
          \begin{minipage}{0.8\linewidth}
            $
            {\renewcommand{\arraystretch}{1.1}
              \left[
                  \begin{array}{cc|c|cc|c}%
                    &&&&&\\[-11pt]
                     a & \raisebox{-2.5pt}{\textbf{*}} & * & 0|1 & & \\[-2pt]
                    \raisebox{-6pt}{\textbf{*}} & \raisebox{-6pt}{\scalebox{1.8}{$A_{\scalebox{0.5}{$0$}}$}} &  \raisebox{-6pt}{\scalebox{1.8}{$\vec c_{\scalebox{0.5}{$0$}}$}} & 
                    \phantom{*} & 
                    \raisebox{-6pt}{\scalebox{1.8}{$\phantom{a}I_{\scalebox{0.5}{$0$}}'$}} & \\
                    &&&&&\\[-10pt]
                    \hline
                    &&&&&\\[-8pt]
                    \textbf{*} & 
                    \scalebox{1.8}{$\,T_{\scalebox{0.5}{$0$}}\,$} & 
                    \raisebox{2pt}{\scalebox{1.8}{$\vec g_{\scalebox{0.5}{$0$}}$}} &
                    \phantom{*} &&
                      {\arraycolsep=0.1pt\def\arraystretch{0.1}
                        \begin{array}{ccc}
                          d_1\phantom{a} \\[-5pt] 
                          & \ddots \\ 
                          && d_n\phantom{a}
                      \end{array}}%
                    \vspace{4pt}%
                  \end{array} \right]
              } \cdot
            \left[ \begin{array}{c} x \\ \phantom{'}\vec x' \\[-2pt] \vec z \\\hline -1\phantom{-} \\\hline y\\ \vec y' \\\hline \vec u \end{array} \right]= \vec 0
            $
          \end{minipage}
        };
      \end{tikzpicture}
    };

    \node (after) at (0,-3.8) {
      \begin{tikzpicture}

        \draw[fill=gray!20,draw=black!40] (-5.35,0.89) rectangle (-2,1.4);

        \draw[fill=blue!20,draw=blue!40] (-4.78,-0.15) rectangle (-0.79,0.83);
        \draw[fill=orange!40,draw=orange] (-1.9,-0.07) rectangle (-0.85,0.73);

        \draw[fill=gray!20,draw=black!40] (-4.78,-1.42) rectangle (1.1,-0.22);

        \draw[fill=magenta!30,draw=magenta!50] (-2.7,-1.3) rectangle (-2.05,1.33);

        \draw[fill=orange!40,draw=orange!40] (1.73,-1.12) rectangle (3,-0.65);

        \draw[fill=magenta!30,draw=magenta!30] (1.73,-0.65) rectangle (3,-0.25);

        \node at (0,0) {
          \begin{minipage}{0.8\linewidth}
            $
            {\renewcommand{\arraystretch}{1.1}
              \left[
                  \begin{array}{c|c|cc|c|c}%
                    &&&&&\\[-10pt]
                     a & \raisebox{-2.5pt}{\textbf{*}} & * & 0|1 & & \\
                    \hline
                    &&&&&\\[-10pt]
                    \phantom{\raisebox{-6pt}{\textbf{*}}} & \raisebox{-6pt}{\scalebox{1.8}{$A_{\scalebox{0.5}{$1$}}$}} &  \raisebox{-6pt}{\scalebox{1.8}{$\vec c_{\scalebox{0.5}{$1$}}$}} & 
                    \raisebox{-6pt}{\textbf{*}} 
                    & 
                    \raisebox{-6pt}{\scalebox{1.8}{$aI_{\scalebox{0.5}{$0$}}'$}} & \\
                    &&&&&\\[-10pt]
                    \hline
                    &&&&&\\[-8pt]
                    & 
                    \scalebox{1.8}{$\,T_{\scalebox{0.5}{$1$}}\,$} & 
                    \raisebox{2pt}{\scalebox{1.8}{$\vec g_{\scalebox{0.5}{$1$}}$}}  &
                    \textbf{*} &&
                    {\arraycolsep=0.1pt\def\arraystretch{0.1}
                      \begin{array}{ccc}
                        d_1a \\[-5pt] 
                        & \ddots \\ 
                        && d_na
                    \end{array}}%
                    \vspace{4pt}%
                  \end{array} \right]
              } \cdot
            \left[ \begin{array}{c} x \\\hline \phantom{'}\vec x' \\[-2pt] \vec z \\\hline -1\phantom{-} \\ y\\\hline \vec y' \\\hline \vec u \end{array} \right]= \vec 0
            $
          \end{minipage}
        };
      \end{tikzpicture}
    };

    \draw[line width=1mm] (before.east) edge[->,bend left=90] node[left,align=center] {first\\[-2pt] elimination\\[-2pt] step} (after.east);
  \end{tikzpicture}

  \caption{Matrix representation of the linear system in input of~\GaussianQE (above), and matrix obtained after the first iteration of the (main) \textbf{foreach} loop of line~\ref{gauss:mainloop} (below). Submatrices that are shown as empty only contain zeros. \textbf{\color{gray}Gray rectangles} represent divisibility constraints $d \divides \tau$, given in the matrix representation as rows $d u + \rho = 0$, where $u$ is a fresh variable ranging over~$\Z$. \textbf{\color{blue!40}Blue rectangles} represent the equalities in the systems; line~\ref{gauss:guess-equation} guesses equalities only from these lines. Inequalities are translated into equalities by introducing \textbf{\color{orange}slack variables} ranging over $\N$ (line~\ref{gauss:introduce-slack}).
  The matrix highlighted with the \textbf{\color{orange}orange rectangle} is the identity matrix interspersed with zero rows; its non-zero rows correspond to inequalities in the original system.
  When a variable is eliminated, the procedure may ``lazily'' assign a value to a slack variable (see variable~$y$ highlighted in \textbf{\color{magenta!60}magenta}; the corresponding \textbf{\color{magenta!60}magenta column} should be understood as a ``constant column'' once the procedure assigns a value to $y$). The variables in this figure have the following roles: $\vec u$ are the auxiliary variables encoding divisibilities, $x$ and $\vec x'$ are the variables to be eliminated, $\vec z$ are free variables in the input system that will not be eliminated, $y$ and $\vec y'$ are the slack variables. 
  The $-1$ in the column vector corresponds to the column of constants.} 
  \label{figure:gaussian-process}
\end{figure}

\subparagraph*{Operations on nontrivial iterations (lines~\ref{gauss:guess-equation}--\ref{gauss:divide}).}

The process followed by a single non-deterministic branch
is
the \emph{Bareiss-style fraction-free one-step elimination}~\cite{Bareiss68}.
We compare it against the standard Gauss--Jordan process from linear algebra 
for variable elimination (see, e.g.,~\cite{Strang}).
An example with $k=1$, that is, the first nontrivial iteration of the process, 
is depicted in \Cref{figure:gaussian-process}.

To begin with, note that
applying a substitution $\vigsub{\frac{-\tau}{a}}{x}$ to an equality
$b x + \sigma = 0$ in line~\ref{gauss:vigorous} is equivalent to first multiplying the equality by the lead
coefficient~$a$
and then subtracting the equality $a x + \tau = 0$ multiplied by $b$.
The result is $- b \tau + a \sigma = 0$.
(For example, for $a = b$ and $\tau = \sigma$, we have
 $(a x + \tau = 0) \vigsub{\frac{-\tau}{a}}{x} = (- a \tau + a \tau = 0)$,
 which simplifies to true.)
Thus, the formula $\phi \vigsub{\frac{-\tau}{a}}{x}$ is the result
of applying this operation to all constraints in $\phi$.
We use vigorous substitutions instead of standard substitutions
in the main elimination process.
This is in order to simplify
the handling of (and reasoning about) ``Bareiss factors'',
explained in more detail in Section~\ref{app:gaussian-elimination:fundamental}.

Let us give a matrix representation to these operations.
Note that, by Assumption~\ref{a-rows},
nontrivial iteration~$k$ uses the $k$th row of the matrix~$B_k$ as the ``lead row''.
Thus, in effect,
on iteration $k$,
\GaussianQE applies the following row operations to~$B_k$.
Line~\ref{gauss:vig}
multiplies all rows by the lead coefficient, then
subtracts from each row indexed~$i$
the lead row (indexed~$k$) with the multiplier equal to the original coefficient of~$x$ in row~$i$.
The lead row gets temporarily ``zeroed out''; we discuss this in more detail below.
Line~\ref{gauss:divide} divides each row by $p$;
Lemma~\ref{l:no-remainder} below will prove that
this division is without remainder.
The substitution in line~\ref{gauss:append-seq} corresponds to subtracting the $y$-column of the matrix
multiplied by~$v$ from the constant-term column; but this operation is only carried out later
in line~\ref{gauss:apply-seq}
(effectively, this operation is $\vec c \gets \vec c - v \vec s$, where $\vec s$
 is the vector of coefficients of the slack variable~$y$).

Notice that the Gauss--Jordan elimination process
would, in comparison,
leave the lead row (indexed~$k$) unaffected by these operations.
(In particular,
it would subtract the lead row from all other rows (indexed~$i \ne k$) only, not from itself.)

We remark that for the rows of $B_k$ that correspond to divisibility constraints
the same reasoning applies.
For example, recall from Section~\ref{sec:preliminaries} that multiplication and division
in \GaussianQE
are applied to
both sides of divisibility constraints: for $\lambda$ a nonzero integer,
$d \divides \rho$ may evolve (at line~\ref{gauss:vigorous}) into $\lambda d \divides \lambda \rho$
or, if $d$ and all numbers occurring in $\rho$ are divisible by~$d$ (at line~\ref{gauss:divide}),
into $\lambda^{-1} d \divides \lambda^{-1} \rho$.
When represented as an equation:
$d u + \rho = 0$ may evolve into $(\lambda d) u + \lambda \rho = 0$
or $(\lambda^{-1} d) u + \lambda^{-1} \rho = 0$, respectively.

\subparagraph*{Recasting the lead row as a divisibility constraint (line~\ref{gauss:restore}).}

We already saw that, in line~\ref{gauss:vig}, $(a x + \tau = 0) \vigsub{\frac{-\tau}{a}}{x} = (- a \tau + a \tau = 0)$,
which simplifies to true.
Thus, this line removes the equality $a x + \tau = 0$ from~$\phi_k$.
Almost immediately, line~\ref{gauss:restore} reinstates it,
although recasting it as a divisibility constraint.
This is indeed possible: 
in the Gauss--Jordan elimination process the equality $a x + \tau = 0$ would still be present,
but at the same time
the variable~$x$ would now only be occurring in this equality.
Therefore, a suitable value from \Z can be assigned to this variable
if and only if the lead coefficient~$a$ divides the value of~$\tau$.

The fact that no other constraint apart from $a x + \tau = 0$ contains the variable~$x$
depends on the fact
that all divisibility constraints within $\phi_k$ are affected by
the vigorous substitution. The alternative --- leaving them unaffected ---
would correspond to the Gauss-style elimination process, which brings the matrix
to non-reduced row echelon form, rather than reduced row echelon form;
see, e.g.,~\cite{Strang}.
However, care need to be taken in such a modified process to keep the same values
of subdeterminants to establish an analogue of Lemma~\ref{l:no-remainder}.

\subsection{Fundamental lemma of the Gauss--Jordan--Bareiss process}
\label{app:gaussian-elimination:fundamental}

This subsection is devoted to a key property of the Gauss--Jordan variable elimination process
with modifications \`a la Bareiss~\cite{Bareiss68}.
We first discuss a similar property of the \emph{standard} Gaussian elimination process
and then carry the idea over to our algorithm.

Consider the standard process
for a system of equations (equalities) over $\mathbb R$ or $\mathbb Q$
in which rows are never permuted or multiplied
by constants (but only added to one another, possibly with some real multipliers).
Then, subject to assumptions related to \ref{a-rows} and \ref{a-cols} above,
the following two properties hold;
we refer the reader to, e.g.,~\cite[Section~3.3]{Schrijver99}:
\begin{itemize}
\item
The leading principal minors%
\footnote{%
    A \emph{minor} of a matrix is the determinant of a square submatrix.
    A synonym is \emph{subdeterminant}.
} $\overline\mu_1, \overline\mu_2, \ldots$ of the matrix
(i.e., those formed by the first $k$~rows and first $k$~columns, for some $k$)
remain unchanged throughout.
For convenience, we denote $\overline\mu_0 = 1$.
\item
After $k \ge 0$ steps of the process, the entry in position $(i, j)$ for $i, j > k$
is the ratio $a_{i j}^{(k)} / \overline\mu_k$,
where $a_{i j}^{(k)}$ is the $k$th leading principal minor $\overline\mu_k$ bordered by the $i$th row
and $j$th column, that is,
\begin{equation*}
a_{i j}^{(k)} \eqdef
\begin{vmatrix}
a_{1 1} & \ldots & a_{1 k} & a_{1 j} \\
\vdots  & \ddots & \vdots  & \vdots  \\
a_{k 1} & \ldots & a_{k k} & a_{k j} \\
a_{i 1} & \ldots & a_{i k} & a_{i j}
\end{vmatrix}\enspace.
\end{equation*}
This is because, after $k$ steps, the entries in positions $(i, 1)$, \ldots, $(i, k)$ are all $0$.
\end{itemize}

The statement of the following (fundamental) lemma refers to
the bordered minor $b_{i j}^{(k)}$, which is determined, as above (and as in Bareiss' paper~\cite{Bareiss68}),
by the entries of the original matrix~$B_0 = (b_{i j})$ of the system, from Equation~\eqref{eq:systems-and-matrices}:
\begin{equation*}
b_{i j}^{(k)} \eqdef
\begin{vmatrix}
b_{1 1} & \ldots & b_{1 k} & b_{1 j} \\
\vdots  & \ddots & \vdots  & \vdots  \\
b_{k 1} & \ldots & b_{k k} & b_{k j} \\
b_{i 1} & \ldots & b_{i k} & b_{i j}
\end{vmatrix}\enspace.
\end{equation*}
The index~$j$ in this notation
may refer to a column within any of the four blocks $A$, $-\vec c$, $I'$, $D$.
A particular case of bordered minors is the $k$th leading principal minor:
\begin{equation*}
\mu_k \defeq b_{k,k}^{(k-1)},
\end{equation*}
that is, the determinant formed by the first $k$ rows and the first $k$ columns.

\begin{lemma}[fundamental lemma]
\label{l:fundamental}
Consider a branch of non-deterministic execution of \GaussianQE. The following statements hold:
\begin{quote}
\begin{enumerate}
\renewcommand{\theenumi}{\textup{(\alph{enumi})}}
\renewcommand{\labelenumi}{\theenumi}
\item
\label{det-property}
For all $k \ge 0$, all $i> k$ and $j > k$, the entry in position $(i,j)$ of the matrix~$B_k$
is~$b_{i j}^{(k)}$.
\item
\label{lead}
For all $k \ge 1$,
the lead coefficient in the $k$th nontrivial iteration is $\mu_k$.
\item
\label{no-remainder}
The division in line~\ref{gauss:divide} is without remainder,
except possibly in constraints introduced in some earlier iteration(s) by line~\ref{gauss:restore}.
\end{enumerate}
\end{quote}
\end{lemma}

\begin{proof}
Induction on $k$.
Part~\ref{no-remainder} for iteration~$k$ is
proved inductively along with parts~\ref{det-property} and~\ref{lead}.

For $k = 0$:
\begin{quote}
\begin{enumerate}
\renewcommand{\theenumi}{\textup{(}\alph{enumi}\textup{)}}
\renewcommand{\labelenumi}{\theenumi}
\item For all $i \ge 1$, $j \ge 1$, the entry in position $(i,j)$ is simply $b_{i j}$.
\item There is nothing to prove.
\item Trivial.
\end{enumerate}
\end{quote}

Let $k \ge 1$.
Notice that part~\ref{lead} is a consequence of the inductive hypothesis (for $k-1$),
namely of part~\ref{det-property} with $i = j = k$.
We show how, in an inductive step, to arrive at the statement of part~\ref{det-property}.
Observe that row and column operations applied to the initial matrix $B_0$ can be represented
by rational square matrices~$L_k$ and~$U_k$ such that $B_k = L_k \cdot B_0 \cdot U_k$.
We track the evolution of the matrix $F_k \defeq L_k \cdot B_0$, without considering
column operations (i.e., substitutions of integers for slack variables).
This corresponds to the \emph{laziness} of our algorithm. 

By the inductive hypothesis for $k-1$, part~\ref{lead}, and thanks to
Assumptions~\ref{a-rows} and~\ref{a-cols}, the lead coefficient
in iteration~$k \ge 1$ is $b_{k,k}^{(k-1)} = \mu_{k}$.
According to the definition of vigorous substitution,
each row~$i > k$ in the system is first multiplied by $a = \mu_k$, and then
row~$k$ is subtracted from it $b$~times where $b$ is the coefficient at~$x$
in the equality that corresponds to row~$i$.
Let $F$ be the matrix associated to the system
$\phi \vigsub{\frac{-\tau}{a}}{x}$.
(Note that $F$ does not reflect any substitution of values for slack variables
 during the execution of the algorithm.)
Observe that, for $i > k$ and $j > k$,
\begin{equation*}
f_{i j} =
\begin{vmatrix}
b_{k k}^{(k-1)} & b_{k j}^{(k-1)} \\[1ex]
b_{i k}^{(k-1)} & b_{i j}^{(k-1)}
\end{vmatrix}
=
\mu_{k-1} \cdot b_{i j}^{(k)}.
\end{equation*}
Here, the entries of the $2 \times 2$ determinant are the entries
of the matrix obtained on the previous step of the algorithm (whilst skipping all column operations);
this follows from the inductive hypothesis for $k-1$, part~\ref{det-property}.
The fact that $f_{i j}$ is equal to this determinant can be verified directly
(and is crucial, e.g., for Bareiss' paper~\cite{Bareiss68}).
Finally, the final equality is the Desnanot--Jacobi identity
(or the Sylvester determinantal identity); see, e.g.,~\cite{Bareiss68,Dodgson1867,KarapiperiRR}.

Note that the argument above applies regardless of whether the index~$j$
is referring to a column of~$A$, to the slack part of the matrix,
or to the constant-term column.
This completes the proof of part~\ref{det-property}.

Finally,
we now see that every entry of each row of the matrix~$F$ indexed $i > k$ is $\mu_{k-1}$
multiplied by an integer. By the inductive hypothesis, part~\ref{lead}, the lead coefficient
in iteration $k-1$ is $\mu_{k-1}$, and therefore at iteration~$k$
the algorithm has $p = \mu_{k-1}$.
Notice that
line~\ref{gauss:guess-equation} at iteration~$k-1$ ensures that the coefficient at~$x$
is non-zero,
so $\mu_{k-1} \ne 0$.
Therefore, the division
in line~\ref{gauss:divide} is without remainder, proving part~\ref{no-remainder}.
\end{proof}

In the following 
Section~\ref{app:gaussian-elimination:divisibilities}
we strengthen Lemma~\ref{l:fundamental}\ref{no-remainder}, 
showing that in fact there cannot be any exceptions at all:
division in line~\ref{gauss:divide} is always without remainder.

It would be fair to dub
the divisor~$p$ in line~\ref{gauss:divide} of \GaussianQE,
or equivalently
the factor~$\mu_{k-1}$ identified in the final paragraph of the proof above,
`the Bareiss factor'.
(We will not actually need this term in the paper.)

\begin{remark}
\label{remark:nonzero}
As seen from the proof of Lemma~\ref{l:fundamental},
Assumptions~\ref{a-rows} and~\ref{a-cols}
imply that $\mu_k \ne 0$ as long as the main \textbf{foreach} loop runs for at least~$k$ iterations.
This is because equations in which the coefficient at~$x$ equals zero
are not considered in line~\ref{gauss:guess-equation} of \GaussianQE.
\end{remark}

The following two basic facts about~\GaussianQE,
which can be established independently, are in fact direct consequences of
Lemma~\ref{l:fundamental}.

\begin{lemma}
\label{l:slack-separate}
At the beginning of each iteration of the main \textbf{foreach} loop,
each slack variable that is not yet assigned a value by~$s$
occurs in a single constraint of~$\phi$ (it is an equality).
\end{lemma}

\begin{proof}
Let column index~$j$ correspond to the slack variable in question, which we denote by~$y$.
If $y$ appears in a divisibility constraint that was originally introduced by line~\ref{gauss:restore},
then it must have been introduced to that constraint by a substitution of
line~\ref{gauss:vigorous}, but then line~\ref{gauss:guess-slack} must have assigned a value to~$y$.
Therefore, if the slack variable is not yet assigned a value, then it cannot appear in the divisibility
constraints introduced by line~\ref{gauss:restore}.

We now consider equalities present in the system; this restricts us to rows of the current matrix
with indices $i>k$, where $k$ denotes the number of completed nontrivial iterations so far.
Lemma~\ref{l:fundamental}\ref{det-property} applies to these rows.
Take the constraint corresponding to row~$i$.
If variable~$y$ was not present in it originally, then the $j$th column of the
bordered minor~$b_{i j}^{(k)}$ is zero, and thus $b_{i j}^{(k)} = 0$. In other words, $y$
is still absent from this constraint after~$k$ iterations.
Since originally each slack variable only features in one constraint, the lemma follows.
\end{proof}

Lemma~\ref{l:slack-separate} shows that the operation in line~\ref{gauss:append-seq} is
unambiguous, as there cannot be more than one suitable variable~$y$ in the preceding line~\ref{gauss:branch}.
This is because there is originally at most one slack variable in each equality, and
the lemma holds throughout the entire run of the algorithm.
For the same reason, line~\ref{gauss:drop-slack} is unambiguous as well.

\begin{lemma}
\label{l:vigorous-slack-same}
At every iteration of the main \textbf{foreach} loop,
all slack variables that are not yet assigned a value by~$s$
occur in the constraints of~$\phi$ with identical coefficients,
namely~$\mu_k$ after $k$~nontrivial iterations.
\end{lemma}

\begin{proof}
By Lemma~\ref{l:fundamental}\ref{lead},
the lead coefficient in nontrivial iteration~$i$ is $\mu_i$; $1 \le i \le k$.
This is also the divisor in line~\ref{gauss:divide} in nontrivial iteration~$i+1 \le k$.
Therefore, after $k$~iterations the slack variable~$y$ occurs with coefficient
\begin{equation*}
1 \cdot
\dfrac{\mu_1 \cdot \mu_2 \ldots \mu_{k}}{1 \cdot \mu_1 \ldots \mu_{k-1}} = \mu_k.\qedhere
\end{equation*}
\end{proof}

\subsection{Integers that appear in the run of Algorithm \ref{algo:gaussianqe} (\GaussianQE)}
\label{app:gaussian-elimination:divisibilities}

An important milestone is to prove that
division in line~\ref{gauss:divide} of~\GaussianQE is without remainder.
Lemma~\ref{l:fundamental}, part~\ref{no-remainder}, already proves this for
all constraints except those introduced by line~\ref{gauss:restore}.
In this section, we analyse these constraints, as well as other divisibility constraints.

We use the same notation, namely $B_k$ and $\mu_k$, as in the previous Section~\ref{app:gaussian-elimination:fundamental}.

\begin{lemma}
\label{l:divisors}
Assume that all divisions in line~\ref{gauss:divide}
in the first~$k$ nontrivial iterations of the main \textbf{foreach} loop are without remainder.
Let
$d \divides \rho$ be a divisibility constraint in $\phi_k$. Then:
\begin{itemize}
\item
$d = d^\circ \cdot \mu_k$
if this constraint has evolved from a constraint $d^\circ \divides \rho^\circ$ present in~$\phi_0$, and
\item
$d = \mu_k$
if this constraint was introduced to the system at line~\ref{gauss:restore}.
\end{itemize}
\end{lemma}

\begin{proof}
Let $d^\circ \divides \rho^\circ$ be a constraint present in~$\phi_0$.
By our definition of vigorous substitutions,
during the run of \GaussianQE the divisor (modulus)~$d^\circ$ is multiplied by
\begin{equation*}
1, \mu_1, 1^{-1}, \mu_2, \mu_1^{-1}, \ldots, \mu_k, \mu_{k-1}^{-1}.
\end{equation*}
Here we used Lemma~\ref{l:fundamental}, part~\ref{lead}, which shows that the lead coefficients
are principal minors of the matrix~$B_0$; and Remark~\ref{remark:nonzero}.
The product of the factors listed above is $\mu_k$; thus, the divisor (modulus)
evolves from~$d^\circ$ into~$d^\circ \cdot \mu_k$.

Now consider a divisibility constraint introduced to the system by line~\ref{gauss:restore},
say in the $i$th nontrivial iteration.
By Lemma~\ref{l:fundamental}, part~\ref{lead}, the divisor in this constraint is $\mu_i$.
In (nontrivial) iterations~$i+1$ through~$k$, this divisor is multiplied
by $\mu_{i+1}$, $\mu_i^{-1}$, $\mu_{i+2}$, $\mu_{i+1}^{-1}$, \ldots, $\mu_k$, $\mu_{k-1}^{-1}$.
The product of these factors is $\mu_k / \mu_i$, and so the result
is $\mu_k$.

We remark that, in both scenarios, the assumption of the lemma is required so that
$d$ and $\rho$ are well-defined.
\end{proof}


Our next result is an analogue of Lemma~\ref{l:fundamental}, part~\ref{det-property}, for
integers that appear in divisibility constraints.
Part~\ref{l:cramer:mid-iteration} of Lemma~\ref{l:cramer} will be key for the inductive proof that
all divisions in line~\ref{gauss:divide} are without remainder.

\begin{lemma}
\label{l:cramer}
Assume that all divisions in line~\ref{gauss:divide}
in the first~$k$ nontrivial iterations of the main \textbf{foreach} loop are without remainder.
\begin{quote}
\begin{enumerate}
\renewcommand{\theenumi}{\textup{(\alph{enumi})}}
\renewcommand{\labelenumi}{\theenumi}
\item\label{l:cramer:in-between}
After these iterations,
for all $i \le k$ and all $j$, the entry in position $(i,j)$ of the obtained matrix
is equal to the minor of~$B_0$ formed by the first $k$~rows and columns $1, \ldots, i-1,j,i+1, \ldots, k$.
\item\label{l:cramer:mid-iteration}
In the $(k+1)$th nontrivial iteration (if it exists),
just before line~\ref{gauss:divide},
for all $i \le k+1$ and all $j$, the entry in position $(i,j)$ of the obtained matrix
is equal to $\mu_k$ times the minor of~$B_0$ formed by the first $k+1$~rows and columns $1, \ldots, i-1,j,i+1, \ldots, k+1$.
\end{enumerate}
\end{quote}
\end{lemma}

Before proving this lemma, we give two restatements of its first part for the reader's convenience.
For the proof as well as for these restatements,
we will need the standard notion of the \emph{adjugate} of a $k \times k$ rational matrix~$M$.
It is the transpose of the cofactor matrix of $M$:
$\adjj{M} = (m_{i j})$, where
$m_{i j}$ is the determinant of the matrix obtained from $M$ by removing
the $j$th row and the $i$th column.
If $M$ is invertible, then $\adjj{M} = M^{-1} \cdot \det M$.

Now,
let $M_k$ be the submatrix of~$B_0$
formed by
the first $k$~rows and first $k$~columns.
Suppose line~\ref{gauss:restore} of \GaussianQE introduces,
on the $i$th nontrivial iteration
of the main \textbf{foreach} loop,
a~divisibility constraint which becomes
$d \divides \rho$ after the $k$th nontrivial iteration ($1 \le i \le k$).
Then,
assuming as above that all the performed divisions
are without remainder:
\begin{itemize}
\item
The term $\rho$ is the linear combination of the first~$k$ constraints of the original system
with coefficients from the $i$th row of $\adjj{M_k}$, and with the first~$k$ handled variables
removed. In other words,
row $\ell$ is taken with coefficient equal to the $(\ell,i)$ cofactor of $M_k$.
\item
For every variable~$w$ occurring in $\rho$, its coefficient
is equal to the determinant of the matrix obtained from $M_k$ by replacing
the $i$th column with the vector of coefficients of~$w$ in the first~$k$
constraints of the original system.
\item
The constant term of $\rho$
is equal to the determinant of the matrix obtained from $M_k$ by replacing
the $i$th column with the vector of constant terms of the first~$k$
constraints of the original system.
\end{itemize}

\begin{proof}
We focus on part~\ref{l:cramer:in-between} first.
If all the divisions mentioned in the statement of the lemma are without remainder,
then, by Lemma~\ref{l:divisors} (second part), after the $k$~nontrivial iterations,
the square submatrix formed by the first $k$~rows and first $k$~columns
is equal to $\mu_k I$, where $I$ is the $k \times k$ identity matrix.
This is because
our algorithm is a variant of the Gauss--Jordan variable elimination:
in particular, after each nontrivial iteration (say~$e$) the entire $e$th column
of the matrix becomes zero, with the exception of the entry in position~$(e,e)$.

Let $B^{(1..k)}$ denote the submatrix of $B_0$
formed by the first~$k$ rows.
Consider the effect of the (nontrivial) $k$~iterations
of the main \textbf{foreach} loop on $B^{(1..k)}$.
These operations amount to manipulating the rows of $B^{(1..k)}$, namely to multiplication of $B^{(1..k)}$ from the left
by a square matrix, which we denote by~$L$.
Denote, as above, by $M_k$ the submatrix of $B^{(1..k)}$ formed by
the first~$k$ columns.
We know that $L \cdot M_k = \mu_k I$.
Since $\mu_k = \det M_k \ne 0$ by Lemma~\ref{l:fundamental}\ref{lead} and Remark~\ref{remark:nonzero}, it follows that
$L = \adjj{M_k}$.

We now consider three cases, depending on the position of the entry $(i,j)$ in the matrix.
\begin{description}
\item[Case $i = j \le k$.]
This is a diagonal entry of the obtained matrix, within the first $k$~rows.
We already saw above that,
by Lemma~\ref{l:divisors} (second part), this entry is equal to $\mu_k$.
And indeed, the minor of~$B_0$ formed by the first $k$~rows and columns $1, \ldots, i-1,j,i+1, \ldots, k$
is in this case simply $\det M_k$.
\item[Case $i \ne j \le k$.]
This is an off-diagonal entry within the first $k$ columns.
Again, we have already seen that, by the definition of the Gauss--Jordan process,
this entry must be $0$.
Indeed, the minor of~$B_0$ formed by the first $k$~rows and columns $1, \ldots, i-1,j,i+1, \ldots, k$
in this case contains a repeated column, namely column~$j$.
\item[Case $j > k$.]
This is the main case, when the entry in question lies to the right of the $\mu_k I$ submatrix.
Let $\vec b$ denote the vector formed by entries of~$B_0$ in column~$j$ and rows $1$ through~$k$.
The entry in question is the $i$th component of $\adjj{M_k} \cdot \vec b = M_k^{-1} \vec b \cdot \det M_k$,
or in other words of the solution to the system of equations $M_k \cdot \vec w = \vec b \cdot \det M_k$, where $\vec w$ is a vector of fresh variables.
By Cramer's rule, this component is equal to the determinant of the matrix obtained from $M_k$
by replacing the $i$th column by $\vec b \cdot \det M_k$, divided by $\det M_k$. This is exactly
the minor from the statement of the lemma.
\end{description}
This completes the proof of part~\ref{l:cramer:in-between}.
To justify part~\ref{l:cramer:mid-iteration}, we observe
that all of our arguments remain valid for the $(k+1)$th iteration,
except that the submatrix formed by the first $k+1$~rows and first~$k+1$ columns
is now $\mu_{k+1}\mu_k I$. This is because the $(k+1)$st iteration multiplies $\mu_k I$
by the new lead coefficient, which is $\mu_{k+1}$ by Lemma~\ref{l:fundamental}, part~\ref{lead}.
The remaining reasoning goes through almost unchanged:
we now have $L' \cdot M_{k+1} = \mu_{k+1} \mu_k I$, where $I$ is $(k+1) \times (k+1)$,
and so $L' = \mu_k \cdot \adjj{M_{k+1}}$ instead of $L = \adjj{M_k}$.
This introduces the extra factor of $\mu_k$, matching the statement of the present
Lemma~\ref{l:cramer}, part~\ref{l:cramer:mid-iteration}.
\end{proof}

\begin{lemma}
\label{l:no-remainder}
The division in line~\ref{gauss:divide} is without remainder.
\end{lemma}

\begin{proof}
By Lemma~\ref{l:fundamental}, part~\ref{no-remainder}, we can focus
on constraints introduced by line~\ref{gauss:restore} only.
The proof is by induction on~$k$, the index of a nontrivial iteration
of the main \textbf{foreach} loop.

The base case is $k = 1$. No constraints have been introduced priori to the $1$st iteration,
and thus there is nothing to prove.
In the inductive step, we assume that the statement holds for the first $k$~nontrivial iterations.
Thus, Lemma~\ref{l:cramer}, part~\ref{l:cramer:mid-iteration}, applies.
But the factor $\mu_k$ from its statement is, by Lemma~\ref{l:fundamental}, part~\ref{lead}, the lead coefficient of the
$k$th nontrivial iteration: $p = \mu_k$. Therefore, the division by~$p$ in
line~\ref{gauss:restore} is indeed without remainder.
\end{proof}

\subsection{Correctness of Algorithm \ref{algo:gaussianqe} (\GaussianQE)}
\label{app:gaussian-elimination:correctness}

In this subsection we show that~\GaussianQE
correctly implement its specification.

We first make a basic observation underpinning the proof of correctness of \GaussianQE.
Fix some assignment to target variables~$\vec x$ and free variables $\vec z$.
It is clear that
the input conjunction~$\phi$ of equalities, inequalities, and divisibility constraints
is true if and only if there are nonnegative amounts of slack
(that is, an assignment of values from \N to slack variables~$\vec y$)
that make the equations with slack
produced in line~\ref{gauss:introduce-slack} of \GaussianQE as well as all the divisibility constraints true.
This equivalence justifies line~\ref{gauss:introduce-slack} of \GaussianQE.

We now prove a sequence of three lemmas, which will later be combined
into a proof of correctness of~\GaussianQE.
For all lemmas in this section,
it is convenient to think of $\phi$ as a logical \emph{formula}.

\begin{lemma}
\label{l:no-shrink}
An assignment satisfies formula~$\phi$ after lines~\ref{gauss:vigorous}--\ref{gauss:restore}
if and only if
it has an extension to~$x$ that satisfies~$\phi$ just before these lines.
In other words:
\[ 
  \exists x \,\phi_{k} \iff \phi_{k+1}\,,
  \qquad\quad k \geq 0.
\]
\end{lemma}

\begin{proof}
Let $\nu$ be a satisfying assignment for the formula $\phi$
just before line~\ref{gauss:vigorous}.
In this context, $\nu$~assigns values to all variables in $\vec x$, $\vec y$, $\vec z$.
All constraints in the formula $\phi \vigsub{\frac{-\tau}{a}}{x}$ are obtained
by multiplying constraints of $\phi$ by integers as well as
adding such constraints together.
Therefore, all these constraints are also satisfied by $\nu$.
Dividing both sides of a constraint by a non-zero integer does not change
the set of satisfying assignments either.
Finally, for line~\ref{gauss:restore} we note that if $\nu(a x + \tau) = 0$, then
certainly $a$ divides $\nu(\tau)$.

In the other direction,
let $\nu'$ be a satisfying assignment for the formula obtained after line~\ref{gauss:restore}.
Here $\nu$ assigns values to all variables to $\vec x$, $\vec y$, $\vec z$ except~$x$.
Observe that reversing the application of line~\ref{gauss:divide} preserves the satisfaction,
because it amounts to multiplying both sides of each constraint by a non-zero integer.
We now show that the value for~$x$ can be chosen such that the formula $\phi$
at hand \emph{just before} line~\ref{gauss:vigorous} is satisfied.
Indeed, thanks to line~\ref{gauss:restore} the integer $\nu'(\tau)$ is a multiple of~$a$.
We will show that assigning $-\nu'(\tau) / a$ to $x$ satisfies $\phi$;
note that this choice of value ensures that $\nu'(a x + \tau) = 0$.

Consider any nontrivial (different from Boolean true) constraint in~$\phi\vigsub{\frac{-\tau}{a}}{x}$.
If it stems from a constraint $b x + \sigma = 0$, for some (possibly zero) integer $b$, then
this (satisfied) constraint is actually $- b\tau + a\sigma = 0$, simply by the definition
of substitution. Therefore, we have $-b\nu'(\tau) + a\nu'(\sigma) = 0$ and $a \nu'(x) + \nu'(\tau) = 0$,
from which we conclude $a b \nu'(x) + a \nu'(\sigma) = 0$. Since $a \ne 0$,
the constraint $b x + \sigma = 0$ is indeed satisfied by our choice for the value of~$x$.
\end{proof}

The reader should not be lulled into a false sense of security
by the seemingly very powerful equivalence in the displayed equation
in Lemma~\ref{l:no-shrink}.
While lines~\ref{gauss:vig}--\ref{gauss:restore} of \GaussianQE indeed eliminate~$x$,
the primary objective after the introduction of slack (line~\ref{gauss:introduce-slack})
is to bound the range of slack variables.
This is not achieved by Lemma~\ref{l:no-shrink}, but is the subject
of the following Lemma~\ref{l:amount-of-slack}.

For a sequence of substitutions~$s$,
we denote by $\varsofsub(s)$ the set of variables that $s$~assigns values to. 
Also, by $\vec x \setminus \{y_1,\dots,y_k\}$ we denote the vector obtained from $\vec x$ by removing $y_1,\dots,y_k$.
Lemma~\ref{l:amount-of-slack} shows that
the range of the guessed slack in line~\ref{gauss:guess-slack} of \GaussianQE suffices
for completeness.

\begin{lemma}
\label{l:amount-of-slack}
Consider the $(k+1)$th nontrivial iteration of the main \textbf{foreach} loop, for $k \ge 0$.
Fix an arbitrary assignment~$\xi$ to $\vec z$, $\vec x \setminus x$,
and~$\varsofsub(s)$.
Then, under~$\xi$, the formula~$\phi_k$ is satisfiable
if and only if,
for some choice of guesses in lines~\ref{gauss:guess-equation}--\ref{gauss:append-seq},
the formula~$\phi_{k+1}\sub{v}{y}$
(or~$\phi_{k+1}$ if there is no slack at line~\ref{gauss:branch})
is satisfiable. In other words: 
\[ 
  \exists x \exists \vec y' \,(\vec y' \geq \vec 0 \land \phi_k)
  \iff 
  \bigvee_{i \in I} \exists \vec y_i' (\vec y_i' \geq \vec 0 \land  \phi_{k+1}^{(i)}\sub{v_i}{y_i}) \lor \bigvee_{j \in J}\exists \vec y'(\vec y' \geq \vec 0 \land \phi_{k+1}^{(j)}),
\]
where $I$ and $J$ are sets of indices corresponding to the guesses at lines~\ref{gauss:guess-equation}--\ref{gauss:append-seq}, with $I$ corresponding to those featuring equalities with slack variables, and $\vec y'$ and $\vec y_i'$ are those slack variables that are not assigned by $s$ before and after the $(k+1)$th iteration takes place.  
\end{lemma}

\begin{proof}
Consider the set of equalities in $\phi_k$ in which $x$ appears.
If in some equality of $\phi_k$ involving $x$ all slack variables have already been assigned
values by substitutions of~$s$,
then we may restrict the choice at line~\ref{gauss:guess-equation} to such equalities only.
Indeed, in this case
there is no slack at line~\ref{gauss:branch}
and line~\ref{gauss:guess-slack} is not executed at all.
Thus, in this case the statement follows directly from Lemma~\ref{l:no-shrink}.

We can therefore assume without loss of generality that,
in every equality of $\phi_k$ that contains $x$,
there is some slack variable that has not been assigned a value by~$s$ yet.
With no loss of generality, we assume these slack variables are exactly $\vec y' = (y_1, \ldots, y_r)$.
(In general, the vector $\vec y'$ from the statement of the lemma may contain
 more slack variables, but these will play no further role as their values will remain unchanged.)
Each $y_i$ belongs to the vector of all slack variables, $\vec y$, but $\vec y$ may well contain
further slack variables too.
By Lemma~\ref{l:slack-separate}, each $y_e$ appears in $\phi$ exactly once.

The nontrivial direction of the proof is left to right (``only if'').
In line with the statement of the lemma, we fix an assignment~$\xi$ that
assigns values to
free variables $\vec z$,
each (remaining) variable from~$\vec x \setminus x$, and
each slack variable from~$\varsofsub(s)$.
These values will remain fixed throughout the proof.

We assume that $\phi_k$ is satisfiable at the beginning of the iteration,
and accordingly we can additionally assign a value to the variable~$x$ as well as
to all slack variables that have not been yet assigned an integer value,
so that $\phi_k$ is satisfied.
We let $\nu$ be the assignment that extends $\xi$ accordingly.

All components of the vector $\nu(\vec y') = (\nu(y_1), \ldots, \nu(y_r))$ belong to $\mathbb N$.
Consider the auxiliary rational vector
\begin{equation*}
\vec q^{\nu} = (q_1, \ldots, q_r) \defeq \left(\frac{\nu(y_1)}{|a'_1|}, \ldots, \frac{\nu(y_r)}{|a'_r|}\right),
\end{equation*}
where $a'_e$ is the coefficient at $x$ in the equality in which the variable $y_e$ appears.
Suppose $\nu$ is one of the assignments which is an extension of~$\xi$,
which satisfies~$\phi_k$, and
for which the \emph{smallest}
component of $\vec q^{\nu}$ is minimal.
Assume without loss of generality that $q_1$ is this component.

Consider the non-deterministic branch of the algorithm that guesses in line~\ref{gauss:guess-equation}
the equality in which the variable $y_1$ is present.
We will now show that 
$q_1 < \fmod(\phi_k)$ or, in other words,
$\nu(y_1)< |a'_1| \cdot \fmod(\phi_k)$.
The number $\fmod(\phi_k)$ is the one appearing in the expression for
the right endpoint of the interval in line~\ref{gauss:guess-slack} of \GaussianQE,
before the formula manipulation at line~\ref{gauss:vig}.

\medskip

\emph{An informal aside.}
Recall that
the standard argument dating back to Presburger picks an interval to which the value assigned to $x$ belongs,
and moving this value so that it is close to an endpoint of this interval.
In the scenario where the interval is bounded (and not an infinite ray), 
the two endpoints of the interval correspond to two constraints that have the smallest slack.
(More precisely, the set of constraints can be partitioned into two --- think left and right --- so that,
among all the constraints on each side, the ``endpoint constraint'' has the smallest slack.)
This justifies our choice, above, of the assignment $\nu$ that minimises the slack in constraints
that involve~$x$.
In particular, if these constraints include an equation without slack, then this equation is chosen.
We remark that slack for different equations is measured \emph{relative} to the absolute value
of the coefficient at $x$: indeed, the inequalities that define the above-mentioned intervals for $x$
are obtained by first dividing each equation (and thus, intuitively, the corresponding slack variable)
by the coefficient of $x$ in it.

\medskip

More formally,
denote $a = a'_1$ and
assume for the sake of contradiction that $\nu(y_1) \ge |a| \cdot \fmod(\phi_k)$.
Let $m$ be the coefficient of the slack variable $y_1$ in the equality $a x + \tau = 0$.
Denote by $\nu'$ the assignment that agrees with $\nu$ on all variables except $x$ and $\vec y'$
as well as
on all slack variables eliminated previously,
and such that
\begin{align*}
\nu'(x) &= \nu(x) \pm m \cdot \fmod(\phi_k), &&\\
\nu'(y_1) &= \nu(y_1) \mp a \cdot \fmod(\phi_k) < \nu(y_1),\\
\text{and}\quad
\nu'(y_e) &=
\nu(y_e) \mp a'_e \cdot \fmod(\phi_k),
&&e>1,
\end{align*}
where
the signs are chosen depending on $a > 0$ or $a < 0$, so that the inequality constraining $\nu'(y_1)$ is satisfied.
(As in the definition of $\vec q^{\nu}$ above, we use $a'_e$ to denote the coefficient at $x$ in the equality
in which the variable $y_e$ is present.)
We are now going to show that $\nu'$ is a satisfying assignment for the formula $\phi_k$.

Observe that if $\nu$ satisfies all divisibility constraints of $\phi_k$, then so does $\nu'$.
Let us verify that
$\nu'$ respects the range of each slack variable.
Indeed, since $\nu(y_1) \ge |a| \cdot \fmod(\phi_k)$, we have $\nu'(y_1) \ge 0$.
We now consider the amount of slack in other constraints. Note that
\begin{align*}
\nu'(y_e) &=
\nu(y_e) \mp a'_e \cdot \fmod(\phi_k) &\text{(by the choice of $\nu'$)}\\& =
|a'_e| \cdot q_e
          \mp a'_e \cdot \fmod(\phi_k) &\text{(by the definition of $q_e$)}\\&\ge
|a'_e| \cdot (q_e 
          - \fmod(\phi_k))              &\text{(by cases)}\\
&\ge |a'_e| \cdot (q_1 - \fmod(\phi_k)) &\text{(since $q_1$ is the smallest component of $\vec q^{\nu}$)}\\
&\ge 0.                                      &\text{(by assumption)}
\end{align*}
Therefore, $\nu'(y_e) \ge 0$ for all $e = 1, \ldots, r$.

It remains to verify that the assignment $\nu'$ satisfies all equalities from $\phi_k$.
By Lemma~\ref{l:slack-separate}, there is only one equality that contains $y_1$.
For this equality, we have
\begin{equation*}
\nu'(a x + \tau) = \nu(a x + \tau) \pm a \cdot m \cdot \fmod(\phi_k) \mp m \cdot a \cdot \fmod(\phi_k) = 0,
\end{equation*}
so under the new assignment this equality remains satisfied.
Take any other equality involving $x$ from $\phi_k$, say $a'_2 x + \sigma = 0$
in which the unassigned slack variable is $y_2$.
By Lemma~\ref{l:vigorous-slack-same},
the coefficient at $y_2$ in this equality is equal to~$m$,
the coefficient at $y_1$ in $a x + \tau = 0$.
We have
\begin{equation*}
\nu'(a'_2 x + \sigma) = \nu(a'_2 x + \sigma) \pm a'_2 \cdot m \cdot \fmod(\phi_k) \mp m \cdot a'_2 \cdot \fmod(\phi_k) = 0.
\end{equation*}
Thus, $\nu'$ is also a satisfying assignment for the formula $\phi_k$.
However, $\nu'(y_1) < \nu(y_1)$ by the choice of $\nu'$, and therefore
$\nu'(y_1) / |a| < \nu(y_1) / |a| = q_1$. This inequality contradicts
our choice of $\nu$, because we assumed that the smallest component of the vector $\vec q^{\nu}$ is minimal.
Thus, we conclude that the inequality $q_1 < \fmod(\phi_k)$ must hold,
that is, 
$\nu(y_1)< |a'_1| \cdot \fmod(\phi_k)$.
This means that the range specified in line~\ref{gauss:guess-slack} suffices to keep the formula
satisfiable after at least one of the possible substitutions.

It remains to prove the other direction (right to left, ``if'').
A satisfying assignment to the formula $\phi_{k+1}\sub{v}{y}$ can,
by Lemma~\ref{l:no-shrink}, always be extended to~$x$ so that
it also satisfies~$\phi_k$. This completes the proof.
\end{proof}

We turn our attention to lines~\ref{gauss:unslackloop}--\ref{gauss:drop-slack}
of \GaussianQE.
We need to prove that, when the slack variables are removed
in line~\ref{gauss:drop-slack}, there is no need to keep
the divisibility constraints on the slack.
More precisely, suppose that, when the algorithm reaches line~\ref{gauss:unslackloop}, $\phi$ contains an equality with
a slack variable~$y$ which is not assigned any value by substitutions from~$s$.
Let us assume that the coefficient at $y$ is positive; the negative case is analogous.
Thus, the equality has the form $\rho + g y = 0$, $g > 0$.
The range of $y$ is $\N$, and by Lemma~\ref{l:slack-separate} this variable occurs
in no other constraint of $\phi$; therefore, this constraint can be replaced
with a conjunction $(\rho \le 0) \land (g \divides \rho)$.
Line~\ref{gauss:drop-slack} only introduces the inequality $\rho \le 0$,
omitting the divisibility constraint $g \divides \rho$.

The following lemma shows that this divisibility
is \emph{implied} by other constraints and can thus, indeed, be removed safely.
(In practice, it might be beneficial to keep and use the constraint.)

\begin{lemma}
\label{l:dropping-slack}
Denote by $\phi'$ the formula obtained at the end of the first \textbf{foreach} loop
(in lines~\ref{gauss:mainloop}--\ref{gauss:endmainloop}),
   and by $\psi'$ the one                     after the second \textbf{foreach} loop
(in lines~\ref{gauss:unslackloop}--\ref{gauss:drop-slack}), which
removes variables~$\vec y'$.
Then every assignment $\nu$ that satisfies $\psi'$ can be extended to $\vec y'$
so that the resulting assignment $\nu'$ has $\nu'(y) \in \N$ for all $y$ in $\vec y'$ and
moreover $\nu'$ satisfies $\phi'$.
In other words: 
\[ 
  \exists \vec y' \,(\vec y' \geq \vec 0 \land \phi') \iff \psi'.
\]
\end{lemma}

\begin{proof}
Take an assignment $\nu$ that satisfies the assumptions of the lemma
(in particular, $\nu$ satisfies $\psi'$).
	Take an inequality $\eta\sub{0}{y} \le 0$ that is introduced by line~\ref{gauss:drop-slack}
of the algorithm.
(The case $\eta\sub{0}{y} \ge 0$ is analogous and we skip it.)
Suppose this inequality originates from a constraint $\eta = 0$ picked by
line~\ref{gauss:unslackloop}, where $y$ is the one slack variable in $\eta$
that is not assigned a value by substitutions of the sequence~$s$.

Let $k$ be the number of iterations of the first \textbf{foreach} loop.
We need the following auxiliary notation.
Let $x_1, \ldots, x_k$ be the variables picked by the header of the loop in line~\ref{gauss:mainloop}.
Assume without loss of generality that, for each $i$, at the beginning of iteration~$i$ the
variable $x_i$ is present in at least one of the equalities in~$\phi$ (i.e., we consider $k$ nontrivial iterations).
Suppose the constraint in question, $\eta$, arose from the sequence of transformations
depicted below, where each $\eta_i$ is obtained after $i$~iterations,
and in particular $\eta = \eta_k$:
\begin{equation*}
\begin{array}{lccccccccc}
\text{Equality:}&\eta_0=0&\rightarrow&\eta_1=0&\rightarrow&\cdots&\rightarrow&\eta_{k-1}=0&\rightarrow&\eta_k=0 \\
\text{Term without $y$:}&\eta'_0&&\eta'_1&&\cdots&&\eta'_{k-1}&&\eta'_k \\
\text{Slack:}&1 \cdot y && a_1 \cdot y && \cdots && a_{k-1} \cdot y && a_k \cdot y
\end{array}
\end{equation*}
Each term $\eta'_i$ is obtained from $\eta_i$ by dropping the slack variable:
$\eta'_i = \eta_i \sub{0}{y}$.
The slack row shows the coefficients of the variable~$y$ in $\eta_0, \eta_1, \ldots, \eta_k$;
see Lemma~\ref{l:vigorous-slack-same}.
Here we assume that, on iteration~$i$, the equality picked in line~\ref{gauss:guess-equation}
is $a_i x_i + \tau_i = 0$.

Our goal is to prove that $\nu(\eta'_k)$ is divisible by $a_k$.
(This enables us to extend $\nu$ to $y$ by assigning
 $\nu(y) \eqdef - \nu(\eta'_k) / a_k$.)
Denote by $p_i$ the divisor in line~\ref{gauss:divide} in iteration~$i$;
then $p_i = a_{i-1}$ for all $i \ge 2$ and $p_1 = 1$.
For each $i = k, \ldots, 1$ we will show that
\begin{equation*}
\nu(\eta'_{i-1}) = \frac{p_i}{a_i} \cdot \nu(\eta'_i).
\end{equation*}
It will then follow that $\nu(\eta'_0) = \nu(\eta'_k) \cdot \prod_{i=1}^{k} p_i / \prod_{i=1}^{k} a_i = \nu(\eta'_k) / a_k$.
Since $\nu(\eta'_0) \in \Z$, this will conclude the proof.

To justify the equation
$\nu(\eta'_{i-1}) = p_i / a_i \cdot \nu(\eta'_i)$,
notice that, from the pseudocode of the algorithm, we obtain
$\nu(a_i x_i + \tau_i) = 0$.
We now develop this intuition into a formal argument.
Consider the $i$th iteration of the first \textbf{foreach} loop, and in particular
lines~\ref{gauss:vig}--\ref{gauss:restore}.
Suppose that, just before line~\ref{gauss:vig}, the constraint $\eta_{i-1} = 0$
has the form $b_i x_i + \hat\eta_{i-1} + m_{i-1} y = 0$ for some $m_{i-1} \in \Z$,
so that $\eta'_{i-1} = b_i x_i + \hat\eta_{i-1}$ for some term $\hat\eta_{i-1}$.
The result of applying substitution~$\vigsub{\frac{-\tau_i}{a_i}}{x_i}$ to $\eta'_{i-1}$
in line~\ref{gauss:vig} is $ - b_i \tau_i + a_i \hat\eta_{i-1} $.
Since $\eta'_{i-1} = \eta_{i-1} \sub{0}{y}$, and
the variable $y$ does not occur in the term~$\tau_i$ thanks to Lemma~\ref{l:slack-separate},
we can apply Lemma~\ref{l:fundamental}, part~\ref{no-remainder}, concluding that
all coefficients in the term $ - b_i \tau_i + a_i \hat\eta_{i-1} $ are divisible by $p_i = a_{i-1}$.
(Notice that Lemma~\ref{l:fundamental} does not involve any assignments:
 it applies directly to the syntactic objects that~\GaussianQE works with.)

Thus, we can, in fact, write $\eta'_i = ( - b_i \tau_i + a_i \hat\eta_{i-1} ) / p_i$,
and moreover $\nu(\eta'_i) \in \Z$.
Notice that we can assume $\nu(a_i x_i + \tau_i) = 0$, because even though
the formula $\psi'$ (satisfied by $\nu$) does not contain the equality $a_i x_i + \tau_i = 0$,
it contains the divisibility constraint $a_i \divides \tau$, and no occurrences of $x_i$,
so we may as well stipulate $\nu(a_i x_i + \tau_i) = 0$.
Therefore,
\begin{equation*}
\nu(\eta'_i) = \frac{a_i b_i \nu(x_i) + a_i \nu(\hat\eta_{i-1})}{p_i} \in \Z,
\end{equation*}
and hence
$
a_i \nu(\hat\eta_{i-1}) = p_i \nu(\eta'_i) - a_i b_i \nu(x_i)
$.
Since $a_i \ne 0$, the number $p_i \nu(\eta'_i)$ must be divisible by~$a_i$,
and moreover
\begin{equation*}
\nu(\hat\eta_{i-1}) = \frac{p_i}{a_i} \cdot \nu(\eta'_i) - b_i \nu(x_i).
\end{equation*}
Recalling that $\eta'_{i-1} = b_i x_i + \hat\eta_{i-1}$, we conclude that
\begin{equation*}
\nu(\eta'_{i-1}) = \nu(b_i x_i) + \nu(\hat\eta_{i-1}) = b_i \nu(x_i) + \frac{p_i}{a_i} \cdot \nu(\eta'_i) - b_i \nu(x_i) =
                                                                       \frac{p_i}{a_i} \cdot \nu(\eta'_i).
\end{equation*}
This completes the proof.
\end{proof}

\begin{lemma}\label{lemma:corr-gaussianqe}
    \Cref{algo:gaussianqe} (\GaussianQE) complies with its specification.
\end{lemma}

\begin{proof}
Consider an input system $\phi(\vec x, \vec z)$, with $\vec x$ being the set of variables to be eliminated. 

We combine the auxiliary results obtained earlier in this section.
The main argument shows that, for a given assignment to free variables~$\vec z$,
if the input system~$\phi$ is satisfiable,
then at least one non-deterministic branch~$\beta$
produces a satisfiable output $\psi_\beta$.
\begin{enumerate}
\item
The guessing in line~\ref{gauss:guess-slack} restricts the choice to a finite set.
This `amount of slack' is shown sufficient in Lemma~\ref{l:amount-of-slack}.
Thus, the guesses performed in the main (first) \textbf{foreach} loop of the algorithm are sufficient.
\item
Lemma~\ref{l:dropping-slack} handles the removal of remove the remaining slack variables,
showing that the second \textbf{foreach} loop correctly recasts the equalities that still contain slack variables into inequalities. 
\item
The algorithm also contains a third \textbf{foreach} loop (lines~\ref{gauss:loop-rem}--\ref{gauss:subst-rem}). When this \textbf{foreach} loop is reached, the variables from $\vec x$ that still appear in the formula $\phi$ do so in divisibility constraints only. Thus, assigning to these variables values in the interval $[0,\fmod(\phi)-1]$, as done in line~\ref{gauss:guess-rem}, suffices.
\end{enumerate}
We now formalise the sketch given above.
As above, we write $\phi_0$ for the system obtained from $\phi$ after executing line~\ref{gauss:introduce-slack}.
We denote by $\mathcal B$ a set of indices for the possible non-deterministic branches of the algorithm. In particular, to each $\beta \in \mathcal B$ corresponds a sequence of guesses done in lines~\ref{gauss:guess-equation}--\ref{gauss:append-seq}.
We write $\phi_\beta(\vec x_\beta,\vec z)$ and $s_\beta$ for the system and for the sequence of substitutions obtained in the non-deterministic branch of $\beta$, when the algorithm completes the main \textbf{foreach} loop.
Here, $\vec x_\beta$ are the variables from $\vec x$ that are removed in the third \textbf{foreach} loop; this means that, in $\phi_\beta$, these variables only occur in divisibility constraints. 
Similarly, we write $\psi_\beta'$ for the system obtained when the non-deterministic computation completes the second \textbf{foreach} loop. 
As in the specification, we denote by $\psi_\beta$ the output of the branch.
Lastly, $\vec y_{\beta} \coloneqq \vec y \setminus \varsofsub(s_\beta)$, 
where $\vec y$ is the set of slack variables introduced in 
line~\ref{gauss:introduce-slack}.
We have 
\begin{align*}
  \exists \vec x\, \phi 
  &\iff \exists \vec x\, \exists \vec y\, ( \vec y \geq \vec 0 \land \phi_0)
  &\text{(introduction of slack variables)}\\
  &\iff \bigvee_{\beta \in \mathcal B} \exists \vec x_\beta \, \exists \vec y_\beta\, (\vec y_\beta \geq \vec 0 \land \phi_\beta)s_\beta
  &\text{(by Lemma~\ref{l:amount-of-slack})}\\
  &\iff \bigvee_{\beta \in \mathcal B} \exists \vec x_\beta \, (\psi_\beta's_\beta)
  &\text{(by Lemma~\ref{l:dropping-slack})}\\
  &\iff \bigvee_{\beta \in \mathcal B} \psi_\beta
  &\text{(lines~\ref{gauss:apply-seq}--\ref{gauss:subst-rem}).}
  & \qedhere
\end{align*}
\end{proof}

\subsection{Analysis of complexity of Algorithm \ref{algo:gaussianqe} (\GaussianQE)}
\label{app:gaussian-elimination:complexity}

The following Lemma~%
\ref{lemma:complexity-gaussianqe}
is not required for the proof of
Theorem~\ref{thm:gaussianQE-in-np},
but rather
for a more fine-grained analysis in the subsequent proof that
\Cref{algo:master} (\Master) runs in non-deterministic polynomial time
(\Cref{prop:master-in-np}).

\begin{lemma}\label{lemma:complexity-gaussianqe}
  Consider a linear system $\phi(\vec x, \vec z)$ having $n \geq 1$ variables.
  Let $\alpha \in \N$ be the largest absolute value of coefficients of variables from $\vec x$
  in equalities and inequalities of $\phi$. 
  Let $\psi$ be the output of~\Cref{algo:gaussianqe} on input $(\vec x, \phi)$. 
  Then, $\onenorm{\psi} \leq \tocheck{(\onenorm{\phi}+2)^{4 \cdot (n+1)^2} \cdot \fmod(\phi)}$, $\card\psi \leq \card\phi$, 
  and $\fmod(\psi)$ divides $c \cdot \fmod(\phi)$ for some positive $c \leq \tocheck{(\alpha+2)^{(n+1)^2}}$.
\end{lemma}

\begin{proof}
The inequality $\card\psi \leq \card\phi$ follows directly from the description of the algorithm:
new constraints are only introduced at line~\ref{gauss:restore}, which means that they replace
equalities that are eliminated previously by substitutions at line~\ref{gauss:vig}.
As previously, we use the observation that
$(a x + \tau = 0) \vigsub{\frac{-\tau}{a}}{x} = (- a \tau + a \tau = 0)$,
which simplifies to true.

To prove that $\fmod(\psi)$ divides $c \cdot \fmod(\phi)$ for some positive $c \leq \tocheck{(\alpha+2)^{(n+1)^2}}$, we
apply Lemma~\ref{l:divisors}:
the least common multiple of all moduli in~$\phi$
gets multiplied (during the course of the procedure)
by $\mu_k$, or possibly by a divisor of this integer.
Here $k$ is the total number of nontrivial iterations.

As the next step, we show that $|\mu_k| \le (\alpha + 2)^{(n+1)^2}$.
Indeed, the square submatrix whose determinant is $\mu_k$ is formed by
the first $k$ rows and first $k$ columns of the matrix $B_k$ from Equation~\eqref{eq:systems-and-matrices}.
The columns, in particular, correspond to the variables from~$\vec x$
eliminated by the procedure; so all entries of the submatrix are coefficients
of some $k$~variables from $\vec x$ in equalities and inequalities of~$\phi$.
The absolute value of each entry is at most~$\alpha$.
Therefore,
\begin{equation}
\label{eq:std-det}
|\mu_k| \le k! \cdot \alpha^k \le n! \cdot \alpha^n \le
n^n \alpha^n \le 2^{n^2} \alpha^n \le (\alpha + 2)^{(n+1)^2},
\end{equation}
and so indeed
$\fmod(\psi)$ divides $c \cdot \fmod(\phi)$ for some $c \le (\alpha + 2)^{(n+1)^2}$.

It remains to prove the upper bound on the 1-norm of the output, namely
  $\onenorm{\psi} \leq \tocheck{(\onenorm{\phi}+2)^{4 \cdot (n+1)^2} \cdot \fmod(\phi)}$.
Let $C \eqdef (\onenorm{\phi} + 2)^{(n+1)^2}$.
Take an arbitrary equality $\eta = 0$ present in the output system~$\psi$.
The following quantities contribute to the 1-norm of $\eta$:
\begin{itemize}
\item
The coefficients of all variables in~$\eta$.
Observe that, since the equality~$\eta = 0$ is present in~$\psi$, it was never
picked at line~\ref{gauss:guess-equation} of the procedure.
Thus, for each \emph{non-zero} coefficient of variables in $\eta$
we can apply Lemma~\ref{l:fundamental}, part~\ref{det-property}:
this coefficient is equal to a $(k+1) \times (k+1)$ minor of the matrix~$B_0$.
The rows for the corresponding square submatrix are associated with equalities
and inequalities of the original input system~$\phi$,
and the columns with the~$k$ eliminated variables and the one variable, say~$w$,
that we are considering.
This variable $w$ must belong to~$\vec x$ or~$\vec z$, because all slack variables
are removed by lines~\ref{gauss:unslackloop}--\ref{gauss:drop-slack}.
Thus, $k + 1 \le n$ and all entries of this submatrix have absolute
value at most $\alpha$.
Therefore, by a calculation similar to Equation~\eqref{eq:std-det},
the coefficient at~$w$ is at most
$(\alpha + 2)^{(n+1)^2} \le C$.
\item
The constant term of the equality $\eta' = 0$
that is present in the system just before line~\ref{gauss:apply-seq}
and later rewritten into $\eta = 0$.
To estimate this term, we can again use Lemma~\ref{l:fundamental}\ref{det-property}.
The columns of the square submatrix are now associated with
the~$k$ eliminated variables and with the constant terms.
It follows that all entries of the submatrix have absolute
value at most $\onenorm{\phi}$, so, in analogy with the chained
inequality~\eqref{eq:std-det}, we obtain the upper bound of
$C$ for this term.
\item
Values of slack variables and their coefficients.
Line~\ref{gauss:apply-seq} updates the constant terms in equalities;
let us estimate the effect.
Firstly, the coefficients of slack variables in these equalities
are, again by Lemma~\ref{l:fundamental}\ref{det-property}, minors of the matrix~$B_0$.
The final column of the $(k+1) \times (k+1)$ submatrix
is now the column of coefficients of the slack variable in question.
Clearly, this column must be a $0$--$1$ vector with at most one `$1$',
so the upper bound of
$(\alpha + 2)^{(n+1)^2}$ from Equation~\eqref{eq:std-det} remains valid.
To estimate the values assigned to slack variables, consider
line~\ref{gauss:guess-slack} of the procedure.
By Lemma~\ref{l:divisors}, for each $i \ge 0$, the system~$\phi_i$
obtained after $i$~nontrivial steps of the procedure
satisfies $\fmod(\phi_i) \divides c \cdot \fmod(\phi)$, with $c \le C$.
By Lemma~\ref{l:fundamental}\ref{lead}, the lead coefficients in the iterations
are $\mu_1$, \ldots, $\mu_k$, and they are also at most~$c$.
Therefore, no value assigned to a slack variable can exceed
$c^2 \cdot \fmod(\phi)$, and the contribution of one slack variable with its coefficient
is at most $c^3 \cdot \fmod(\phi) \le C^3 \cdot \fmod(\phi)$.
\end{itemize}
Let us now put these contributions together. Observe that there are
at most $n$ variables in $\vec x$ in $\vec z$ combined; and in fact only~$k \le n$ slack variables
may have been assigned values by~$s$: all other slack variables are removed
by lines~\ref{gauss:unslackloop}--\ref{gauss:drop-slack}.
So
\begin{equation*}
\onenorm{\phi} \le
n \cdot C +
C +
n \cdot C^3 \cdot \fmod(\phi)
\le
(n + 1) \cdot C^3 \cdot \fmod(\phi)
\le
C^4 \cdot \fmod(\phi),
\end{equation*}
because $n + 1 \le 2^{(n + 1)^2} \le (\onenorm{\phi} + 2)^{(n+1)^2} = C$ for all $n \ge 1$.
This completes the proof.
\end{proof}

\subsection{Proof of Theorem~\ref{thm:gaussianQE-in-np}}

The correctness of~\GaussianQE
is provided by Lemma~\ref{lemma:corr-gaussianqe}.
All the integers appearing in the run have polynomial bit size
by
Lemma~\ref{l:fundamental}
and by Lemma~\ref{l:cramer} (parts~\ref{l:cramer:in-between} and~\ref{l:cramer:mid-iteration})
combined with Lemma~\ref{l:no-remainder}.
Therefore, the running time is also polynomial in the bit size of the input.

\section{Proofs from Section~\ref{sec:procedure-details}: correctness of Algorithm~\ref{algo:master} (\Master)}
\label{app:procedure-correctness}

This appendix provides a proof of~\Cref{prop:master-correct}. 
This is done by first showing the correctness of~\Cref{algo:linearize} (\SolvePrimitive), 
then the correctness of~\Cref{algo:elimmaxvar} (\ElimMaxVar), and lastly the correctness of~\Cref{algo:master} (\Master).
The interaction between \Cref{algo:master,algo:elimmaxvar,algo:linearize} has been described 
in~\Cref{sec:procedure,sec:procedure-details}. 
The flowchart in~\Cref{fig: flowchart} is provided to remind the reader of the key steps of these algorithms.


\tikzset{
    invisible/.style={opacity=0},
    emph/.style={color=magenta},
    alert/.style={color=red},
    anotherhl/.style={color=newblue},
    bold/.style={very thick},
    emph on/.style={alt={#1{emph}{}}},
    alert on/.style={alt={#1{alert}{}}},
    anotherhl on/.style={alt={#1{anotherhl}{}}},
    arrow on/.style={alt={#1{->}{}}},
    bold on/.style={alt={#1{bold}{}}},
    visible on/.style={alt={#1{}{invisible}}},
    }

\pgfdeclaredecoration{dashsoliddouble}{initial}{
  \state{initial}[width=\pgfdecoratedinputsegmentlength]{
    \pgfmathsetlengthmacro\lw{.5pt+.5\pgflinewidth}
    \begin{pgfscope}
      \pgfpathmoveto{\pgfpoint{0pt}{\lw}}%
      \pgfpathlineto{\pgfpoint{\pgfdecoratedinputsegmentlength}{\lw}}%
      \pgfmathtruncatemacro\dashnum{%
        round((\pgfdecoratedinputsegmentlength-3pt)/6pt)
      }
      \pgfmathsetmacro\dashscale{%
        \pgfdecoratedinputsegmentlength/(\dashnum*6pt + 3pt)
      }
      \pgfmathsetlengthmacro\dashunit{3pt*\dashscale}
      \pgfsetdash{{\dashunit}{\dashunit}}{0pt}
      \pgfusepath{stroke}
      \pgfsetdash{}{0pt}
      \pgfpathmoveto{\pgfpoint{0pt}{-\lw}}%
      \pgfpathlineto{\pgfpoint{\pgfdecoratedinputsegmentlength}{-\lw}}%
      \pgfusepath{stroke}
    \end{pgfscope}
  }
}

\begin{figure}
    \caption{Flowchart of~\Cref{algo:master,algo:elimmaxvar}.}\label{fig: flowchart}
    \begin{tikzpicture}[
      >=stealth,
      a/.style={align=left, node font=\itshape},
      s/.style={align=left},
      es/.style={align=right, anchor=north east},
      f/.style={draw=gray,fill=white},
      new/.style={draw=newblue,text=newblue},
      fornew/.style={draw=newblue},
      node distance=1.05cm and 0.6cm,
      aster/.style={circle,draw=black,fill=white,inner sep=1pt}
    ]

      {[on background layer]
        \fill[lipicsLightGray] (-5.85,-6.53) rectangle ++(15,-2.41);
        \fill[black!28] (-5.85,-9) rectangle ++(15,-3.71);
        \fill[lipicsLightGray!25] (-5.85,-12.77) rectangle ++(15,-5.53);
      }

      \node [f] (input) {$\Phi(\vec x)$};
      \node [a,right=of input,  xshift=-1.5em] {$\Phi$ linear-exponential system}; 
      \node [f, below= 1.5 of input] (input-after-ordering) {$\theta(\vec x) \land \Phi(\vec x, \vec u)$};
      \node [a, right=of input-after-ordering, xshift=-1.5em, align=left] (input-annot) 
        { $x$ largest and $y$ second largest variables in $\theta$;\\[-2pt]
          variables in $\vec u$ only occur linearly
        };
      \node[draw=black, align=center, diamond, aspect=2.5,fill=white] (les-return) 
        [below = 0.5cm of input-after-ordering] {$\vec x$ empty?};
      
      \draw [->] (input) -- node[right,align=left]{
        $\vec u$ empty vector of auxiliary variables\\[-2pt]
        guess ordering $\theta$ on $\vec x$} (input-after-ordering);
      
      \draw [->] (input-after-ordering) to (les-return);
      \draw (les-return.east) edge[-implies,double equal sign distance] node[a,above]{yes} ++ (2.4cm,0cm);
      \node[a] (les-yes) [right = 2.5cm of les-return] {\textbf{return}  $\,\Phi(\vec 0)$};
  
      \node [a] (lexp-comment) [below= 1 of les-return] {\small{quotient system}};
      \node [f] (lexp) [below=of lexp-comment, yshift=7ex] {$\varphi(\vec x, \xquot, \quot, \xrema, \rema) \land \xrema < 2^y \land \rema < 2^y$};
      \node [f] (delayed) [right=of lexp] {$x = \xquot \cdot 2^y + \xrema$};
      \node [a] (delayed-comment) [above=of delayed, yshift=-7ex] {\small{delayed substitution}};
      \node [f] (emv-order) [right=of delayed, fornew] {$\theta(\vec x)$};
      \node at ($(emv-order.west)!0.5!(delayed.east)$) {$\land$};
      \node at ($(delayed.west)!0.5!(lexp.east)$) {$\land$};
      {[on background layer]
        \node [f] (emv-in) [fit=(emv-order) (lexp) (lexp-comment) (delayed)] {};
      }
  
      \draw (les-return) edge[-implies,double equal sign distance] node [right, s] {$\auxi = \quot \cdot 2^y + \rema$} (lexp|-emv-in.north);
  
      \node (emv-split) [below= 1.2cm of lexp] {$\land$};
      \node [f] (msb) [left=of emv-split, xshift=0.55cm] {$\gamma_1(2^{x-y}, \xquot, \quot)$};
      \node [f] (lsb) [right=of emv-split, xshift=-0.55cm, fornew] {$\psi_1(\vec x \setminus x, \xrema, \rema)$};
      {[on background layer]
        \node [f] (emv-in-two) [fit=(msb) (lsb)] {};
      }
  
      \draw [->] (lexp) -- node[a,right] {split \textup{(see~\large$\ast$)}} (emv-in-two);
  
      \node [f] (gauss1-output) [below= 1.15cm of msb] {$\gamma_2(2^{x-y}, \xquot)$};
      \draw [->] (msb) edge node [right, s] (gauss-qe) {\textsc{GaussQE}: eliminate $\quot$} (gauss1-output);
      \node (monadic-split) [below=of gauss1-output] {$\land$};
      \node [f] (monadic-target) [right=of monadic-split, xshift=-0.55cm] {$\chi(x-y)$};
      \node [f] (monadic-other) [left=of monadic-split, xshift=0.55cm] {$\gamma_3(\xquot)$};
      \draw [->] (gauss1-output) edge node [right, s] (solve-primitive) {\textsc{SolvePrimitive}} (monadic-split);
  
      \node [a,right = 0pt of gauss1-output] {$x-y \geq \xquot$}; 
  
      \node [f] (apply-delayed) [below right = 1cm of monadic-target] {$\chi(\xquot \cdot 2^y + \xrema - y)$};
      \node (apply-delayed-split) [below=of apply-delayed] {$\land$};
      \node [f] (second-split-msb) [left=of apply-delayed-split, xshift=0.55cm] {$\gamma_4(\xquot)$};
      \node [f] (second-split-lsb) [right=of apply-delayed-split, xshift=-0.55cm, fornew] {$\psi_2(y, \xrema)$};
      {[on background layer]
        \node [f] (emv-in-three) [fit=(second-split-msb) (second-split-lsb)] {};
      }
  
      \draw [->, rounded corners] (monadic-target.east) -- ++ (5em,0em) coordinate (turn) -- (turn|-apply-delayed.north);
      \draw [->, rounded corners] (delayed.south) -- ($(delayed.south|-turn)  + (0em,-1.4em)$)
        -- ($(turn)  + (1pt,-1.4em)$);
      \draw [->] (apply-delayed) -- node[a,right] {split \textup{(see~\large$\ast$)}}  ($(apply-delayed-split|-emv-in-three.north)$);
      \node [a, align=center] (subs-text) [below right = -0.3cm and 2pt of turn] {\small{apply}\\[-2pt]\small{substitution}};
  
      \node [draw=black,circle,fill=white] (pregauss2) [below left= 0.5cm and 0.8cm of second-split-msb] {$\land$};
      \node[draw=black, align=center, diamond, aspect=2.5,fill=white] (gauss2) 
        [below = 1cm of pregauss2] {$\gamma_3 \land \gamma_4$ satisfiable?};
  
      \draw [->, rounded corners] (second-split-msb.south) -- (second-split-msb.south|-pregauss2.east) -- (pregauss2.east);
      \draw [->, rounded corners] (monadic-other.south) -- (monadic-other.south|-pregauss2.west) -- (pregauss2.west);
      \draw [->] (pregauss2) edge node [right, s] (gauss-qe2) {\textsc{GaussQE}: eliminate $\xquot$} (gauss2);
  
      \draw (gauss2.east) edge[-implies,double equal sign distance] node[a,above]{yes} ++ (2.4cm,0cm);
  
      \node[a] (gauss-no) [below = 1cm of gauss2] {\textbf{return} $\bot$};
      \draw (gauss2.south) edge[-implies,double equal sign distance] node[a,right]{no} (gauss-no);
  
  
      \node [draw=newblue, circle, fill=white] (blueand) [right = 3.23cm of gauss2] {\color{newblue}$\land$};
  
      \node [new, align=left, fill=white] (next) [below = 0.7cm of blueand]
            {
        New linear-exponential system:\\
        exponential variables $\vec x \setminus x$\,;\\
        auxiliary variables $\rema, \xrema$\,;\\
        ordering equivalent to $\exists x . \theta(\vec x)$
            };
  
      \draw [->, rounded corners, new] (lsb.east) -- (blueand.north|-lsb.east) -- (blueand.north);
      \draw [->, rounded corners, new] (emv-order.south) -- (emv-order.south|-blueand.east) -- (blueand.east);
      \draw [->, rounded corners, new] (
        second-split-lsb.east) -- ++ (0.4cm,0em) coordinate (turn2) -- (turn2|-blueand.west) -- (blueand.west);
      \draw [->, rounded corners, new] (blueand.south) -- node[right, align=center] {exclude $x$ from $\theta$} (next);  
  
      \draw [->, rounded corners, new] (next.east) -- ++ (4em,0cm) coordinate (turn3) -- (turn3|-input-after-ordering) coordinate (turn4) -- ($(turn4) + (-2.5em,0cm)$);
  
      \node [es] (linexpsat-text) [left = 2.5cm of input] {\textsc{LinExpSat}};
      \node [es] (elimmaxvar-text) [below = 5.1cm of linexpsat-text] {\textsc{ElimMaxVar}};
  
      \node [es] (stepone) [below = 0.9cm of elimmaxvar-text] {Step (i)};
      \node [es] (steptwo) [below = 1.9cm of stepone] {Step (ii)};
      \node [es] (stepthree) [below = 3.8cm of steptwo] {Step (iii)};
  
      \node (stepone-start) [above right = 0.2cm and 0.2cm of stepone] {};
      \draw[|-|] (stepone-start.south) -- ++ (0cm,-2.4cm) coordinate (steptwo-start);
      \draw[|-|] ($(steptwo-start.south)+(0pt,-2pt)$) -- ++ (0cm,-3.7cm) coordinate (stepthree-start);
      \draw[|-|] ($(stepthree-start.south)+(0pt,-2pt)$) -- ++ (0cm,-5.5cm);

      \draw [decoration={dashsoliddouble}, decorate]
                            (emv-in.west)++(0ex,3.5ex) -- ($(emv-in.west) + (-8.5em,3.5ex)$);
      \draw [decoration={dashsoliddouble}, decorate]
                            ($(emv-in.east) + (4.7em,3.5ex)$) -- 
                            ($(emv-in.east) + (0ex,3.5ex)$);
      
  
    \end{tikzpicture}
    
    \vspace{1em}
    \footnotesize
    \begin{itemize}
    \item[\raisebox{-1pt}{\large$\ast$\,}:]
        In a nutshell, assuming $\xrema < 2^y$, $\vec r < 2^y$ and $z < 2^y$ for all $z \in \vec x \setminus x$,
        we can rewrite an inequality 
        \mbox{${a \cdot 2^x + f(\xquot,\vec q) \cdot 2^y + \rho(\vec x \setminus x, \xrema, \vec r) \leq 0}$ 
        into  ${\bigvee_{r = -\norm{\rho}_1}^{\norm{\rho}_1} (\underbrace{a \cdot 2^{x-y} + f(\xquot,\vec q) + r \leq 0}_{\textit{part of $\gamma_1 / \gamma_4$}} 
        \land \color{newblue}\underbrace{\color{black}{(r-1) \cdot 2^y < \rho \leq r \cdot 2^y)}}_{\textit{part of $\psi_1/\psi_2$}}}$.}
        For equalities, strict inequalities and divisibility constraints, 
        see~\Cref{lemma:split:congruences,lemma:split:inequalities}.
    \end{itemize}
\end{figure}


\subparagraph*{Non-deterministic branches.} In this appendix, for the convenience of exposition it is sometimes useful to fix a representation for the non-deterministic branches of our procedures. As done throughout the body of the paper, we write $\beta$ for a non-deterministic branch. 
It is represented as its \emph{execution trace}, that is a list of entries (tuples) following the control flow of the program. Each entry contains the name of the algorithm being executed, the line number that is currently being executed, and, for lines featuring the non-deterministic \textbf{guess} statement, the performed guess. 
As an example, to line~\ref{algo:master:guess-order} of~\Master 
correspond entries of the form $(\Master,\text{``line 2''},\theta)$, 
where $\theta$ is the ordering that has been guessed.

We write $\beta = \beta_1\beta_2$ whenever $\beta$ can be decomposed on a prefix $\beta_1$ and suffix $\beta_2$.

\subparagraph*{Global outputs.}
Our proofs often consider the \textbf{global outputs}
of~\Cref{algo:elimmaxvar,algo:linearize}. This is defined as the set containing
all outputs of all branches~$\beta$ of a procedure. For~\Cref{algo:elimmaxvar},
the \textbf{global output} is a set of linear-exponential systems.
For~\Cref{algo:linearize}, this is a set of pairs of linear-exponential systems. 
Observe that the \textbf{ensuring} part of the specification of these algorithms provides properties of 
their \textbf{global outputs}. For instance, in the case of~\Cref{algo:elimmaxvar}, 
the \textbf{ensuring} part specifies some properties of the disjunction $\bigvee_{\beta} \psi_\beta$, 
which ranges over all branches of the procedure, 
and the \textbf{global output} is just the set of all disjuncts~$\psi_\beta$ appearing in that formula.

In~\Cref{app:procedure-correctness-elimmaxvar} we also define \textbf{global outputs} for each of the three steps of~\Cref{algo:elimmaxvar}.
For Steps~(\ref{item:elimmaxvar-i}) and~(\ref{item:elimmaxvar-ii}), 
they will be sets of pairs of specific linear-exponential systems (similarly to~\Cref{algo:linearize}). 
Step~(\ref{item:elimmaxvar-iii}) will have a \textbf{global output} of the same type as~\Cref{algo:elimmaxvar}.  

\subsection{Correctness of Algorithm~\ref{algo:linearize} (\SolvePrimitive)}  

\begin{lemma}\label{lemma:corr-linearize}
    \Cref{algo:linearize} (\SolvePrimitive) complies with its specification. 
\end{lemma}
\begin{proof}
    Recall that $(u,v)$-primitive systems are composed of equalities, (non)strict inequalities, and divisibilities of the form $(a\cdot 2^u+b\cdot v+ c\sim0)$,
    where $u,v$ are non-negative integer variables, the predicate symbol $\sim$ is from the set $\bigl\{=,<,\leq,\equiv_d : d \geq 1\bigr\}$, 
    and $a,b,c$ are integers.

    Consider an input $(u,v)$-primitive system $\varphi$.
    \Cref{algo:linearize} first splits~$\phi$ into a conjunction of two subsystems $\chi \land \gamma$ (line~\ref{line:linearize-decompose}) such that
    \begin{itemize}
      \item the system $\chi$ comprises 
      all (in)equalities $(a\cdot 2^u+b\cdot v+ c\sim0)$ with $a\neq0$;
      \item the system $\gamma$ is composed of all other constraints of $\varphi$.
    \end{itemize}
    In the sequel, most of our attention is given to the formula~$\chi$. 
    
    As explained in the body of~the paper, the integer constant~$C$ defined in line~\ref{line:linearize-max-constant} plays 
    a key role in the correctness proof. 
    The following equivalence is immediate: 
    \begin{equation}
      \label{eq:proof:linerise:one} 
      \phi \iff \Big(\bigvee_{c = 0}^{C-1} (u = c) \land \phi\sub{2^c}{2^u}\Big) \lor \Big(u \geq C \land \phi\Big).
    \end{equation}
    In the subformula $\bigvee_{c = 0}^{C-1} (u = c) \land \phi\sub{2^c}{2^u}$ 
    the variable $u$ has already been linearised. 
    This is reflected in the procedure in lines~\ref{line:linearize-guess-const}--\ref{line:linearize-chi-equality}: if the procedure guesses $c \in [0,C-1]$ in line~\ref{line:linearize-guess-const}, then it 
    returns a pair~$\chi(u) \coloneqq (u = c)$ 
    and~$\gamma(v) \coloneqq \phi\sub{2^c}{2^u}$, 
    which corresponds to one of the disjuncts of~$\bigvee_{c = 0}^{C-1} (u = c) \land \phi\sub{2^c}{2^u}$. 
    The case when $u \geq C$ (equivalently, when the algorithm  guesses~$c =
    \star$ in line~\ref{line:linearize-guess-const}) is more involved and dealt
    with in the remaining part of the proof. Most of the work is concerned with
    manipulations of the formula~$(u \geq C \land \phi)$. 

    Recall that we are working under the assumption that $u \geq v \geq 0$ 
    (as stated in the signature of the procedure), which is essential for handling this case. 
    Let us start by considering the (in)equalities in $\chi$.
    Take one such (in)equality $a \cdot 2^u + b \cdot v + c \sim 0$, with ${\sim} \in \{<,\leq,=\}$ and $a \neq 0$. 
    Define $C'\coloneqq 2 \cdot\bigl\lceil\log(\frac{\abs{b}+\abs{c}+1}{\abs{a}})\bigr\rceil+3$. 
    For this constant, we can prove that~$\abs{a} \cdot 2^{C'}>\big(\abs{b}+\abs{c}\big)\cdot C'$.  
    \begin{align*}
        |a|\cdot 2^{C'} = |a|\cdot 2^3\cdot\left(\frac{|b|+|c|+1}{|a|}\right)^2&>\,\big(|b|+|c|\big)\cdot2^{\log(\frac{|b|+|c|+1}{|a|})+3}
        \\
        &>\,\big(|b|+|c|\big)\cdot2^{\bigl\lceil\log(\frac{|b|+|c|+1}{|a|})\bigr\rceil+2}
        \\
        &>\,\big(|b|+|c|\big)\cdot\Big(2\cdot\Big(\Bigl\lceil\log\big(\frac{|b|+|c|+1}{|a|}\big)\Bigr\rceil+1\Big)+1\Big)
        \\
        &=\,\big(|b|+|c|\big)\cdot\Big(2\cdot\Bigl\lceil\log\big(\frac{|b|+|c|+1}{|a|}\big)\Bigr\rceil+3\Big)\\
        &=\,\big(|b|+|c|\big)\cdot C',
    \end{align*}
    where for the second to last derivation we have used the fact that $2^{x+1}>2x+1$ for every $x\in\N$. 
    Note that both $\abs{a} \cdot 2^x$ and $(\abs{b} + \abs{c}) \cdot x$ are monotonous functions. 
    Then, since $2^x$ grows much faster than $x$ and (by definition) $C\geq C'$, we conclude that 
    $|a|\cdot 2^{x} > \big(|b|+|c|\big)\cdot x$ for every $x \geq C$.
    From this fact, we observe that 
    given two integers $u,v$ such that $u\geq v\geq 0$ and $u\geq C \geq 1$,
    the absolute value of $a\cdot2^u$ is greater than the absolute value of $b\cdot v+c$:
    \begin{equation*}
        |a|\cdot2^u> (|b|+|c|)\cdot u\geq|b|\cdot v+|c| \geq |b\cdot v+c|.
    \end{equation*} 
    This means that the inequality $a \cdot 2^u + b \cdot v + c \sim 0$ from $\chi$ simplifies:
    \begin{romanenumerate}
        \item if $a$ is negative and $\sim$ is either $<$ or $\leq$, then the inequality simplifies to $\top$;
        \item if $a$ is positive or $\sim$ is $=$, then the inequality simplifies to $\bot$.
        \label{item:linearise-bottom}
    \end{romanenumerate}
    Let $\chi'$ be the formula obtained from $\chi$ by applying the two rules above to each inequality. Note that this formula simplifies to either $\top$ or $\bot$. We have: 
    \begin{equation}
      \label{eq:proof:linearise:two}
      (u \geq v \geq 0) \implies \Big((u \geq C \land \phi) \iff (u \geq C \land \chi' \land \gamma)\Big).
    \end{equation}
    The restriction $u \geq v \geq 0$ implies that if $\chi'$ is false, then so is $u \geq C \land \phi$. 
    Thus, in terms of the procedure, the guess~$c = \star$ in line~\ref{line:linearize-guess-const} should lead to 
    an unsuccessful execution. This case is covered by line~\ref{line:assert-not-bottom}, where the procedure checks that 
    the condition~\eqref{item:linearise-bottom} above is never satisfied.

    In the formula $u \geq C \land \chi' \land \gamma$, the exponential term $2^u$ does not occur in (in)equalities, but it might still occur in the divisibility constraints of $\gamma$. The next step is thus to ``linearise'' these constraints. 
    Back to the definition of $C$ in
    line~\ref{line:linearize-max-constant},
    note that $C \geq n$ where $n \in \N$ is the maximum natural number such that $2^n$ divides $\fmod(\gamma)$ (note that $\fmod(\gamma) = \fmod(\phi)$). 
    Let $d$ be the positive odd number $\frac{\fmod(\gamma)}{2^n}$, as defined in line~\ref{line:linearize-mod-factor}. 
    Assuming $u \geq C$, the fact that $C \geq n$ implies that the remainder of $2^u$ modulo $\fmod(\phi)$ is of the form $2^n \cdot r$ for some $r \in [0,d-1]$ 
    (in particular, note that $r = 0$ when $d = 1$). 
    The following tautology holds:
    \begin{equation} \label{eq:linearize-divisibility}
      u \geq C \implies \Big(\gamma \iff \bigvee_{r=0}^{d-1}\big(( \fmod(\phi) \divides 2^{u}- 2^n \cdot r)\land\gamma\sub{2^n \cdot r}{2^u}\big)\Big).
    \end{equation}
    We have thus eliminated all occurrences of $u$ from $\gamma$, and we are left with the divisibility constraints $\fmod(\phi) \divides 2^{u}- 2^n \cdot r$. Note that these constraints can be reformulated as follows: 
    \begin{align*}  
      \fmod(\phi) \divides 2^{u}- 2^n \cdot r 
      &\iff{} d \divides 2^{u-n} - r\\
      &\iff{} d \divides 2^u - 2^n \cdot r 
      &\text{since $d$ and $2^n$ are coprime}
    \end{align*}
    We now want to ``linearise'' $d \divides 2^u - 2^n \cdot r$. There are two cases. 

    The first case corresponds to $(d\mid 2^{u}-2^n \cdot r)$
    not having a solution. That is, the discrete logarithm
    of $2^n \cdot r$ modulo $d$ does not exist. This
    happens, for instance, if $r = 0$ and $d \geq 3$. In
    this case, the corresponding disjunct
    in~\Cref{eq:linearize-divisibility} simplifies
    to~$\bot$, and accordingly the execution of the procedure aborts, 
    following the \textbf{assert} command in~line~\ref{line:linearize-assert}. 

    In the other case, 
    $(d\mid 2^{u}-2^n \cdot r)$ has a solution. 
    We remark (even though this will also be discussed later in the complexity analysis of the procedure) 
    that computing a solution for this discrete logarithm problem can be done in non-deterministic polynomial 
    time in the bit size of~$d$ and $2^n \cdot r$ (see~\cite{Joux14}). 
    In this case, we define $\ell_r$ to be the discrete logarithm of $2^n \cdot r$ modulo~$d$ ($\ell_r$ corresponds to $r'$ 
    in~line~\ref{line:linearize-discrete-log}).
    Let $R \subseteq [0,d-1]$ be the set of all~$r \in [0,d-1]$ such that $\ell_r$ is defined. 
    Furthermore, let $d'$ be the multiplicative order of $2$ modulo $d$ (which exists because $d$ is odd), that is, 
    $d'$ is the smallest integer from the interval $[1,d-1]$ such that $d \divides 2^{d'} - 1$ (as defined in line~\ref{line:linearize-mult-ord}). 
    Let us for the moment assume the following result:
    \begin{claim}
      \label{claim:arithmetic-progression}
      For every $r \in R$ and~$u \in \Z$,
      $(d\mid 2^{u}-2^n \cdot r)$ if and only if $(d' \mid u - \ell_r)$.
    \end{claim}
    Thanks to this claim, form~\Cref{eq:linearize-divisibility}, we obtain 
    \begin{equation}
      \label{eq:proof:linearise:three}
      u \geq C \implies \Big(\gamma \iff \bigvee_{r \in R} \big(( d' \divides u - \ell_r)\land\gamma\sub{2^n \cdot r}{2^u}\big)\Big).
    \end{equation}
    For the non-deterministic procedure, in the case when $\ell_r$ defined, the procedure returns 
    $\chi(u) \coloneqq u \geq C \land (d' \divides u - \ell_r)$ 
    and $\gamma\sub{2^n \cdot r}{2^u}$, as shown in lines~\ref{line:linearize-chi-case-2} and~\ref{line:linearize-chi-div}. 
    Across all non-deterministic branches, the \textbf{ensuring} part of the specification of~\Cref{algo:linearize} follows 
    by~\Cref{eq:proof:linerise:one,eq:proof:linearise:two,eq:proof:linearise:three}:
    \begin{align}
      &(u \geq v \geq 0) 
      \implies
      \\
      &\Big( 
        \phi \iff
        \bigvee_{c = 0}^{C-1} (u = c) \land \phi\sub{2^c}{2^u} \lor 
        \bigvee_{r \in R}
        (\chi' \land (u \geq C \land d' \divides u - \ell_r) \land \gamma\sub{2^n \cdot r}{2^u})
      \Big).\label{eq:proof:linearise:global-output}
    \end{align}
    Notice that the right-hand side of the double implication~\eqref{eq:proof:linearise:global-output} corresponds to the set of pairs
    \[ 
      \{(u=c, \phi\sub{2^c}{2^u}) : c \in [0,C-1]\} 
      \cup 
      \{(u \geq C \land d' \divides u - \ell_r, \gamma\sub{2^n \cdot r}{2^u})
      \},
    \]
    which
    is exactly the \textbf{global output} of~\Cref{algo:linearize}. 
    This completes the proof~\Cref{lemma:corr-linearize},
    subject to the proof of~\Cref{claim:arithmetic-progression} 
    which is given below.
\end{proof}

\begin{proof}[Proof of~\Cref{claim:arithmetic-progression}]
  Let $r \in R$ and $u \in \Z$.
  First, recall that by definition of $\ell_r$ and $d'$ we have $d \mid 2^{\ell_r} - 2^n \cdot r$ and $d \mid 2^{d'} - 1$. 
  This directly implies that, for every $\lambda \in \Z$, the divisibility $d \mid 2^{\lambda \cdot d' + \ell_r} - 2^n \cdot r$ holds. 
  Then, the right-to-left direction is straightforward, as $d' \divides u - \ell_r$ is indeed equivalent to 
  the statement ``there is $\lambda \in \Z$ such that $u = \lambda \cdot d' + \ell_r$''. 

  For the left-to-right direction, we need to prove that if $d \mid 2^{u} - 2^n \cdot r$ 
  then $u$ is of the form $u = \lambda \cdot d' + \ell_r$, for some $\lambda \in \Z$. 
  Consider $t_1,t_2 \in [0,d'-1]$ and $\lambda_1,\lambda_2 \in \Z$ 
  such that 
  $u = \lambda_1 \cdot d' + t_1$ and $\ell_r = \lambda_2 \cdot d' + t_2$. 
  We show that $t_1 - t_2$ is a multiple of $d'$, which in turn implies $u \equiv_{d'} \ell_r$ and therefore that $u$ is of the required form $u = \lambda \cdot d' + \ell_r$, for some $\lambda \in \Z$. 

  Suppose $t_1 \geq t_2$ (the case of $t_1 < t_2$ is analogous).
  From $d \divides 2^u - 2^n \cdot r$ and 
  $d \divides 2^{\ell_r} - 2^n \cdot r$ we derive $d \divides 2^u - 2^{\ell_r}$. Then, 
  \begin{align*}  
    d \divides 2^u - 2^{\ell_r}
    &\iff d \divides 2^{\lambda_1 \cdot d' + t_1} - 2^{\lambda_2 \cdot d' + t_2}\\
    &\iff
    d \divides 2^{t_1} - 2^{t_2} 
    &\text{by definition~of~$d'$}\\
    &\iff 
    d \divides 2^{t_2}(2^{t_1-t_2} - 1) 
    &\text{since $t_1 \geq t_2$}\\
    &\iff d \divides 2^{t_1-t_2} - 1 
    &\text{since $2$ and $d$ are coprime.}
  \end{align*}
  Hence, since $d'$ is the multiplicative order of $2$ modulo $d$, $t_1-t_2$ is a multiple of $d'$.  
\end{proof}

\subsection{Correctness of Algorithm~\ref{algo:elimmaxvar} (\ElimMaxVar)}\label{app:procedure-correctness-elimmaxvar}
In this section of~\Cref{app:procedure-correctness}, we prove that~\Cref{algo:elimmaxvar} (\ElimMaxVar)
complies with its specification (see~\Cref{lemma:corr-elimmaxvar} at the end of the section; page~\pageref{lemma:corr-elimmaxvar}). 

  Before moving to the correctness proof, we would like to remind the reader that
  \Cref{algo:elimmaxvar} (\ElimMaxVar) is a \emph{non-deterministic
  quantifier-elimination procedure}. 
  Non-determinism is used to guess integers (see lines~\ref{line:elimmaxvar-guess-rho}, 
  \ref{line:elimmaxvar-guess-mod}, \ref{line:elimmaxvar-guess-c}, and \ref{line:elimmaxvar-guess-times}) 
  as well as in the calls to~\Cref{algo:gaussianqe,algo:linearize}. 
  The result of computations in each non-deterministic branch $\beta$ is a linear-exponential
  system $\psi_\beta$. For the set $B$ of all such non-deterministic
  branches, the \textbf{global output} of
  \Cref{algo:elimmaxvar} is the set $\{\psi_\beta\}_{\beta\in B}$. 
  According to the specification given to \Cref{algo:elimmaxvar},
  for the disjunction $\bigvee_{\beta\in B}\psi_\beta$ (we will prove that) the
  following equivalence holds: 
      \begin{align}
          \exists x\exists\vec x'\big(\theta(\vec x)\land\varphi(\vec x, \vec x', \vec z')\land(x=\xquot\cdot2^y+\xrema)\big)&\label{eq:elimmaxvar-input}
          \\
          \iff
          \big(\exists x\,\theta(\vec x)\big)&\land\bigvee_{\beta\in B}\psi_\beta(\vec x\setminus x,\vec z').\label{eq:elimmaxvar-output}        
      \end{align}
      This equivalence demonstrates that the variables $x,\vec x'$ have been
      eliminated from the input formula $\varphi$. Since the ordering $\theta$ is
      a simple formula that plays only a structural role in the elimination
      process, we do not modify it explicitly (this is done in~\Cref{algo:master}
      directly after the call to~\Cref{algo:elimmaxvar}). 
      Because $2^x$ is the leading exponential term of the ordering $\theta$, it
      is sufficient to exclude the inequality $(2^y\leq 2^x)$ from this formula. 
      This leads to the full elimination of the variables $x$ and $\vec x'$. 
      Of course, while~\Cref{algo:elimmaxvar} can be regarded as a standard
      quantifier-elimination procedure, it only works for formulae of a very
      specific language (namely, the existential formulae of the
      form~(\ref{eq:elimmaxvar-input})).
 

    \subsubsection*{Three steps: their inputs and outputs}

    Following the discussion in the body of the paper, \Cref{algo:elimmaxvar} can be split into three steps, each achieving a specific goal:
    \begin{romanenumerate}
        \item lines~\ref{line:elimmaxvar-u}--\ref{line:elimmaxvar-gauss-1}: elimination of the quotient variables $\quot\coloneqq\vec x'\setminus \xquot $.
        \label{item:elimmaxvar-i} 
        \item lines~\ref{line:elimmaxvar-2-u}--\ref{line:elimmaxvar-delayed-subs}: linearisation of the variable $x$ and then its elimination by application of the delayed substitution $[\xquot \cdot 2^y + z'/x]$.
        \label{item:elimmaxvar-ii}
        \item lines~\ref{line:elimmaxvar-x1-start}--\ref{line:elimmaxvar-gauss-2}: elimination of the quotient variable $\xquot$.
        \label{item:elimmaxvar-iii}
    \end{romanenumerate} 
    Steps~(\ref{item:elimmaxvar-i}) and~(\ref{item:elimmaxvar-iii}) are quite similar. 
    The intermediate Step~(\ref{item:elimmaxvar-ii}) is essentially made of a single call to~\Cref{algo:linearize}. 

    Each step can be considered as an independent non-deterministic subroutine with its input and branch/global outputs. 
    Specification~\ref{spec:step1}, Specification~\ref{spec:step2}
    and Specification~\ref{spec:step3} 
    formally define these \textbf{inputs} and \textbf{global outputs}, for Step~(\ref{item:elimmaxvar-i}), 
    Step~(\ref{item:elimmaxvar-ii}) and Step~(\ref{item:elimmaxvar-iii}),
    respectively. With these specifications at hand, we organise the proof of correctness of~\Cref{algo:elimmaxvar} as follows. 
    We start by showing that \Cref{algo:elimmaxvar} is correct as soon as one assumes that the three steps comply with their specification.
    Afterwards, we prove (in independent subsections) that each step does indeed follow its specification. 
    
    \begin{figure}
    \begin{specification}[H]
        \caption{Step~(\ref{item:elimmaxvar-i}): Lines~\ref{line:elimmaxvar-u}--\ref{line:elimmaxvar-gauss-1} of~\Cref{algo:elimmaxvar}}
        \label{spec:step1}
        \begin{algorithmic}[1]
          \Require
            {\setlength{\tabcolsep}{0pt}
            \begin{tabular}[t]{rl}
                $\theta(\vec x):{}$ &\ ordering of exponentiated variables;\\
                $\varphi(\vec x,\vec x', \vec z'):{}$ &\ quotient system induced by~$\theta$, with $\vec x$ exponentiated,\\[-2pt] 
                &\ $\vec x'$ quotient, and $\vec z'$ remainder variables;\\
                $[\xquot \cdot 2^y + \xrema/x]:{}$ &\ delayed substitution for $\varphi$. (Recall: $\vec q \coloneqq \vec x' \setminus x'$.)  
            \end{tabular}}
          \GlobalOutput
          $\{(\gamma_{i}(u,\xquot),\psi_i(\vec x\setminus x,\vec z'))\}_{i\in I}:{}$ set of pairs such that for every $i \in I$, 
          \begin{itemize}
            \item $\gamma_i$ is a linear system containing the inequality $x' \geq 0$.
            \item $\psi_i$ is a linear-exponential system containing inequalities $0\leq z <2^y$ for every $z$ in $\vec z'$.
          \end{itemize}
          Over $\N$, the formula $\exists\vec q\big(\theta(\vec x)\land\varphi(\vec x, x', \vec q, \vec z')\land(x=\xquot\cdot2^y+\xrema)\big)$
          is equivalent~to%
          \begin{equation}
              \textstyle\bigvee_{i\in I}\theta(\vec x)\land(x=\xquot\cdot2^y+\xrema)\land \exists u
              \big(\gamma_i(u, \xquot) \land u=2^{x-y}\big)\land\psi_i(\vec x\setminus x,\vec z').\label{eq:elimmaxvar-after-s1}
          \end{equation}  
        \end{algorithmic}
    \end{specification}%

    \begin{specification}[H]
        \caption{Step~(\ref{item:elimmaxvar-ii}): Lines~\ref{line:elimmaxvar-2-u}--\ref{line:elimmaxvar-delayed-subs} of~\Cref{algo:elimmaxvar}}
        \label{spec:step2}
        \begin{algorithmic}[1]
          \Require
            {\setlength{\tabcolsep}{0pt}
            \begin{tabular}[t]{rl}
                $\theta(\vec x):{}$ &\ ordering of exponentiated variables (same as in~Step~\eqref{item:elimmaxvar-i});\\
                $(\gamma(u,\xquot),\psi(\vec x\setminus x,\vec z')):{}$ &\ branch output of Step~(\ref{item:elimmaxvar-i});\\
                $[\xquot \cdot 2^y + \xrema/x]:{}$ &\ delayed substitution (same as in~Step~\eqref{item:elimmaxvar-i}).  
            \end{tabular}}
          \GlobalOutput 
            $\{(\chi_{j}(y,\xquot,\xrema),\gamma_{j}(\xquot),\psi(\vec x\setminus x,\vec z'))\}_{j\in J}:{}$ set of triples such that for $j \in J$
            \begin{enumerate}
                \item $\chi_{j}(y,\xquot,\xrema)$ is either 
                $(\xquot\cdot2^y+\xrema-y=a)$ 
                or $(\xquot\cdot2^y+\xrema-y\geq b) \land (d \mid \xquot\cdot2^y+\xrema-y-r)$, 
                for some  
                $a,b,d,r\in\N$ and $b > 2$ (that depend on $j$); 
                \label{item:elimmaxvar-chi-1} 
                \item $\gamma_{j}(\xquot)$ is a linear system containing the inequality $x' \geq 0$.
                \label{item:elimmaxvar-chi-2}
                \item $\psi(\vec x\setminus x,\vec z')$ is the system in the input of the algorithm (it is not modified).
                \label{item:elimmaxvar-chi-3}
            \end{enumerate}
            Over $\N$, given $\Psi(\vec x,x',z') \coloneqq \theta(\vec x)\land(x=\xquot\cdot2^y+\xrema)\land \exists u \big(\gamma(u, \xquot) \land u=2^{x-y}\big)$, we have
            \begin{equation}\label{eq:elimmaxvar-psi}
                (\exists x \Psi(\vec x, x',z')) \iff \bigvee_{j\in J} (\exists x \theta(\vec x))\land\chi_{j}(y,\xquot,\xrema)\land\gamma_{j}(\xquot).
            \end{equation}
        \end{algorithmic}
    \end{specification}%

    \begin{specification}[H]
        \caption{Step~(\ref{item:elimmaxvar-iii}): Lines~\ref{line:elimmaxvar-x1-start}--\ref{line:elimmaxvar-gauss-2} of~\Cref{algo:elimmaxvar}}
        \label{spec:step3}
        \begin{algorithmic}[1]
          \Require
            {\setlength{\tabcolsep}{0pt}
            \begin{tabular}[t]{rl}
                $\theta(\vec x):{}$ &\ ordering of exponentiated variables (as in~Steps~\eqref{item:elimmaxvar-i});\\
                \hspace{-0.5cm}$(\chi(y,\xquot,\xrema),\gamma(\xquot),\psi(\vec x\setminus x,\vec z')):{}$ &\ branch output of Step~(\ref{item:elimmaxvar-ii}).
            \end{tabular}}
          \GlobalOutput 
          $\{\psi_{k}(y,\xrema)\land\psi(\vec x\setminus x,\vec z')\}_{k\in K}:{}$ set of linear-exponential systems
          (the system $\psi(\vec x\setminus x,\vec z')$ is as in the input of the algorithm). 
          Over $\N$, the following formula holds:
          \begin{equation}\label{eq:elimmaxvar-psi-x'}
            0 \leq z' < 2^y \implies \Big(\exists \xquot\big(\chi(y,\xquot,\xrema)\land\gamma(\xquot)\big) 
            \iff \bigvee_{k\in K} \psi_{k}(y,\xrema)\Big).
          \end{equation}
        \end{algorithmic}
    \end{specification}%
  \end{figure}

    Let us first briefly discuss the three specifications. From
    Specification~\ref{spec:step1}, we see that the input of
    Step~(\ref{item:elimmaxvar-i}) is the same as of~\Cref{algo:elimmaxvar}. Its
    \textbf{global output} is a set of pairs of systems
    $\{(\gamma_i,\psi_i)\}_{i\in I}$, where the $\psi_i$ are linear-exponential
    systems featuring only variables from the vector $(\vec x\setminus x,\vec z')$, 
    and the $\gamma_i$ are linear systems only featuring the variables $x'$ and $u$ 
    (which is a placeholder for $2^{x-y}$). By \textbf{branch output} of this step we mean a particular
    pair~$(\gamma,\psi)$ of systems from the \textbf{global output}. The systems
    $\gamma$ and $\psi$ correspond to the content of the homonymous ``program
    variables'' in the pseudocode, after line~\ref{line:elimmaxvar-gauss-1} has
    been executed.

    To visualise the specification of Step~(\ref{item:elimmaxvar-ii}), consider
    the diagram given in~\Cref{fig: flowchart}. In this diagram,
    Step~(\ref{item:elimmaxvar-ii}) is enclosed within the darker background
    rectangle. Four arrows enter this rectangle. Two of them correspond to a
    branch output of Step~(\ref{item:elimmaxvar-i}), and the other two come from
    the substitution for the variable $x$ and the ordering $\theta(\vec x)$.
    Observe that computations at Step~(\ref{item:elimmaxvar-ii}) do not affect
    the second parameter of the branch output of Step~(\ref{item:elimmaxvar-i}).
    The linear-exponential system $\psi(\vec x\setminus x,\vec z')$ is just
    propagated to the next step. The \textbf{branch output} of
    Step~(\ref{item:elimmaxvar-ii}) is a triple of systems
    $(\chi,\gamma,\psi)$, corresponding to the content of the homonymous
    program variables, after line~\ref{line:elimmaxvar-delayed-subs} has been
    executed.

    The specification of Step~(\ref{item:elimmaxvar-iii}) also follows the
    diagram from~\Cref{fig: flowchart}. This step takes as input a branch output
    of the previous step (three arrows that correspond to the systems
    $\gamma_3$, $\chi$,~$\psi_1$) and the ordering $\theta(\vec x)$. Its
    \textbf{branch output} is a single linear-exponential system that
    corresponds to the output of~\Cref{algo:elimmaxvar}.

    We prove the aforementioned conditional statement about the correctness of~\Cref{algo:elimmaxvar}.

\begin{lemma}\label{lemma:corr-elimmaxvar-cond}
    If Steps~(\ref{item:elimmaxvar-i}), (\ref{item:elimmaxvar-ii}), 
     and~(\ref{item:elimmaxvar-iii}) comply with, 
     respectively, Specifications~\ref{spec:step1}, \ref{spec:step2}, and \ref{spec:step3} 
     then \Cref{algo:elimmaxvar} (\ElimMaxVar) complies with its specification.
\end{lemma}

\begin{proof}
    In a nutshell, the proof of the equivalence between the formulas~\eqref{eq:elimmaxvar-input} and~\eqref{eq:elimmaxvar-output} 
    follows by chaining the equivalences appearing in Specifications~\ref{spec:step1}, \ref{spec:step2}, and \ref{spec:step3}.

    Consider an input to~\Cref{algo:elimmaxvar}:
    an ordering $\theta(\vec x)$, a quotient system $\varphi(\vec x,\vec x',\vec z')$ induced by~$\theta$, and a delayed substitution $[\xquot \cdot 2^y + \xrema/x]$ for $\phi$. 
    According to the specification of~\Cref{algo:elimmaxvar}, 
    it suffices to show the aforementioned equivalence 
    and that each \textbf{branch output}~$\psi_\beta$ of the algorithm 
    contains the inequalities $0 \leq z < 2^y$, for every $z$ in $\vec z'$. Below we focus on proving the equivalence, and derive the additional property on $\psi_\beta$ as a by-product.


    Following~Specification~\ref{spec:step1}, with this input,
    Step~(\ref{item:elimmaxvar-i}) constructs a set of pairs of systems
    $\{(\gamma_{i}(u,\xquot),\psi_i(\vec x\setminus x,\vec z'))\}_{i\in I}$. 
    Given $i \in I$, we define the formula 
    \[
        \Psi_i(\vec x, x',\xrema)\coloneqq \theta(\vec x)\land(x=\xquot\cdot2^y+\xrema)\land \exists u \big(\gamma_i(u, \xquot) \land u=2^{x-y}\big).
    \]
    The disjunction~(\ref{eq:elimmaxvar-after-s1}) appearing in~Specification~\ref{spec:step1} can be rewritten in a compact way
    by using this formula. From the \textbf{global output} of Specification~\ref{spec:step1}, we obtain the equivalence: 
    \begin{align}
        \exists x\exists\vec x'\big(\theta(\vec x)\land\varphi(\vec x, \vec x', \vec z')\land(x=\xquot\cdot2^y+\xrema)\big)&&\hspace{-0.1cm}\text{(i.e.~\eqref{eq:elimmaxvar-input})}\nonumber
        \\
        &\hspace{-2cm}\iff
        \bigvee_{i\in I}\big(\exists x \exists x' \Psi_i(\vec x, x', \xrema)\big)\land\psi_i(\vec x\setminus x,\vec z').\label{eq:elimmaxvar-after-s1-psi}
    \end{align} 

    Following Specification~\ref{spec:step2},
    in addition to the ordering $\theta(\vec x)$ and the delayed substitution $[\xquot \cdot 2^y + \xrema/x]$, 
    Step~(\ref{item:elimmaxvar-ii}) takes as input a branch output $(\gamma_i(u,\xquot),\psi_i(\vec x\setminus x ,\vec z'))$ of Step~(\ref{item:elimmaxvar-i}), 
    for some $i\in I$. 
    Note that this pair of linear-exponential systems corresponds to a single disjunct of the formula~\eqref{eq:elimmaxvar-after-s1-psi}.
    The \textbf{global output} of Step~(\ref{item:elimmaxvar-ii}) on this input is a set of triples 
    of linear-exponential systems 
    $\{(\chi_{i,j}(y,\xquot,\xrema),\gamma_{i,j}(\xquot),\psi_i(\vec x\setminus x,\vec z'))\}_{j\in J_i}$
    such that, according to the equivalence~\eqref{eq:elimmaxvar-psi} in Specification~\ref{spec:step2}, for every $i \in I$,
    \begin{equation}\label{eq:elimmaxvar-psi-i}
        (\exists x \exists x' \Psi_i(\vec x,\xquot,\xrema))\iff\bigvee_{j\in J_i} (\exists x\theta(\vec x))\land \exists \xquot  \big(\chi_{i,j}(y,\xquot,\xrema)\land\gamma_{i,j}(\xquot)\big).
    \end{equation}

    Consider the combination of Steps~(\ref{item:elimmaxvar-i}) and~(\ref{item:elimmaxvar-ii}).
    The output of the two steps combined is given by the following set:
    $\bigl\{(\chi_{i,j},\gamma_{i,j},\psi_i)\,:\,i\in I, j\in J_i\bigr\}$.     
    From the equivalences in~\eqref{eq:elimmaxvar-after-s1-psi} and~\eqref{eq:elimmaxvar-psi-i}, we obtain the following chain of equivalences:
    \begin{align}
        \exists x\exists\vec x'\big(\theta(\vec x)\land{}&\varphi(\vec x, \vec x', \vec z')\land(x=\xquot\cdot2^y+\xrema)\big)\nonumber
        \\
        &\iff \bigvee_{i\in I}\big( \exists x \exists x' \Psi_i(\vec x, \xquot, \xrema) \big)\land\psi_i(\vec x\setminus x,\vec z')\nonumber
        \\
        &\iff
        \bigvee_{i\in I}\bigvee_{j\in J_i}\big(\exists x\theta(\vec x)\big)\land\exists \xquot\big(\chi_{i,j}(y,\xquot,\xrema)\land\gamma_{i,j}(\xquot)\big)\land\psi_i(\vec x\setminus x,\vec z').\label{eq:elimmaxvar-after-s2}
    \end{align}  

    Following Specification~\ref{spec:step3},
    the input of Step~(\ref{item:elimmaxvar-iii}) is the ordering $\theta(\vec x)$ together with 
    a branch output $(\chi_{i,j},\gamma_{i,j},\psi_i)$ of Step~(\ref{item:elimmaxvar-ii}), for some $(i,j)\in I\times J_i$.
    The \textbf{global output} of Step~(\ref{item:elimmaxvar-iii}) on this input is a set 
    $\{\psi_{i,j,k}\land\psi_i\}_{k\in K_{i,j}}$ of linear-exponential systems.
    Since, given the specification of Step~(\ref{item:elimmaxvar-i}), 
    each $\psi_i$ contains the inequality $0 \leq z' < 2^y$,
    thanks to the formula~\eqref{eq:elimmaxvar-psi-x'} in Specification~\ref{spec:step3} we have
    \begin{equation}\label{eq:elimmaxvar-psi-x1-i}
        \exists \xquot\big(\chi_{i,j}(y,\xquot,\xrema)\land\gamma_{i,j}(\xquot)\big) \land\psi_i(\vec x\setminus x,\vec z')
        \iff\bigvee_{k\in K_{i,j}} \psi_{i,j,k}(y,\xrema) \land\psi_i(\vec x\setminus x,\vec z').
    \end{equation}

    Let $\{\psi_{\beta}\}_{\beta \in B}$ be the \textbf{global output} 
    of \Cref{algo:elimmaxvar}, where $B$ is a set of non-deterministic branches.
    This corresponds to the \textbf{global output} of Step~(\ref{item:elimmaxvar-iii}), i.e.,
    \[
        \Bigl\{\psi_\beta\,:\,\beta\in B\Bigr\}=
        \Bigl\{\psi_{i,j,k}(y,\xrema)\land\psi_i(\vec x\setminus x ,\vec z')\,:\,i\in I, j\in J_i,\text{ and }k\in K_{i,j}\Bigr\}.
    \]
    (Notice that this means that $\psi_\beta$ features variables from $(\vec x\setminus x,\vec z')$.)
    Combining~(\ref{eq:elimmaxvar-after-s2}) and~(\ref{eq:elimmaxvar-psi-x1-i}), 
    we obtain the desired equivalence between~\eqref{eq:elimmaxvar-input} and~\eqref{eq:elimmaxvar-output}:
    \begin{align*}
        & \exists x\exists\vec x'\big(\theta(\vec x)\land\varphi(\vec x, \vec x', \vec z')\land(x=\xquot\cdot2^y+\xrema)\big)
        & \text{(i.e.~\eqref{eq:elimmaxvar-input})}\\
        \iff{}& 
        \bigvee_{i\in I}\bigvee_{j\in J_i}\big(\exists x\theta(\vec x)\big)\land\exists \xquot\big(\chi_{i,j}(y,\xquot,\xrema)\land\gamma_{i,j}(\xquot)\big)\land\psi_i(\vec x\setminus x,\vec z')
        \\
        \iff{}&
        \big(\exists x\theta(\vec x)\big)\land\bigvee_{i\in I}\;\bigvee_{j\in J_i}\,\bigvee_{k\in K_{i,j}}\big(\psi_{i,j,k}(y,\xrema)\land\psi_i(\vec x\setminus x ,\vec z')\big)
        \\
        \iff{}&
        \big(\exists x\theta(\vec x)\big)\land\bigvee_{\beta\in B}\psi_{\beta}(\vec x\setminus x ,\vec z')
        &\text{(i.e.~\eqref{eq:elimmaxvar-output})}. 
    \end{align*}

    To conclude the proof, observe that each $\psi_\beta$ features a system~$\psi_i$, for some $i \in I$. 
    From the specification of Step~(\ref{item:elimmaxvar-i}), 
    each $\psi_i$ contains the inequalities $0 \leq z < 2^y$ for every $z$ in~$\vec z'$.
\end{proof}

To make~\Cref{lemma:corr-elimmaxvar-cond} unconditional,
it now suffices to prove correctness of Steps~(\ref{item:elimmaxvar-i}), 
(\ref{item:elimmaxvar-ii}), and~(\ref{item:elimmaxvar-iii}). 
This is done in the following subsections.

\subsubsection*{Correctness of Step~(\ref{item:elimmaxvar-i})}

The goal of this subsection is to prove that the non-deterministic algorithm that corresponds 
to lines~\ref{line:elimmaxvar-u}--\ref{line:elimmaxvar-gauss-1}
of~\Cref{algo:elimmaxvar} complies with Specification~\ref{spec:step1}.   

Our main concern is the \textbf{foreach} loop of line~\ref{line:elimmaxvar-foreach}, which considers sequentially 
all constraints $(\tau\sim0)$ of the input quotient system $\varphi(\vec x, \vec x',\vec z')$.  
As discussed in the body of the paper, the goal of this loop is to split 
each constraint into a ``left part'' and a ``right part'' (see, respectively, $\gamma_1$ and $\psi_1$ in the diagram of~\Cref{fig: flowchart}). The left part corresponds to a linear constraint over the quotient variables $\vec x'$ and the auxiliary variable $u$ (which is an alias for $2^{x-y}$). 
The right part is a linear-exponential system over the variables $\vec x \setminus x$ and $\vec z'$. 
In a nutshell, the split into these two parts is possible because of the three equivalences given in the following two lemmas.

\begin{lemma} 
    \label{lemma:split:congruences}
    Let $d,M \in \N$, with $M \geq d \geq 1$. Given $y \in \N$ and $z,w \in \Z$, we have
    \[ 
        z \cdot 2^y + w \equiv_d 0
        \iff \textstyle\bigvee_{r = 1}^{M} \big(z - r \equiv_d 0 \land r \cdot 2^y + w \equiv_d 0\big).
    \]
\end{lemma}

\begin{proof} 
    Informally, this lemma states that $z$ can be replaced with a number in $[1,M]$ congruent to $z$ modulo $d$. 
    The proof is obvious.
\end{proof}

\begin{lemma} 
    \label{lemma:split:inequalities}
    Let $C,D \in \Z$, with $C \leq D$. For $y \in \N$, $z \in \Z$, and ${w \in [C \cdot 2^y,D \cdot 2^y]}$, the following equivalences hold
    \begin{enumerate} 
        \item\label{step1:fund-equiv-1} $z \cdot 2^y + w = 0
        \iff \bigvee_{r = C}^{D} \big(z + r = 0 \land w = r \cdot 2^y\big)$,
        \item\label{step1:fund-equiv-2} $z \cdot 2^y + w \leq 0 
        \iff \bigvee_{r = C}^{D} \big(z + r \leq 0 \land (r-1) \cdot 2^y < w \leq r \cdot 2^y\big)$,
        \item\label{step1:fund-equiv-3} $z \cdot 2^y + w < 0 
        \iff \bigvee_{r = C}^{D} \big(z + r+1 \leq 0 \land  w = r \cdot 2^y\big) \lor \big(z + r \leq 0 \land (r-1) \cdot 2^y < w < r \cdot 2^y\big)$.
    \end{enumerate}
\end{lemma}

\begin{proof} 
    Firstly, notice that since $w \in [C \cdot 2^y,D \cdot 2^y]$, 
    there is $r^* \in [C,D]$ such that $\ceil{\frac{w}{2^y}} = r^*$.
    
    \textit{Proof of~\eqref{step1:fund-equiv-1}.} 
    For the left-to-right direction, note that $z \cdot 2^y + w = 0$ forces $w$ to be divisible by $2^y$. Hence $w = r^* \cdot 2^y$, 
    and we have $z \cdot 2^y + r^* \cdot 2^y = 0$, i.e., $z + r^* = 0$. 
    Since $r^*$ belongs to $[C,D]$, the right-hand side is satisfied.
    The right-to-left direction is trivial. 

    \textit{Proof of~\eqref{step1:fund-equiv-2}.} 
    Observe that given $r \in [C,D]$ satisfying 
    $(r-1) \cdot 2^y < w \leq r \cdot 2^y$, 
    we have $r = r^*$. Therefore, it suffices to show the equivalence $z \cdot 2^y + w \leq 0 \iff z + r^* \leq 0$. If $w$ is divisible by $2^y$, then $w = r^*\cdot 2^y$ and the equivalence easily follows: 
    \[ 
        z \cdot 2^y + w \leq 0 \iff z \cdot 2^y + r^* \cdot 2^y \leq 0 \iff z+r^* \leq 0.
    \]
    Otherwise, when $w$ is not divisible by $2^y$, it holds that $\floor{\frac{w}{2^y}} = r^*-1$. 
    Below, given $t \in \R$ we let $\frpart{t} \coloneqq t - \floor{t}$, i.e., $\frpart{t}$ is the fractional part of $t$. 
    Observe that $0 \leq \frpart{t} < 1$.
    We have: 
    \begin{align} 
        z \cdot 2^y + w \leq 0 
        &\iff 
        z \cdot 2^y + \left(\floor{\frac{w}{2^y}} + \frpart{\frac{w}{2^y}}\right) \cdot 2^y \leq 0 \notag\\
        &\iff 
        z + \floor{\frac{w}{2^y}} + \frpart{\frac{w}{2^y}} \leq 0 \notag\\
        &\iff 
        z + \floor{\frac{w}{2^y}} < 0
        &\text{since $w$ is not divisible by $2^y$} \notag\\
        &\iff 
        z + r^*-1 < 0 \notag\\
        &\iff z + r^* \leq 0. \label{step1:fund-equiv:last-step}
    \end{align}

    \textit{Proof of~\eqref{step1:fund-equiv-3}.}
    While their equivalences look different, 
    this and the previous case have very similar proofs.
    This similarity stems from the fact that, when $w$ is not divisible by $2^y$, 
    then $z \cdot 2^y + w$ cannot be $0$, and thus the cases of 
    $z \cdot 2^y + w \leq 0$
    and $z \cdot 2^y + w < 0$ 
    become identical. Below, we formalise the full proof for completeness.

    Since $w \in [C \cdot 2^y,D \cdot 2^y]$, 
    there must be $r \in [C,D]$ such that either 
    $w = r \cdot 2^y$ or
    $(r-1) \cdot 2^y < w < r \cdot 2^y$.
    In both cases, $r = r^*$,
    and thus to conclude the proof it suffices to establish 
    that: 
    \begin{itemize}
        \item $w = r^* \cdot 2^y$ implies $z \cdot 2^y + w < 0 \iff z+r^*+1 \leq 0$, and
        \item $(r^*-1) \cdot 2^y < w < r^* \cdot 2^y$ implies $z \cdot 2^y + w < 0 \iff z+r^* \leq 0$.
    \end{itemize}
    The proof of the first item is straightforward. Assuming $w = r^* \cdot 2^y$, we get: 
    \begin{align*}
        z \cdot 2^y + w < 0 &\iff z \cdot 2^y + r^* \cdot 2^y < 0 \iff z+r^* < 0 \iff z+r^* +1 \leq 0.
    \end{align*}
    For the second item, assume that $(r^*-1) \cdot 2^y < w < r^* \cdot 2^y$. 
    In this case $w$ is not divisible by $2^y$, and $\floor{\frac{w}{2^y}} = r^*-1$. Hence, $z \cdot 2^y + w$ cannot be $0$, which in turn means 
    ${z \cdot 2^y + w < 0 \iff z \cdot 2^y + w \leq 0}$. 
    Therefore, we can apply the same sequence of equivalences from~\eqref{step1:fund-equiv:last-step}
    to show that $z \cdot 2^y + w < 0 \iff z + r^* \leq 0$.
\end{proof}

Looking at the pseudocode of~\Cref{algo:elimmaxvar}, 
one can see that the \textbf{foreach} loop of line~\ref{line:elimmaxvar-foreach} does indeed follow 
the equivalences in~\Cref{lemma:split:congruences,lemma:split:inequalities}. 
The equivalences in~\Cref{lemma:split:inequalities} are applied in lines~\ref{line:elimmaxvar-delta-undefined}--\ref{line:elimmaxvar-psi-equality}, setting $[C,D] = [-\onenorm{\rho},\onenorm{\rho}]$.
The equivalence in~\Cref{lemma:split:congruences} 
is applied in lines~\ref{line:elimmaxvar-guess-mod}--\ref{line:elimmaxvar-psi-divisibility}, setting $M = \fmod(\phi)$. 
Ultimately, the correctness of  Step~(\ref{item:elimmaxvar-i}), which we now formalise,  
follows from these equivalences (and from the correctness of~\Cref{algo:gaussianqe}).

We divide the proof of correctness into the following four steps:
\begin{enumerate}
    \item We show that the map $\Delta$ has no influence in the correctness of the algorithm and can be ignored during the analysis. This is done to simplify the exposition of the next step.
    \item We analyse the body of the \textbf{foreach} loop of line~\ref{line:elimmaxvar-foreach}. Here we use~\Cref{lemma:split:congruences,lemma:split:inequalities}.
    \item We analyse the complete execution of the \textbf{foreach} loop, hence obtaining a specification for lines~\ref{line:elimmaxvar-u}--\ref{line:elimmaxvar-psi-divisibility} of Step~(\ref{item:elimmaxvar-i}). 
    \item We incorporate the call to~\Cref{algo:gaussianqe} (\GaussianQE) performed in line~\ref{line:elimmaxvar-gauss-1} into the analysis, proving that Step~(\ref{item:elimmaxvar-i}) follows Specification~\ref{spec:step1}.
\end{enumerate}

\subparagraph*{The map $\Delta$ is not needed for correctness.}
For the correctness of Step~(\ref{item:elimmaxvar-i}), 
the first simplifying step consists in doing a program transformation that removes the uses of the map $\Delta$.  
This map is introduced exclusively for complexity reasons, and the correctness of the algorithm is preserved if one removes it. 
More precisely, instead of guessing the integer~$r$ in line~\ref{line:elimmaxvar-guess-rho} only once for each least significant part $\rho$, one can perform one such guess every time $\rho$ is found.

\begin{lemma}
    \label{lemma:remove-Delta}
    Consider the code obtained from~Step~(\ref{item:elimmaxvar-i}) by replacing lines~\ref{line:elimmaxvar-delta-undefined}--\ref{line:elimmaxvar-delta-from} with lines~\ref{line:elimmaxvar-guess-rho} and~\ref{line:elimmaxvar-psi-inequality}. 
    If it complies with Specification~\ref{spec:step1}, then so 
    does~Step~(\ref{item:elimmaxvar-i}).
\end{lemma}

\begin{proof}
    Suppose that the modified Step~(\ref{item:elimmaxvar-i})
    complies with Specification~\ref{spec:step1}, 
    which in particular means that its \textbf{global output}
    is a set  $\{(\gamma_{i}(u,\xquot),\psi_i(\vec x\setminus x,\vec z'))\}_{i\in I}:{}$ set of pairs such that 
    \begin{enumerate}
      \item\label{remove-Delta:item1} $\gamma_i$ is a linear system containing the inequality $x' \geq 0$,
      \item\label{remove-Delta:item2} $\psi_i$ is a linear-exponential system containing inequalities $0\leq z <2^y$ for every $z$ in $\vec z'$,
    \end{enumerate}
    and, over $\N$, the formula $\exists\vec q\big(\theta(\vec x)\land\varphi(\vec x, x', \vec q, \vec z')\land(x=\xquot\cdot2^y+\xrema)\big)$
    is equivalent~to%
    \begin{equation}
        \textstyle\bigvee_{i\in I}\theta(\vec x)\land(x=\xquot\cdot2^y+\xrema)\land \exists u
        \big(\gamma_i(u, \xquot) \land u=2^{x-y}\big)\land\psi_i(\vec x\setminus x,\vec z').\label{eq:elimmaxvar-after-s1-modified}
    \end{equation}  

    Observe that the modification done to the algorithm only influences the number of constraints of the form $(r-1)\cdot2^y<\rho \land \rho\leq r\cdot2^y$
    that are added to the system~$\psi$, 
    and corresponding to the guesses of $r \in [-\norm{\rho}_1,\norm{\rho}_1]$.
    More precisely, when a single $\rho$ is encountered multiple times during the procedure, the modified Step~(\ref{item:elimmaxvar-i}) is allowed to guess multiple values for $r$, whereas the original Step~(\ref{item:elimmaxvar-i}) reuses the same $r$. 
    So, 
    for every pair of systems $(\gamma,\psi)$ in the \textbf{global output} of the original Step~(\ref{item:elimmaxvar-i}),
    there is a system $\psi'$ such that 
    \begin{itemize} 
        \item  $(\gamma,\psi')$ belongs to the \textbf{global output} of 
        the modified Step~(\ref{item:elimmaxvar-i}), 
        \item $\psi'$ can be obtained from~$\psi$ by duplicating a certain number of times formulae of the form $(r-1)\cdot2^y<\rho \land \rho\leq r\cdot2^y$ that already appear in $\psi$.
    \end{itemize}
    Then, clearly, also the \textbf{global output} of the original Step~(\ref{item:elimmaxvar-i})
    satisfies~\Cref{remove-Delta:item1,remove-Delta:item2} above. It also satisfies the equivalence involving the formula~\eqref{eq:elimmaxvar-after-s1-modified}.
    This is because all the systems $\psi_i$ (with $i \in I$) in the \textbf{global output} of the modified Step~(\ref{item:elimmaxvar-i}) that do not correspond, 
    in the sense we have just discussed, to a $\psi$ in the \textbf{global output} of the original Step~(\ref{item:elimmaxvar-i}), are unsatisfiable.
    The reason for their unsatisfiability is that these systems $\psi_i$ feature
    constraints $(r_1-1)\cdot2^y<\rho \land \rho\leq r_1\cdot2^y$ and 
    $(r_2-1)\cdot2^y<\rho \land \rho\leq r_2\cdot2^y$ with $r_1 \neq r_2$. 
    As the same term $\rho$ appears in these constraints,
    their conjunction is unsatisfiable. 
    Thus, the disjunct of formula~\eqref{eq:elimmaxvar-after-s1-modified} corresponding to such a $\psi_i$
    can be dropped without changing the truth of the equivalence, 
    and we conclude that the original Step~(\ref{item:elimmaxvar-i}) complies with Specification~\ref{spec:step1}.
\end{proof}

\begin{remark}
    \label{remark:remove-Delta}
    Following~\Cref{lemma:remove-Delta}, for the remaining part of the proof of correctness of~Step~(\ref{item:elimmaxvar-i}) 
    we assume this step to only feature lines~\ref{line:elimmaxvar-guess-rho} and~\ref{line:elimmaxvar-psi-inequality} in place of lines~\ref{line:elimmaxvar-delta-undefined}--\ref{line:elimmaxvar-delta-from}. 
\end{remark}

\subparagraph*{Analysis of the body of the foreach loop (of line~\ref{line:elimmaxvar-foreach}).}     
    For the rest of~\Cref{app:procedure-correctness-elimmaxvar},
    we simply say \emph{``the \textbf{foreach} loop''} without referring to its line number, as there are no other loops in~\Cref{algo:elimmaxvar}.
    We start the analysis by considering a single iteration of this loop.
    Given an input $(\theta,\phi,[\xquot \cdot 2^y + \xrema/x])$ of~\Cref{algo:elimmaxvar} (which corresponds to the input of Step~(\ref{item:elimmaxvar-i}), see Specification~\ref{spec:step1}), 
    and a constraint $(\tau \sim 0)$ from $\phi$, 
    below we say that \emph{executing the \textbf{foreach} body on the state $(\tau\sim0,\gamma,\psi)$ yields as a \textbf{global output} a set $S$}
    whenever: 
    \begin{itemize}
        \item $\tau \sim 0$ is the constraint selected in line~\ref{line:elimmaxvar-foreach}, and $\gamma$ and $\psi$ are the systems 
        stored in the homonymous variables when $\tau \sim 0$ is selected
        (these systems are initially~$\top$, see line~\ref{line:elimmaxvar-initialize}).
        \item $S$ is the union over all non-deterministic branches of the pairs of systems $(\gamma',\psi')$ stored in the variables $(\gamma,\psi)$ after 
        the \textbf{foreach} loop completes its iteration on $\tau \sim 0$ (i.e., the body of the loop is executed exactly once, and the program reaches line~\ref{line:elimmaxvar-foreach} again).
    \end{itemize}
    The following lemma describes the effects of one iteration of the \textbf{foreach} loop.

    \begin{lemma}\label{claim:elimmaxvar-3}
        Let $(\theta,\varphi,[\xquot \cdot 2^y + \xrema/x])$ be an input of Step~(\ref{item:elimmaxvar-i}) 
        described in Specification~\ref{spec:step1}. 
        Let $u$ be the fresh variable defined in line~\ref{line:elimmaxvar-u},
        and let $(\tau\sim0)$ be a constraint from $\varphi$. 
        Executing the \textbf{foreach} body on the state $(\tau\sim0,\gamma,\psi)$ yields as a \textbf{global output} a set of pairs $\{(\gamma \land \gamma_r, \psi \land \psi_r)\}_{r\in R}$,
        for some finite set of indices $R$, such that
        \begin{enumerate}
            \item $\gamma_r$ is a linear (in)equality or divisibility constraint over the variables $\vec x'$ and $u$.
            \label{item:claim-elimmaxvar-3-1}
            \item $\psi_r$ is a linear-exponential constraint over the variables $\vec x\setminus x$ and $\vec z'$.
            \label{item:claim-elimmaxvar-3-2}
            \item If $\tau$ only features variables from $\vec x \setminus x$ and $\vec z'$, then $R = \{0\}$, and $\gamma_0 = \top$ and $\psi_0 = (\tau \sim 0)$.%
            \label{item:claim-elimmaxvar-3-3}
            \item Over $\N$, $\theta\land\bigwedge_{z \in \vec z'}(0 \leq z < 2^{y})\implies
            \big((\tau \sim 0)\iff\bigvee_{r\in R}(\gamma_r[2^{x-y}/u]\land\psi_r)\big)$.
            \label{item:claim-elimmaxvar-3-4}
        \end{enumerate}
    \end{lemma}
    \begin{proof}
        In every constraint $(\tau \sim 0)$ of the system $\varphi$,
        the term $\tau$ is a quotient term induced by~$\theta$, and 
        ${\sim}$ is a predicate symbol from the set $\{{<},{\leq},=,\equiv_d : d \geq 1\}$.
        Line~\ref{line:elimmaxvar-tau} ``unpacks'' the term $\tau$, 
        according to the definition of quotient term given in~\Cref{sec:procedure}, as 
        \begin{equation}
            \label{eq:elimmaxvar-extended}
            a \cdot 2^{x} + f(\vec x') \cdot 2^{y} + \rho(\vec x\setminus x , \vec z'),
        \end{equation}
        where $2^x$ is the leading exponential term of 
        the ordering $\theta$ and $2^y$ is its successor in this ordering
        (observe that this agrees with the delayed substitution $[\xquot \cdot 2^y + \xrema/x]$).
        In the expression in~\eqref{eq:elimmaxvar-extended}, 
        $a$ is an integer, $f(\vec x')$ is a linear term over the quotient variables $\vec x'$, and
        $\rho$ is the least significant part of $\tau$. 
        The latter means that $\rho$ is of the form
        \begin{equation}
            \label{eq:elimmaxvar-least-significant}
            b\cdot y+\sum\nolimits_{i=1}^\ell \Big(a_i \cdot x_i + c_i \cdot (x_i \bmod 2^y) + \sum\nolimits_{j=1}^m\big(b_j \cdot 2^{x_j} + c_{i,j} \cdot (x_i \bmod 2^{x_j})\big)\Big) + d,
        \end{equation}
        where the coefficients $b,a_i,c_i,b_j,c_{i,j}$ and the constant $d$ are all integers; $m \leq \ell$,  
        the variables $x_1,\dots,x_m$ are from $\vec x\setminus\{x,y\}$, and the variables 
        $x_{m+1},\dots,x_\ell$ are from $\vec z'$.
        (The notation $\vec x\setminus\{x,y\}$ is short for $(\vec x \setminus x) \setminus y$.) 
        Finally, since $\varphi$ is a quotient system induced by~$\theta$, 
        it features the inequalities $0 \leq z < 2^{y}$ for every $z$ in $\vec z'$.

        We divide the proof into three cases, following which of the three branches of the chain of \textbf{if-else} statements starting in line~\ref{line:elimmaxvar-psi-trivial} triggers. 
        \begin{description}
            \item[The guard of the \textbf{if} statement in 
                line~\ref{line:elimmaxvar-psi-trivial} triggers.] 
                In this case, $(\tau\sim0)$ only features variables from $\vec x \setminus x$ and $\vec z'$, 
                and the iteration of the \textbf{foreach} loop on $\tau \sim 0$ completes yielding as a 
                \textbf{global output} a set with only one pair of systems: $(\gamma, \psi \land (\tau\sim0))$.
                Properties~\ref{item:claim-elimmaxvar-3-1}--\ref{item:claim-elimmaxvar-3-4} in the statement are trivially satisfied, and the lemma is proven.
            \item[The \textbf{else-if} statement in 
                line~\ref{line:elimmaxvar-second-case} triggers.] 
                In this case, $\sim$ is a symbol from $\{=,\leq,<\}$
                and with respect to the expression in~(\ref{eq:elimmaxvar-extended}), either $a\neq0$ or $f(\vec x')$ is not an integer. Notice that then Property~\ref{item:claim-elimmaxvar-3-3} trivially holds, as the antecedent of the implication in this property is false. 
                Below we focus on Properties~\ref{item:claim-elimmaxvar-3-1},~\ref{item:claim-elimmaxvar-3-2}
                and~\ref{item:claim-elimmaxvar-3-4}.

                We remind the reader that, following~\Cref{remark:remove-Delta}, we are considering the version of~Step~\eqref{item:elimmaxvar-i} featuring lines~\ref{line:elimmaxvar-guess-rho} and~\ref{line:elimmaxvar-psi-inequality} in place of lines~\ref{line:elimmaxvar-delta-undefined}--\ref{line:elimmaxvar-delta-from}.
                Therefore, in this case the iteration of the \textbf{foreach} loop executes lines~\ref{line:elimmaxvar-guess-rho},~\ref{line:elimmaxvar-psi-inequality} and~\ref{line:elimmaxvar-if-less}--\ref{line:elimmaxvar-psi-equality}.
                
                Observe that, under the assumption that $\theta(\vec x) \land \bigwedge_{z \in\vec z'}(0 \leq z < 2^{y})$ holds, in the expression in~\eqref{eq:elimmaxvar-least-significant} all variables $x_i$ and terms $2^{x_j}$, with $i \in [1,\ell]$ and $j \in [1,m]$, are bounded by $2^y$, 
                which in turns implies $\rho\in\bigl[-\onenorm{\rho}\cdot2^y,\onenorm{\rho}\cdot 2^y\bigr]$.  
                We thus derive the following implication:
                \begin{equation}\label{eq:elimmaxvar-ceiling-disj}
                    \theta\land\bigwedge_{z \in\vec z'}(0 \leq z < 2^{y})\implies -\onenorm{\rho}\cdot2^y \leq \rho \leq \onenorm{\rho}\cdot2^y.
                \end{equation}
                \Cref{algo:elimmaxvar} takes advantage of this implication to estimate the least significant part~$\rho$. 
                In line~\ref{line:elimmaxvar-guess-rho}, it guesses an integer $r\in\bigl[-\onenorm{\rho},\onenorm{\rho}\bigr]$, and 
                in line~\ref{line:elimmaxvar-psi-inequality}, it adds to $\psi$ the formula
                \begin{equation}\label{eq:elimmaxvar-ceiling}
                    \psi_r' \coloneqq ((r-1)\cdot2^y<\rho)\land(\rho\leq r\cdot2^y).                
                \end{equation}
                Essentially, in adding $\psi_r'$ to $\psi$, the algorithm is guessing that $r = \ceil{\frac{\rho}{2^y}}$.

                We now inspect lines~\ref{line:elimmaxvar-if-less}--\ref{line:elimmaxvar-psi-equality}, carefully analysing the three cases of ${\sim} \in \{=,\leq,<\}$ separately. 

                \vspace{2pt}
                \textbf{\textit{case: $=$.}} Let $R \coloneqq \bigl[-\onenorm{\rho},\onenorm{\rho}\bigr]$. Given $r \in R$, let us define
                \begin{align*}
                    \gamma_r&\coloneqq (a\cdot u+f(\vec x')+r=0),
                    \\
                    \psi_r&\coloneqq \psi_r' \land (r\cdot2^y=\rho).
                \end{align*}
                Following lines~\ref{line:elimmaxvar-psi-inequality},~\ref{line:elimmaxvar-gamma-inequality} and~\ref{line:elimmaxvar-psi-equality},
                we deduce that executing the \textbf{foreach} body on the state $(\tau\sim0,\gamma,\psi)$ yields as a \textbf{global output} the set of pairs $\{(\gamma \land \gamma_r, \psi \land \psi_r)\}_{r\in R}$. 
                Properties~\ref{item:claim-elimmaxvar-3-1},~\ref{item:claim-elimmaxvar-3-2}
                and~\ref{item:claim-elimmaxvar-3-4} 
                are easily seen to be satisfied:
                \begin{itemize}
                    \item Properties~\ref{item:claim-elimmaxvar-3-1} and~\ref{item:claim-elimmaxvar-3-2} trivially follow from the definitions of $\gamma_r$ and $\psi_r$. 
                    \item Observe that the expression in~\eqref{eq:elimmaxvar-extended} can be rewritten as $(a \cdot 2^{x-y} + f(\vec x')) \cdot 2^{y} + \rho(\vec x\setminus x , \vec z')$. From the formula~\eqref{eq:elimmaxvar-ceiling-disj} together with Equivalence~\ref{step1:fund-equiv-1} from~\Cref{lemma:split:inequalities}, we obtain 
                    \[ 
                    \theta\land\bigwedge_{z \in\vec z'}(0 \leq z < 2^{y})\implies
                    \big( \tau = 0 \iff\bigvee_{r= -\onenorm{\rho}}^{\onenorm{\rho}}(
                        a \cdot 2^{x-y} + f(\vec x') + r = 0 \land  \rho = r \cdot 2^y)\big).
                    \]
                    The subformula $a \cdot 2^{x-y} + f(\vec x') + r = 0$ 
                    is equal to $\gamma_r\sub{2^{x-y}}{u}$. 
                    The subformula $\rho = r \cdot 2^y$ is equivalent to $\psi_r$. We thus have
                    \[ 
                        \theta\land\bigwedge_{z \in\vec z'}(0 \leq z < 2^{y})\implies
                        \big( \tau = 0 \iff\bigvee_{r \in R}(\gamma_r[2^{x-y}/u]\land\psi_r)\big),
                    \]
                    that is, 
                    Property~\ref{item:claim-elimmaxvar-3-4} holds.
                \end{itemize}
                \vspace{2pt}
                \textbf{\textit{case: $\leq$.}} Let $R \coloneqq \bigl[-\onenorm{\rho},\onenorm{\rho}\bigr]$. Given $r \in R$, let us define
                \begin{align*}
                    \gamma_r&\coloneqq (a\cdot u+f(\vec x')+r \leq 0),\\
                    \psi_r &\coloneqq \psi_r'.
                \end{align*}
                Following lines~\ref{line:elimmaxvar-psi-inequality} and~\ref{line:elimmaxvar-gamma-inequality},
                we deduce that executing the \textbf{foreach} body on the state ${(\tau\sim0,\gamma,\psi)}$ yields as a \textbf{global output} the set of pairs $\{(\gamma \land \gamma_r, \psi \land \psi_r)\}_{r\in R}$. 
                The proof that Properties~\ref{item:claim-elimmaxvar-3-1},~\ref{item:claim-elimmaxvar-3-2}
                and~\ref{item:claim-elimmaxvar-3-4} 
                are satisfied follows as in the previous case
                (relying on Equivalence~\ref{step1:fund-equiv-2} from~\Cref{lemma:split:inequalities} to prove Property~\ref{item:claim-elimmaxvar-3-4}).

                \vspace{2pt}
                \textbf{\textit{case: $<$.}} 
                Let $R \coloneqq \bigl[-\onenorm{\rho},\onenorm{\rho}\bigr] \times \{=,<\}$.
                Given $r \in R$, we define
                \begin{align*}
                    \gamma_{(r,=)}&\coloneqq (a\cdot u+f(\vec x')+r+1 \leq 0),\\
                    \psi_{(r,=)}&\coloneqq \psi_r' \land (\rho = r \cdot 2^y),\\
                    \gamma_{(r,<)}&\coloneqq (a\cdot u+f(\vec x')+r \leq 0),\\
                    \psi_{(r,<)}&\coloneqq \psi_r' \land (\rho < r \cdot 2^y).
                \end{align*}

                Following lines~\ref{line:elimmaxvar-psi-inequality} and~\ref{line:elimmaxvar-if-less}--\ref{line:elimmaxvar-gamma-inequality},
                we deduce that executing \textbf{foreach} body on the state $(\tau\sim0,\gamma,\psi)$ yields as a \textbf{global output} the set of pairs 
                    $\{(\gamma \land \gamma_{(r,\sim')}, \psi \land \psi_{(r,\sim')})\}_{(r,\sim')\in R}$.
                The proof of
                Properties~\ref{item:claim-elimmaxvar-3-1} and~\ref{item:claim-elimmaxvar-3-2} is trivial.
                For the proof of Property~\ref{item:claim-elimmaxvar-3-4}, 
                from the formula~\eqref{eq:elimmaxvar-ceiling-disj} and 
                Equivalence~\ref{step1:fund-equiv-3} from~\Cref{lemma:split:inequalities} we have 
                \begin{align*}  
                    \theta\land\bigwedge_{z \in\vec z'}(0 \leq z < 2^{y})&\implies{}
                    \Big(\tau < 0 
                    \iff \\
                    &\bigvee_{r = -\norm{\rho}_1}^{\norm{\rho}_1}\!\!\big(a \cdot 2^{x-y} + f(\vec x') + r + 1 \leq 0 \land  \rho = r \cdot 2^y\big)\\ 
                    \lor &\bigvee_{r = -\norm{\rho}_1}^{\norm{\rho}_1}\!\!\big(a \cdot 2^{x-y} + f(\vec x') + r \leq 0 \land (r-1) \cdot 2^y < \rho < r \cdot 2^y\big)\!\Big).
                \end{align*}
                The subformulae $a \cdot 2^{x-y} + f(\vec x') + r + 1 \leq 0$ and $a \cdot 2^{x-y} + f(\vec x') + r \leq 0$
                are equal to $\gamma_{(r,=)}\sub{2^{x-y}}{u}$ 
                and $\gamma_{(r,<)}\sub{2^{x-y}}{u}$, respectively. 
                The subformulae $\rho = r \cdot 2^y$ and $(r-1) \cdot 2^y < \rho < r \cdot 2^y$ are equivalent to $\psi_{(r,=)}$ and $\psi_{(r,<)}$, respectively. We thus have 
                \begin{align*}  
                    \theta\land\bigwedge_{z \in\vec z'}(0 \leq z < 2^{y})\implies{}
                    \Big(&\tau < 0 
                    \iff \bigvee_{(r,\sim') \in R} \gamma_{(r,\sim')}\sub{2^{x-y}}{u} \land \psi_{(r,\sim')}\Big),
                \end{align*}
                that is, Property~\ref{item:claim-elimmaxvar-3-4} holds.

            \item[The \textbf{else} statement of in line~\ref{line:elimmaxvar-third-case} triggers.]
                In this case, $\sim$ is $\equiv_d$ for some $d \geq 1$, and
                the algorithm executes lines \ref{line:elimmaxvar-guess-mod}--\ref{line:elimmaxvar-psi-divisibility}.
                Let $R \coloneqq \bigl[1,\fmod(\phi)\bigr]$.
                In line~\ref{line:elimmaxvar-guess-mod}, it guesses an integer $r\in R$.
                Recall that $\fmod(\phi)$ is the least common multiple of all divisors appearing 
                in divisibility constraints of the system $\varphi$, 
                and therefore $d \leq \fmod(\phi)$.
                For an integer $r \in R$, we define  
                \begin{align*}
                    \gamma_r&\coloneqq (a\cdot u+f(\vec x')-r\equiv_d0), &\text{(see line~\ref{line:elimmaxvar-gamma-divisibility})}
                    \\
                    \psi_r&\coloneqq (r\cdot2^y+\rho\equiv_d0).
                    &\text{(see line~\ref{line:elimmaxvar-psi-divisibility})}
                \end{align*}
                Following lines~\ref{line:elimmaxvar-gamma-divisibility}
                and~\ref{line:elimmaxvar-psi-divisibility}, 
                we deduce that executing \textbf{foreach} body on the state $(\tau\sim0,\gamma,\psi)$ yields as a \textbf{global output} 
                the set of pairs $\{(\gamma \land \gamma_r, \psi \land \psi_r)\}_{r\in R}$. 
                Properties~\ref{item:claim-elimmaxvar-3-1}--\ref{item:claim-elimmaxvar-3-4} are again satisfied: 
                \begin{itemize}
                    \item Properties~\ref{item:claim-elimmaxvar-3-1} and~\ref{item:claim-elimmaxvar-3-2} follow directly by definition of $\gamma_r$ and $\psi_r$. 
                    \item Property~\ref{item:claim-elimmaxvar-3-3} is true, as the antecedent of the implication in this property is false (as in the previous case, we have either $a \neq 0$ or $f(\vec x')$ non-constant). 
                    \item Property~\ref{item:claim-elimmaxvar-3-4} follows by~\Cref{lemma:split:congruences}, since $d \leq \fmod(\phi)$ and $\theta$ implies $(x - y) \in \N$.
                    \qedhere
                \end{itemize}
                %
                \qedhere 
        \end{description} 
    \end{proof}
    
    \subparagraph*{Analysis of the complete execution of the foreach loop.}

        We extend the analysis performed in \Cref{claim:elimmaxvar-3} to multiple iterations of the body of the \textbf{foreach} loop. 
        We define the \textbf{global output} of the \textbf{foreach} loop of line~\ref{line:elimmaxvar-foreach}
        to be the set of all pairs $(\gamma,\psi)$ 
        where $\gamma$ and $\psi$ are the systems stored in
        the homonymous variables when, in a non-deterministic branch of the program,  line~\ref{line:elimmaxvar-gauss-1} is reached (and before this line is executed).
        We prove the following lemma.

        \begin{lemma}\label{claim:elimmaxvar-1}
            Let $(\theta(\vec x),\phi(\vec x, \vec x', \vec z'),[\xquot \cdot 2^y + \xrema/x])$ be an input of Step~(\ref{item:elimmaxvar-i}), as 
            described in Specification~\ref{spec:step1}.
            Let $u$ be the fresh variable defined in line~\ref{line:elimmaxvar-u}.
            Executing  the \textbf{foreach} loop of line~\ref{line:elimmaxvar-foreach} on this input 
            yields as a \textbf{global output} a set of pairs $\{(\gamma_i,\psi_i)\}_{i\in I}$ such that
            \begin{enumerate}[A.]
                \item\label{claim:elimmaxvar-1:propA} $\gamma_i$ is a linear system over the variables $\vec x'$ and $u$.
                \item\label{claim:elimmaxvar-1:propB} $\psi_i$ is a linear-exponential system over the variables $\vec x \setminus x$ and $\vec z'$. Moreover,  $\psi_i$ contains inequalities $0 \leq z < 2^y$, for every $z$ in $\vec z'$.
                \item\label{claim:elimmaxvar-1:propC} Over $\N$, $\theta \land \phi$ is equivalent to $\theta \land \exists u \big( u = 2^{x-y} \land \bigvee_{i \in I}(\gamma_i \land \psi_i)\big)$. 
            \end{enumerate}
        \end{lemma}
        \begin{proof}
            This lemma follows by first applying~\Cref{claim:elimmaxvar-3}, and then arguing that every formula $\psi_i$ contains the inequalities $0 \leq z < 2^y$, for every $z$ in $\vec z'$. Roughly speaking, this allows to remove the hypothesis $\bigwedge_{z \in\vec z'}(0 \leq z < 2^{y})$
            from Property~\ref{item:claim-elimmaxvar-3-4} in~\Cref{claim:elimmaxvar-3}, resulting in the equivalence 
            required by Property~\ref{claim:elimmaxvar-1:propC}.

            Let us formalise the above sketch.
            For simplicity, let $\phi = \bigwedge_{i=1}^m \tau_i \sim_i 0$, and assume that the guard of the \textbf{foreach} loop 
            selects the constraints in $\phi$ in the order $\tau_1 \sim_1 0, \tau_2 \sim_2 0, \dots$
            We denote by $\gamma^{(k)}$ and $\psi^{(k)}$ two systems that, in a single non-deterministic branch of the program that executes the body of the \textbf{foreach} loop exactly $k$ times, are stored in the variables $\gamma$ and $\psi$, respectively; 
            and denote by $\{(\gamma_t^{(k)},\psi_t^{(k)})\}_{t \in T_k}$ the set of all such pairs of systems, across all non-deterministic branches.
            Note that, from line~\ref{line:elimmaxvar-initialize}, we have $\gamma^{(0)} = \psi^{(0)} = \top$ (and $T_0$ contains a single index).
            Lastly, let $\phi^{(k)} \coloneqq \bigwedge_{i = k+1}^m \tau_i \sim_i 0$ (hence, $\phi^{(m)} = \top$). 
            
            By relying on~\Cref{claim:elimmaxvar-3},
            we conclude that the following is an invariant for the
            \textbf{foreach} loop: 
            for every $k \in [0,m]$ and $t \in T_k$,
            \begin{romanenumerate} 
                \item\label{foreach-loop-inv-1} $\gamma_t^{(k)}$ is a linear (in)equality or divisibility constraints over the variables $\vec x'$ and $u$,
                \item\label{foreach-loop-inv-2} $\psi_t^{(k)}$ is a linear-exponential constraint over the variables $\vec x\setminus x$ and $\vec z'$,
                \item\label{foreach-loop-inv-3} if $\tau_k \sim_k 0$ is $0 \leq z$ (resp.~$z < 2^y$) for some $z$ in $\vec z'$, then $\psi_t^{(k)}$ contains $0 \leq z$ (resp.~$z < 2^y$).
                Here, recall that $0 \leq z$ and $z < 2^y$ are shorthands for $-z \leq 0$ and $z - 2^y < 0$, respectively.%
                \item\label{foreach-loop-inv-4} Over $\N$, $\theta\land\bigwedge_{z \in\vec z'}(0 \leq z < 2^{y})\implies
                \big( \phi \iff\bigvee_{t \in T_k}(\phi^{(k)} \land \gamma_t^{(k)}[2^{x-y}/u]\land\psi_t^{(k)})\big)$.
            \end{romanenumerate}

            Consider the case of $k = m$. 
            Since $\phi$ is a quotient system induced by $\theta$,  
            it contains $\bigwedge_{z \in\vec z'}(0 \leq z < 2^{y})$ as a subsystem. Hence, from Item~\eqref{foreach-loop-inv-3} of the invariant, for every $t \in T_m$, 
            $\psi_t^{(m)}$ contains $\bigwedge_{z \in\vec z'}(0 \leq z < 2^{y})$ as a subsystem. Together with Items~\eqref{foreach-loop-inv-1} and~\eqref{foreach-loop-inv-2}, 
            this shows that Properties~\ref{claim:elimmaxvar-1:propA} and~\ref{claim:elimmaxvar-1:propB} hold. 
            By Item~\eqref{foreach-loop-inv-4}, we also have the following equivalence: 
            \[ 
                (\theta \land \phi) \iff \big(\theta \land \bigvee_{t \in T_k}(\phi^{(m)} \land \gamma_t^{(m)}[2^{x-y}/u]\land\psi_t^{(m)})\big).
            \]
            Above, $\phi^{(m)} = \top$, and $\gamma_t^{(m)}[2^{x-y}/u]$ is equivalent to $\exists u (u = 2^{x-y} \land \gamma_t^m)$. Since $\theta$ implies $x \geq y$, the variable $u$ can be existentially quantified over $\N$.
            Hence, Property~\ref{claim:elimmaxvar-1:propC} holds.
        \end{proof}

    \subparagraph*{Incorporating the call to~\GaussianQE and completing the analysis of Step~(\ref{item:elimmaxvar-i}).}
        By chaining~\Cref{claim:elimmaxvar-1} and~\Cref{lemma:corr-gaussianqe}, 
        we are now able to prove the correctness of Step~(\ref{item:elimmaxvar-i}).    
        
        \begin{lemma}\label{lemma:corr-step1}
            Step~(\ref{item:elimmaxvar-i}) of~\Cref{algo:elimmaxvar} complies with Specification~\ref{spec:step1}.
        \end{lemma}
        \begin{proof}
        The input of this step corresponds to the input of~\Cref{algo:elimmaxvar}, that is a triple ${(\theta(\vec x),\varphi(\vec x, \vec x', \vec z'),[\xquot \cdot 2^y + \xrema/x])}$ satisfying the properties described in Specification~\ref{spec:step1}.

        By~\Cref{claim:elimmaxvar-1}, the \textbf{global output} of the \textbf{foreach} loop of line~\ref{line:elimmaxvar-foreach} 
        corresponding to this input is a set of pairs of 
        systems $\{(\gamma_j,\psi_j)\}_{j\in J}$
        satisfying Properties~\ref{claim:elimmaxvar-1:propA}--\ref{claim:elimmaxvar-1:propC}.
        Hence, every formula $\gamma_j$ 
        is a linear system over the variables $\vec x'$ and $u$, where $u$ is the fresh variable
        defined in line~\ref{line:elimmaxvar-u}.
        Let us fix some $j \in J$, and consider the non-deterministic branch in which the \textbf{foreach} loop produces the pair $(\gamma_j,\psi_j)$.
        In line~\ref{line:elimmaxvar-gauss-1}, 
        the algorithm calls~\Cref{algo:gaussianqe} (\GaussianQE) to remove the variables $\vec q \coloneqq \vec x' \setminus x'$ from the formula $\gamma_j \land \vec x' \geq 0$.

        By~\Cref{lemma:corr-gaussianqe}, the \textbf{global output} of~\Cref{algo:gaussianqe}  on 
        input $(\quot,\gamma_j\land \vec x'\geq 0)$
        is a set of linear systems $\{\gamma_{j,k}(u,\xquot)\}_{k \in K_j}$ such that
        $\bigvee_{k\in K_j} \gamma_{j,k}(u,\xquot)$ is equivalent to $\exists \vec q (\gamma_j \land \vec x' \geq 0)$ over~$\Z$. 
        Because of the inequalities $\vec x' \geq 0$, the quantification $\exists \vec q$ can be restricted to the non-negative integers, 
        and therefore we conclude that, over $\N$,
        \begin{equation} 
            \label{lemma:step1correct:gaussequiv}
            \exists \vec q \gamma_j
            \iff 
            \bigvee_{k\in K_j} \gamma_{j,k}(u,\xquot).
        \end{equation}
        The \textbf{global output} of Step~(\ref{item:elimmaxvar-i}) is the set:
        \[
            \bigl\{(\gamma_i,\psi_i)\,:\, i\in I\bigr\} \coloneqq \bigl\{(\gamma_{j,k},\psi_j)\,:\, j \in J, k \in K_j \bigr\}.
        \]    
        We show that this set satisfies the requirements of Specification~\ref{spec:step1}.
        \begin{itemize}
            \item Obviously, every $\gamma_i$ a linear system in variables $u$, and $x'$, as it corresponds to some $\gamma_{j,k}$. This system contains the inequality $x' \geq 0$. Indeed, this inequality is present in the system~$\gamma_j\land \vec x'\geq 0$ in input of~\Cref{algo:gaussianqe}. Looking at its pseudocode, observe that~\Cref{algo:gaussianqe} leaves untouched all inequalities that do not feature variables that are to be eliminated. 
            Hence, since $x'$ is not among the eliminated variables~$\vec q$, the output formula $\gamma_{j,k}$ contains $x' \geq 0$.  
            \item By Property~\ref{claim:elimmaxvar-1:propB} of~\Cref{claim:elimmaxvar-1}, every  $\psi_i$ is a linear-exponential system containing the inequalities $0\leq z <2^y$, for every $z$ in $\vec z'$.
            \item By Property~\ref{claim:elimmaxvar-1:propC} of~\Cref{claim:elimmaxvar-1} and the equivalence~\eqref{lemma:step1correct:gaussequiv}, 
            over $\N$ we have,
        \end{itemize}
            \begin{align*} 
                &\exists \vec q \big(\theta \land \phi \land x = x' \cdot 2^y + z'\big)\\
                \iff{}& 
                \exists \vec q \big(\theta \land \exists u \big( u = 2^{x-y} \land \bigvee_{j \in J}(\gamma_j \land \psi_j)\big) \land x = x' \cdot 2^y + z'\big)
                &\text{by Property~\ref{claim:elimmaxvar-1:propC}}\\
                \iff{}& 
                \bigvee_{j \in J}
                \big( 
                \theta \land \exists u \big( u = 2^{x-y} \land (\exists \vec q \gamma_j) \land \psi_j) \land x = x' \cdot 2^y + z'
                \big) &\text{$\vec q$ only occurs in $\gamma_j$}\\
                \iff{}&\bigvee_{i \in I}
                \big( 
                \theta \land \exists u \big( u = 2^{x-y} \land \gamma_i \land \psi_i) \land x = x' \cdot 2^y + z'
                \big) &\text{by def.~of~$I$ and~\eqref{lemma:step1correct:gaussequiv}}\\
                \iff{}&\bigvee_{i \in I}
                \big( 
                \theta \land \exists u \big( u = 2^{x-y} \land \gamma_i) \land \psi_i \land x = x' \cdot 2^y + z'
                \big) &\text{$u$ does not occur in $\psi_i$}.
                &\qedhere
            \end{align*}
    \end{proof}

\subsubsection*{Correctness of Step~(\ref{item:elimmaxvar-ii})}

\begin{lemma}\label{lemma:corr-step2}
    Step~(\ref{item:elimmaxvar-ii}) of~\Cref{algo:elimmaxvar} complies with Specification~\ref{spec:step2}.
\end{lemma}
\begin{proof}
    To prove the lemma, we analyse Step~(\ref{item:elimmaxvar-ii}).
    This step takes as input the ordering $\theta(\vec x)$
    and the delayed substitution $[\xquot \cdot 2^y + \xrema/x]$ that are part of the input of~\Cref{algo:elimmaxvar} (where $x$ and $y$ in the delayed substitution are the largest and second-to-largest variables with respect to $\theta$), 
    together with a branch output of Step~\eqref{item:elimmaxvar-i}.
    By~\Cref{lemma:corr-step1} and according to~Specification~\ref{spec:step1}, 
    the latter is a pair $(\gamma,\psi)$ where $\gamma(u,x')$ is a linear system and $\psi(\vec x \setminus x, \vec z')$ 
    is a linear-exponential system.


    Step~(\ref{item:elimmaxvar-ii}) starts with the replacement of the auxiliary variable $u$. 
    However, in line~\ref{line:elimmaxvar-2-u} we do not only perform the replacement of $u$ with $2^{x-y}$, 
    but immediately replace $(x-y)$ with $u$. 
    That is, the formula $\exists u (\gamma(u,x') \land u=2^{x-y})$
    appearing in $\Psi$ from Specification~\ref{spec:step2} is updated into the \emph{equivalent} formula 
    $\exists u (\gamma(2^u,x') \land u = x - y)$.
    The system $\gamma(2^u,x')$ is a $(u,\xquot)$-primitive linear-exponential system. 

    After this ``change of alias'' for the variable $u$, 
    the algorithm proceeds with linearising its 
    occurrences in $\gamma(2^u,x')$.
    This is done by invoking \Cref{algo:linearize} (\textsc{SolvePrimitive})
    on input $(u,\xquot,\gamma(2^u,\xquot))$; see line~\ref{line:elimmaxvar-solve-primitive}. 
    By correctness of \Cref{algo:linearize}
    (\Cref{lemma:corr-linearize}), its \textbf{global output} is a set of pairs $\{(\hat\chi_{j}(u),\gamma_{j}(\xquot))\}_{j\in J}$, 
    where every $\hat\chi_{j}$ and $\gamma_{j}$ is a linear system, such that
    \begin{equation}\label{eq:elimmaxvar-solve-primitive}
        (u\geq \xquot\geq0)\implies\Big(\gamma(2^u,\xquot)\iff\bigvee_{j\in J}\hat\chi_{j}(u)\land\gamma_{j}(\xquot)\Big). 
    \end{equation}
    To use the double implication of~(\ref{eq:elimmaxvar-solve-primitive}), we show next that $\theta(\vec x)\land(x=\xquot \cdot 2^y+\xrema)\land(u=x-y)$
    entails $u\geq \xquot$ (recall that $\xquot$ ranges over~$\N$).
    
    When $\xquot = 0$, the inequality $u\geq \xquot$ follows from the inequality $2^x \geq 2^y$ appearing in the
    ordering $\theta(\vec x)$. If $\xquot$ is positive, then we have
    \begin{align*}  
        u=x-y&=\xquot\cdot2^y+\xrema-y
        &\text{delayed substitution}
        \\
        &\geq \xquot\cdot(y+1)+\xrema-y
        &\text{since } (2^y\geq y+1)\: \text{for every }y\in\N
        \\
        &= y\cdot(\xquot-1)+\xquot+\xrema\geq \xquot.
        &\text{since } \xquot \geq 1.
    \end{align*}
    Therefore, from the formula~\eqref{eq:elimmaxvar-solve-primitive} we obtain the equivalence 
    \begin{align} 
      \Psi(\vec x,\xquot \xrema)
      & 
      \iff \theta(\vec x)\land(x=\xquot\cdot2^y+\xrema)\land  \exists u \big(\big(\bigvee_{j\in J}\hat\chi_{j}(u)\land\gamma_{j}(\xquot)\big) \land (u={x-y})\big) 
      \label{eq:elimmaxvar-solve-primitive:eq2}  
    \end{align}

    Observe that, from the specification of \Cref{algo:linearize},
    every system $\hat\chi_{j}(u)$ has a very simple form: 
    it is either an equality $(u=a)$ or 
    a conjunction $(u\geq b) \land (d \mid u-r)$, for some $a,b,d,r\in\N$, with $b > 2$. We will use this fact twice in the remaining part of the proof.


    In line~\ref{line:elimmaxvar-delayed-subs}, the algorithm performs on $\hat\chi_{j}(u)$ the substitutions $\sub{x-y}{u}$ and $[\xquot \cdot 2^y + \xrema/x]$, in this order.
    For the moment, let us focus on the effects of the first substitution. Because of the form of $\hat\chi_{j}$, 
    the system $\hat\chi_{j}\sub{x-y}{u}$ entails $x \geq y$.
    This allows us to exclude $2^x \geq 2^y$ from the ordering $\theta(\vec x)$; or alternatively quantifying it away as $(\exists x \theta(\vec x))$. Starting from the equivalence~\eqref{eq:elimmaxvar-solve-primitive:eq2}, we thus obtain 
    \begin{align} 
      \Psi(\vec x,\xquot \xrema)
      & 
      \iff (\exists x\theta(\vec x)) \land (x=\xquot\cdot2^y+\xrema)\land \bigvee_{j\in J}\hat\chi_{j}\sub{x-y}{u}\land\gamma_{j}(\xquot), 
      \label{eq:elimmaxvar-solve-primitive:eq3}  
    \end{align}
    where we highlight the fact that $x$ still occurs free in both sides of the equivalence.

    We now consider the application of the second substitution $[\xquot \cdot 2^y + \xrema/x]$. 
    For every $j\in J$, define $\chi_j(y,x',z') \coloneqq \hat\chi_j\sub{x-y}{u}[\xquot \cdot 2^y + \xrema/x]$, i.e.,
    the system $\hat\chi_{j}(\xquot \cdot 2^y + \xrema-y)$.
    As~Step~(\ref{item:elimmaxvar-ii}) ends in line~\ref{line:elimmaxvar-delayed-subs}, 
    its~\textbf{global output} is the set of triples 
    $\{(\chi_{j}(y,\xquot,\xrema),\gamma_{j}(\xquot),\psi(\vec x\setminus x,\vec z'))\}_{j\in J}$. 
    It is easy to see that this set satisfies \Cref{item:elimmaxvar-chi-1,item:elimmaxvar-chi-2,item:elimmaxvar-chi-3} in Specification~\ref{spec:step2}: 
    \begin{itemize}
      \item \Cref{item:elimmaxvar-chi-1} follows from the form of $\hat\chi_{j}(u)$ and the substitutions applied to it.
      
      \item The first statement of~\Cref{item:elimmaxvar-chi-2} follows from the specification of \Cref{algo:linearize}. For the second statement, observe that by~\Cref{lemma:corr-step1} and according to~Specification~\ref{spec:step1}, the formula $\gamma(2^u,x')$ contains the inequality $x' \geq 0$. 
      Let us study the evolution of this inequality through~\Cref{algo:linearize}. 
      In line~\ref{line:linearize-decompose}, this inequality is part of the formula $\psi$. 
      Then, following the updates in lines~\ref{line:linearize-chi-equality} and~\ref{line:linearize-chi-div}, 
      $x' \geq 0$ appears in the formula $\gamma$ in output of~\Cref{algo:linearize}.
      Therefore, for every $j \in J$, $x' \geq 0$ appears in $\gamma_j$, as required.
      
      \item \Cref{item:elimmaxvar-chi-3} is direct from the fact that Step~(\ref{item:elimmaxvar-ii}) does not manipulate $\psi$.
    \end{itemize}
    Lastly, from the equivalence~\eqref{eq:elimmaxvar-solve-primitive:eq3} and the definition of $\chi_j$, we establish the equivalence~\eqref{eq:elimmaxvar-psi}
    in the specification:
    \begin{align*} 
      &\exists x \Psi(\vec x,\xquot \xrema)\\ 
      \iff{}& \exists x \big((\exists x\theta(\vec x)) \land (x=\xquot\cdot2^y+\xrema)\land \bigvee_{j\in J}\hat\chi_{j}\sub{x-y}{u}\land\gamma_{j}(\xquot)\big)\\
      \iff{}& \bigvee_{j\in J} (\exists x\theta(\vec x)) \land \chi_{j}(y,x',z')\land\gamma_{j}(\xquot).
      \qedhere
    \end{align*}
\end{proof}

\subsubsection*{Correctness of Step~(\ref{item:elimmaxvar-iii})}

We start by giving a high-level overview of Step~(\ref{item:elimmaxvar-iii}). As discussed in the body of the paper, this step can be seen as a simplified version of Step~(\ref{item:elimmaxvar-i}). Recall that Step~(\ref{item:elimmaxvar-i}) manipulates the linear-exponential system $\varphi(\vec x,\vec x',\vec z')$ from the 
input of~\Cref{algo:elimmaxvar}. 
Step~(\ref{item:elimmaxvar-iii}) manipulates instead the formula~$\chi(y,\xquot,\xrema)$ that is part of the output of Step~(\ref{item:elimmaxvar-ii}). 
The similarity between these two manipulations is reflected in the diagram from~\Cref{fig: flowchart}. 
The systems $\gamma$ and $\psi$ in input of Step~(\ref{item:elimmaxvar-iii}) are denoted in the diagram by 
$\gamma_3$ and $\psi_1$, respectively.
As shown in the diagram, within Step~(\ref{item:elimmaxvar-iii}) the system $\chi(y,\xquot,\xrema)$, which is 
in fact a quotient system induced by the ordering $\theta(\vec x)$, is (non-deterministically) split into 
its most significant part $\gamma_4(\xquot)$ and least significant part $\psi_2(y,\xrema)$. 
The former is conjoined with the formula $\gamma_3(\xquot)$.
Since both systems are linear systems with a single variable~$\xquot$, by calling~\Cref{algo:gaussianqe} (\GaussianQE)
we can check whether $\gamma_3(\xquot)\land\gamma_4(\xquot)$ is satisfiable over $\N$. 
If the answer is negative, the computations in this non-deterministic branch do not contribute to the global output. (We recall in passing that an empty disjunction is equivalent to the formula $\bot$. So, if the \textbf{global output} of Step~(\ref{item:elimmaxvar-iii}) is the empty set, its input corresponds to an unsatisfiable formula.)
If the answer is positive, the step returns the linear-exponential system $\psi_1(\vec x\setminus x ,\vec z')\land\psi_2(y,\xrema )$.

\begin{lemma}\label{lemma:corr-step3}
    Step~(\ref{item:elimmaxvar-iii}) of~\Cref{algo:elimmaxvar} complies with Specification~\ref{spec:step3}.
\end{lemma}
\begin{proof}
    This step takes as input the ordering
    $\theta(\vec x)$ that is part of the  
    input of~\Cref{algo:elimmaxvar}, together with a branch output of
    Step~(\ref{item:elimmaxvar-ii}). By~\Cref{lemma:corr-step2} and according to
    Specification~\ref{spec:step2}, the latter is a triple
    $(\chi(y,\xquot,\xrema),\gamma(\xquot),\psi(\vec x\setminus x ,\vec z'))$
    satisfying~\Cref{item:elimmaxvar-chi-1,item:elimmaxvar-chi-2,item:elimmaxvar-chi-3}
    from Specification~\ref{spec:step2}.

    To prove the lemma, consider first lines~\ref{line:elimmaxvar-x1-start}--\ref{line:elimmaxvar-psi-inequality-x1} (that is, Step~(\ref{item:elimmaxvar-iii}) without the call to~\Cref{algo:gaussianqe} performed in line~\ref{line:elimmaxvar-gauss-2}). 
    Define the \textbf{global output} of these lines as the union over all non-deterministic branches 
    of the pairs $(\gamma,\psi)$, where $\gamma$ and $\psi$ are the contents of the homonymous variables 
    when line~\ref{line:elimmaxvar-gauss-2} is reached (before this line is executed).
    Let us for the moment assume the following result:
    \begin{claim}
        \label{claim:elimmaxvar:stepiii}
        Let $(\chi(y,\xquot,\xrema),\gamma(\xquot),\psi(\vec x\setminus x ,\vec z'))$ 
        be an input of Step~(\ref{item:elimmaxvar-iii}), as described in Specification~\ref{spec:step3}. 
        Then, the \textbf{global output} of 
        lines~\ref{line:elimmaxvar-x1-start}--\ref{line:elimmaxvar-psi-inequality-x1}
        is a set of pairs $\{(\gamma_j \land \gamma ,\psi_j \land \psi)\}_{j\in J}$ such that, 
        for every $j \in J$ we have
        \begin{enumerate}[A.]
            \item $\gamma_j$ is a linear system in the single variable $\xquot$.\label{item:elimmaxvar-claim2-1} 
            \item $\psi_j$ is a linear-exponential system over the variables $(y,\xrema)$.\label{item:elimmaxvar-claim2-2}
            \item Over \N, $(0\leq \xrema<2^y)$ entails that $\chi(y,\xquot,\xrema)$ is equivalent 
            to $\bigvee_{j\in J}\gamma_j(\xquot)\land\psi_j(y,\xrema)$.\label{item:elimmaxvar-claim2-3}
        \end{enumerate}
    \end{claim}
    
    With the above claim at hand, it is simple to complete the proof. 
    Let $K \subseteq J$ be the subset of indices $j \in J$ such that $\gamma_j \land \gamma$ is satisfiable (over~$\N$). 
    Observe that, by~\Cref{item:elimmaxvar-claim2-1} and~\Cref{lemma:corr-step2}, $\gamma_j \land \gamma$ is a linear system, and moreover $\gamma(x')$ contains the inequality $x' \geq 0$ (this is ensured by~Specification~\ref{spec:step2}).
    Therefore, despite working over~$\Z$, \Cref{algo:gaussianqe} can be used to perform this satisfiability check over $\N$ (as it is done in line~\ref{line:elimmaxvar-gauss-2}). 
    
    The \textbf{global output} of Step~(\ref{item:elimmaxvar-iii}) 
    is the set $\{\psi_k(y,\xrema) \land \psi(\vec x \setminus x, \vec z') : k \in K\}$. 
    By~\Cref{item:elimmaxvar-claim2-2} and Specification~\ref{spec:step2}, this is a set of linear-exponential systems, as required by Specification~\ref{spec:step3}. 
    Lastly, formula~\eqref{eq:elimmaxvar-psi-x'} follows from the fact that $0 \leq z' < 2^y$ entails
    \begin{align*}  
        \exists \xquot (\chi(y,\xquot,\xrema) \land \gamma(x')) &\iff 
        \exists \xquot \bigvee_{j\in J}\gamma_j(\xquot)\land\psi_j(y,\xrema) \land \gamma(\xquot)
        &\text{by~\Cref{item:elimmaxvar-claim2-3}}\\
        &\iff \bigvee_{j\in J}\psi_j(y,\xrema) \land \exists \xquot \big(\gamma_j(\xquot)\land \gamma(\xquot)\big)\\
        &\iff \bigvee_{k \in K}\psi_k(y,\xrema)
        &\text{by def.~of~$K$.}
    \end{align*}
    This completes the proof~\Cref{lemma:corr-step3},
    subject to the proof of~\Cref{claim:elimmaxvar:stepiii}
    which is given below.
\end{proof}

\begin{claimproof}[Proof of~\Cref{claim:elimmaxvar:stepiii}.]
    The quotient system $\chi(y,\xquot,\xrema)$ induced by the ordering $\theta$ 
    may have one of the following two forms, which are handled in different lines of~Step~(\ref{item:elimmaxvar-iii}):
    \begin{enumerate}[(a)]
        \item lines~\ref{line:elimmaxvar-guess-c}--\ref{line:elimmaxvar-psi-equality-x1}: $(\xquot\cdot2^y+\xrema-y-c=0)$ for $c\in\N$; or\label{item:elimmaxvar-claim2-equality}
        \item lines~\ref{line:elimmaxvar-x1-let}--\ref{line:elimmaxvar-psi-inequality-x1}: $(-\xquot\cdot2^y-\xrema+y+c\leq0)\land (d\mid \xquot\cdot2^y+\xrema-y-r)$ for $c,d,r\in\N$ and $c\geq 3$.
        \label{item:elimmaxvar-claim2-divisibility}
    \end{enumerate}
    Below, we consider 
    the cases~(\ref{item:elimmaxvar-claim2-equality}) and~(\ref{item:elimmaxvar-claim2-divisibility}) separately. 
    Beforehand, we estimate the term $(-\xrema+y+c)$ assuming that the variables $\xrema$ and $y$ are such that 
    $(0\leq \xrema<2^y)$. We have:
    \begin{align*}
        -2^y < -\xrema & \leq(-\xrema+y+c) &\text{since $y,c \geq 0$}\\ 
        & \leq y + c < (1+c) \cdot 2^y.
    \end{align*}
    Therefore, under the hypothesis that $0\leq \xrema<2^y$, 
    we have $(-\xrema+y+c) \in [0 \cdot 2^y,c \cdot 2^y]$, 
    which allows us to handle cases~(\ref{item:elimmaxvar-claim2-equality}) and~(\ref{item:elimmaxvar-claim2-divisibility}) by relying on~\Cref{lemma:split:inequalities}.

    \emph{Case (\ref{item:elimmaxvar-claim2-equality}).} 
    For $b \in [0,c]$, define $\gamma_b \coloneqq (x' = b)$ and 
    $\psi_b \coloneqq (b \cdot 2^y = -z' + y + c)$.
    Following the guess done in~line~\ref{line:elimmaxvar-guess-c},
    the \textbf{global output} of lines~\ref{line:elimmaxvar-x1-start}--\ref{line:elimmaxvar-psi-inequality-x1} is the set of pairs $\{(\gamma_b \land \gamma,\psi_b \land \psi)\}_{b\in[0,c]}$.
    Then, Items~\ref{item:elimmaxvar-claim2-1} and~\ref{item:elimmaxvar-claim2-2} of the claim
    are obviously satisfied. By~\Cref{lemma:split:inequalities}, we have 
    \[ 
        (0 \leq z' < 2^y) \implies 
        \left( \chi \iff \textstyle\bigvee_{b=0}^c (\gamma_b \land \psi_b)\right),
    \]
    which shows Item~\ref{item:elimmaxvar-claim2-3}.

    \emph{Case~(\ref{item:elimmaxvar-claim2-divisibility}).} In this case $\chi$ is a conjunction of 
    an inequality $(-\xquot\cdot2^y-\xrema+y+c\leq0)$ and a divisibility constraint $(d\mid \xquot\cdot2^y+\xrema-y-r)$. 
    Given $(b,g) \in [0,c] \times [1,d]$, define 
    \begin{align*}
        \gamma_{(b,g)} &\coloneqq x' \geq b \land (d \divides x' - g)\,,\\ 
        \psi_{(b,g)} &\coloneqq ((b-1)\cdot2^y<-z'+y+c)\land(-z'+y+c\leq b\cdot2^y)\land(d\mid g\cdot2^y+z'-y-r)\,.
    \end{align*}
    Following the guess done by the algorithm in line~\ref{line:elimmaxvar-guess-times}, 
    the \textbf{global output} of lines~\ref{line:elimmaxvar-x1-start}--\ref{line:elimmaxvar-psi-inequality-x1} is the set of pairs ${\{(\gamma_K \land \gamma,\psi_K \land \psi)\}_{K \in [0,c] \times [1,d]}}$. Then, Items~\ref{item:elimmaxvar-claim2-1} and~\ref{item:elimmaxvar-claim2-2} of the claim
    are obviously satisfied. By~\Cref{lemma:split:congruences} 
    we have  
    \begin{equation*}
        (d\mid \xquot\cdot2^y+\xrema-y-r)\iff\bigvee_{g=1}^d(d\mid \xquot-g)\land(d\mid g\cdot2^y+\xrema-y-r),
    \end{equation*}
    and by~\Cref{lemma:split:inequalities}, we have 
    \begin{equation*}
        \begin{aligned}
            (0\leq \xrema<2^y)\implies
            \Big(&-\xquot\cdot2^y-\xrema+y+c\leq0\iff
            \\
            &\bigvee_{b\in [0,c]}(\xquot \geq b)\land((b-1)\cdot2^y<-\xrema+y+c\leq b\cdot2^y)\Big).
        \end{aligned}
    \end{equation*}
    Therefore, we conclude that 
    \[ 
        (0 \leq z' < 2^y) \implies 
        \left( \chi \iff \textstyle\bigvee_{K \in [0,c] \times [1,d]} (\gamma_K \land \psi_K)\right),
    \]
    which shows Item~\ref{item:elimmaxvar-claim2-3}, completing the proof of the claim.
\end{claimproof}

The correctness of~\ElimMaxVar now follows by
combining~\Cref{lemma:corr-step1,lemma:corr-step2,lemma:corr-step3} with~\Cref{lemma:corr-elimmaxvar-cond}.

\begin{lemma}\label{lemma:corr-elimmaxvar}
    \Cref{algo:elimmaxvar} (\ElimMaxVar) complies with its specification.
\end{lemma}

\subsection{Correctness of Algorithm~\ref{algo:elimmaxvar} (\Master)}


\PropMasterCorrect*

\begin{proof}
  We prove that~\Cref{algo:master} (\Master) complies with its specification.

  Consider an input linear-exponential
  system $\phi(x_1,\dots,x_n)$ (with no divisibility constraints).
  Let $\Theta$ be the set of all orderings 
  $\theta(x_0,\dots,x_n) = {t_1 \geq t_2 \geq \dots
  \geq t_{n} \geq 2^{x_0} = 1}$ such that $(t_1,\dots,t_{n})$ is
  a permutation of the terms $2^{x_1},\dots,2^{x_n}$, and
  $x_0$ is a fresh variable used as a placeholder for $0$.
  The following equivalence is immediate: 
  \begin{equation*} 
    \phi \iff \bigvee_{\theta \in \Theta} \exists x_0 (\theta \land \phi).
  \end{equation*}
  Therefore, $\phi$ is satisfiable if and only if so is some $\exists x_0 (\theta \land \phi)$ with $\theta \in \Theta$.
  The algorithm starts by guessing such a $\theta$ (line~\ref{algo:master:guess-order}).
  Without loss of generality, let us
  assume for simplicity that the procedure guesses the
  ordering $\theta(x_0,\dots,x_n) = {2^{x_n} \geq 2^{x_{n-1}}
  \geq \dots \geq 2^{x_1} \geq 2^{x_0} = 1}$.  

  Throughout the proof, we write $\beta$ for a non-deterministic branch 
  of the procedure, represented as a list of line numbers of the algorithm decorated with guesses for lines featuring the non-deterministic \textbf{guess} statement.

  We remind the reader that we represent a non-deterministic branch~$\beta$ as a list of entries containing the name of the algorithm being executed, the line number being executed, and, for lines featuring the non-deterministic \textbf{guess} statement, the performed guess. We write $\beta = \beta_1\beta_2$ whenever $\beta$ can be decomposed on a prefix $\beta_1$ and suffix $\beta_2$.

  We define $B_i$ for the set of all non-deterministic branches of the procedure 
  ending with the entry $(\Master,\text{``line~\ref{algo:master:outer-loop}''})$ 
  after iterating  $i$ times the loop starting at line~\ref{algo:master:outer-loop} 
  (that is, the entry $(\Master,\text{``line~\ref{algo:master:outer-loop}''})$ appears $i+1$ times in the representation of the non-deterministic branch).
  Note that every $\beta \in B_i$ uniquely corresponds 
  to a system $\phi_\beta$ and an ordering $\theta_\beta$, 
  that are stored in the homonymous program variables. 
  Because of the above assumption on~$\theta$, 
  we can suppose $B_0$ to contain a single non-deterministic branch $\beta_0$ such that $\phi_{\beta_0} = \phi$, and $\theta_{\beta_0} = \theta$.
  
  We show that the loop of line~\ref{algo:master:outer-loop} enjoys the following loop invariant (where $i \in \N$, assuming that the loop executes at least $i$ times). 
  For every $\beta \in B_i$,  
  \begin{enumerate}
    \item\label{inv:it0-apx} the variables in $\phi_\beta$ are from either $\{x_0,\dots,x_{n-i}\}$ or from $\vec z_\beta$, where variables in $\vec z_\beta$ are different from $x_0,\dots,x_n$;
    \item\label{inv:it1-apx} all variables that occur
    exponentiated in $\phi_\beta$ are among $x_0,\dots,x_{n-i}$ (that is, the variables in $\vec z_\beta$ do not occur in exponentials);
    \item\label{inv:it3-apx} all variables $z$ belonging to~$\vec z_\beta$ 
    are such that $(z < 2^{x_{n-i}})$ is an inequality in
    $\phi_\beta$;
    \item\label{inv:it3half-apx} 
      the vector~$\vec z_\beta$ contains at most $i$ variables;
    \item\label{inv:it2-apx} $\theta_\beta$ is the ordering $
    {2^{x_{n-i}} \geq 2^{x_{n-i-1}} \geq \dots \geq 2^{x_1}
    \geq 2^{x_0} = 1}$.
    \item\label{inv:it4-apx} Moreover,  $\bigvee_{\beta \in B_i} \theta_{\beta} \land \phi_{\beta}$ is equisatisfiable with $\theta \land \phi$ over $\N$.
  \end{enumerate}
  First, observe that this invariant is trivially true for $B_0 = \{\beta_0\}$, for which $\vec z_{\beta_0}$ is the empty set of variables. 
  Second, observe that from~\Cref{inv:it2-apx}, for every $\beta_n \in B_n$, the ordering $\theta_{\beta_n}$ equals $2^{x_0} = 1$. 
  This causes the loop in line~\ref{algo:master:outer-loop}
  to exit after $n$ iterations. 
  Then, the definition of~$\theta_{\beta_n}$,
  together with~\Cref{inv:it1-apx,inv:it3-apx} of the
  invariant, force $x_0$ and all variables in $\vec z_{\beta_n}$ to be
  equal to zero, which in turn imply $\phi \land \theta$ to be
  equisatisfiable with $\phi_{\beta_n}(\vec 0)$, by~\Cref{inv:it4-apx}. 
  Therefore, in order to conclude the proof it suffices to show that the above invariant is indeed preserved at every iteration of the loop.

  Below, we assume the loop invariant to be true for some $(i-1) \in \N$, and that the loop is executed at least $i$ times.
  For $\beta \in B_{i-1}$, 
  we define $B_{i}^\beta \coloneqq \{\beta' \in B_{i} : \beta' = \beta\beta'' \text{ for some } \beta'' \}$, that is, 
  $B_{i}^{\beta}$ contains the non-deterministic branches of $B_{i}$ that are obtained by executing the body of the loop in line~\ref{algo:master:outer-loop} once, starting from $\beta$.
  We show that every $\beta' \in B_{i}^\beta$ 
  satisfies~Items~\mbox{\ref{inv:it0-apx}--\ref{inv:it2-apx}}, 
  and that $\bigvee_{\beta' \in B_{i}^\beta} \theta_{\beta'} \land \phi_{\beta'}$ is equisatisfiable with $\theta_\beta \land \phi_{\beta}$ over $\N$. 
  Observe that the latter statement implies~\Cref{inv:it4-apx}, 
  because of the identity $B_{i} = \bigcup_{\beta \in B_{i-1}} B_i^\beta$ that follows directly from the definition of $B_i$.

  Let $\beta \in B_{i-1}$.
  Note that $2^{x_{n-i+1}}$ is the leading exponential term of $\theta_\beta$ and $2^{x_{n-i}}$ is its successor (using the terminology in line~\ref{algo:master:extrY}). 
  For notational convenience, we rename the variables $x_{n-i+1}$ and $x_{n-i}$ with, respectively, $x$ and $y$ (these are the names used in the pseudocode, see lines~\ref{algo:master:extrX} and~\ref{algo:master:extrY}).
  During the $i$th iteration, the loop
  manipulates $\phi_\beta$ so that it becomes a
  quotient system induced by $\theta_\beta$, which
  is then fed to~\Cref{algo:elimmaxvar} in order to remove
  the variable~$x $. 
  By~\Cref{inv:it3-apx}, all occurrences of the modulo operator $(w \bmod 2^{x })$ can be simplified, i.e., 
  \begin{equation} 
    \label{eq:master:phi-prime}
    \phi_\beta \iff \phi_\beta'\,, \qquad\text{where } \phi_\beta' \coloneqq \phi_\beta\sub{w}{(w \bmod 2^{x })}.
  \end{equation}
  The procedure uses this equivalence in line~\ref{algo:master:large-mod-sub}. 
  Following line~\ref{algo:master:def-z}, 
  we now consider all variables~$z$ in $\phi_\beta$ such that $z$ is $x $ or $z$ does not appear in $\theta_\beta$. In other words, these are all the variables $\vec z_\beta$ from~\Cref{inv:it0-apx} of the invariant, say $z_1,\dots,z_m$, plus the variable $x $ (note: by~\Cref{inv:it3half-apx}, $\vec z_\beta$ has at most $i-1$ variables).
  For each variable $z_j$, with $j \in [1,m]$, consider two fresh variables $x_j'$ and $z_j'$. Furthermore, let $x_0'$ and $z_0'$ be two fresh variables to be used for replacing $x $. We see each variable $z_j'$, with $j \in [0,m]$, as remainders modulo $2^{y }$, whereas each $x_j'$ is a quotient of the division by $2^{y }$. 
  Let $\vec x' = (x_0',\dots,x_m')$ and $\vec z' = (z_0',\dots,z_m')$.
  The following equivalence (over $\N$) is straightforward:
  \begin{equation*} 
    \begin{aligned}
    \phi_\beta' \iff \exists \vec x' \exists \vec z' \Big(\phi_\beta' \land \bigwedge\nolimits_{j=1}^m (z_j &= x_j' \cdot 2^{y } + z_j' {} \land 0 \leq z_j' < 2^{y })\\
    {}\land {} (x  &= x_0' \cdot 2^{y } + z_0' \land 0 \leq z_0' < 2^{y }) \Big).
    \end{aligned}
  \end{equation*}
  By~\Cref{inv:it0-apx}, the variables $z_1,\dots,z_m$ only occur linearly in $\phi_\beta'$ and can thus be eliminated by substitution thanks to the equalities $z_j = x_j' \cdot 2^{y } + z_j'$ above. 
  We can also substitute the linear occurrences of $x $ in $\phi_\beta'$ with $x_0' \cdot 2^{y } + z_0'$, but since this variable may occur exponentiated we must preserve the equality 
  $x  = x_0' \cdot 2^{y } + z_0'$.
  Over $\N$, we have, 
  \begin{equation*}
    \begin{aligned}
      \exists \vec z_{\beta} \phi_\beta' \iff \exists \vec x' \exists \vec z' \Big(\phi_\beta'\sub{x_j' \cdot 2^{y } + z_j'}{z_j : j \in [1,m]}
      \sub{x_0' \cdot 2^{y } + z_0'}{x } 
      &\land {}\\
      \bigwedge\nolimits_{j=0}^m 0 \leq z_j' < 2^{y } &\land (x  = x_0' \cdot 2^{y } + z_0')\Big).
      \end{aligned}
  \end{equation*}
  Define $\phi_\beta'' \coloneqq \phi_\beta'\sub{x_j' \cdot 2^{y } + z_j'}{z_j : j \in [1,m]}
  \sub{x_0' \cdot 2^{y } + z_0'}{x }$.
  Following the ordering $\theta_\beta$,
  in $\phi_\beta''$ all occurrences of the modulo operator featuring terms $x_j' \cdot 2^{y } + z_j'$, with $j \in [0,m]$,
  can be simplified.
  In particular, $((x_j' \cdot 2^{y } + z_j') \bmod 2^{y })$ can be rewritten to $z_j'$. 
  Similarly, $((x_j' \cdot 2^{y } + z_j') \bmod 2^{w})$, where $w$ is a variable distinct form $y $ and such that $\theta_\beta$ implies $2^w \leq 2^{y }$, can be rewritten to $(z_j' \bmod 2^w)$. 
  Let $\phi_\beta'''$ be the system obtained from $\phi_\beta''$ after all these modifications.
  In the pseudocode, all these updates to $\phi_\beta'$ are done by the \textbf{foreach} loop from line~\ref{algo:master:inner-loop} of the procedure.
  Over $\N$, we have,
  \begin{equation}
    \label{eq:master:pre-elimmaxvar}
    \begin{aligned}
      &\exists \vec z_{\beta} (\theta_\beta \land \phi_\beta') \iff {}\\ 
      &\hspace{1.5cm}\exists \vec x' \exists \vec z' \Big(\theta_\beta \land (\phi_\beta''' \land \bigwedge\nolimits_{j=0}^m 0 \leq z_j' < 2^{y }) \land (x  = x_0' \cdot 2^{y } + z_0') \Big).
      \end{aligned}
  \end{equation}
  Observe that $\phi_\beta''' \land \bigwedge\nolimits_{j=0}^m 0 \leq z_j' < 2^{y }$ is a quotient system induced by $\theta_\beta$.
  In line~\ref{algo:master:call-elimmaxvar},
  the procedure calls~\Cref{algo:elimmaxvar}
  on input $\theta_\beta$, $(\phi_\beta''' \land \bigwedge\nolimits_{j=0}^m 0 \leq z_j' < 2^{y })$, and ${\sub{x_0' \cdot 2^{y } + z_0'}{x }}$.
  Over $\N$, by~\Cref{lemma:corr-elimmaxvar} and 
  following the \textbf{global output} of~\Cref{algo:elimmaxvar}, we obtain
  \begin{equation}
    \label{eq:master:post-elimmaxvar}
    \begin{aligned}
    \exists x  \exists \vec x' 
    \Big(\theta_\beta \land (\phi_\beta''' \land \bigwedge\nolimits_{j=0}^m 0 \leq z_j' < 2^{y }) \land (x  = x_0' \cdot 2^{y }+z_0')\Big) 
    &\iff {}\\
    &(\exists x  \theta_\beta) \land \bigvee\nolimits_{\beta' \in B_{i}^\beta}\psi_{\beta'},
    \end{aligned}
  \end{equation}
  where each \textbf{branch output} $\psi_{\beta'}$ 
  is a linear-exponential system in variables~$x_0,\dots,x_{n-i-1},y $ and $\vec z'$, where variables in $\vec z'$ do not occur in exponentials, 
  and $\psi_{\beta'}$ features inequalities $0 \leq z_j' < 2^{y }$ for every $j \in [0,m]$.

  By~\Cref{eq:master:phi-prime,eq:master:pre-elimmaxvar,eq:master:post-elimmaxvar}, over~$\N$ we have:
  \begin{equation}
    \label{eq:master:result}
    \exists \vec z_{\beta} \exists x (\theta_\beta \land \phi_\beta) 
    \iff 
    \exists \vec z' \left( 
      (\exists x  \theta_\beta) \land \bigvee\nolimits_{\beta' \in B_{i}^\beta}\psi_{\beta'} 
      \right).
  \end{equation}
  The algorithm then excludes $2^{x }$ from $\theta_\beta$ (line~\ref{algo:master:end-loop}). The iteration of the loop has been performed.
  
  To conclude the proof, it suffices to observe that, for every $\beta' \in B_{i}^\beta$, $\theta_{\beta'}$ and $\phi_{\beta'}$ satisfy the 
  invariant of the loop in line~\ref{algo:master:outer-loop}. 
  Recall that by $x$ and $y$ we denoted the variables $x_{n-i+1}$ and $x_{n-i}$, and that we are assuming the loop invariant to hold for $i-1$. 
  Then, with respect to $\phi_{\beta'}$, the sequence of variables $\vec z_{\beta'}$ is $\vec z'$ and
  \begin{itemize}
    \item \Cref{inv:it0-apx,inv:it1-apx,inv:it3-apx} follow directly from the definition of the formulae $\psi$ returned by~\Cref{algo:elimmaxvar}.
    \item \Cref{inv:it3half-apx} follows by definition of $\vec z'$ and from the fact that $\vec z_\beta$ has at most $i-1$ variables.
    \item \Cref{inv:it2-apx} follows directly from the definition of $\theta_{\beta}$.
    \item \Cref{inv:it4-apx} follows by the equivalence in~\eqref{eq:master:result}, which implies that 
    $\bigvee_{\beta' \in B_{i}^\beta} \theta_{\beta'} \land \phi_{\beta'}$  is equisatisfiable with 
    $\theta_\beta \land \phi_{\beta}$ over \N
    (as already discussed above, this suffices to establish~\Cref{inv:it4-apx}).
    \qedhere
  \end{itemize}
\end{proof}
\section{Proofs from Section~\ref{sec:complexity}: the complexity of Algorithm~\ref{algo:master} (\Master)}
\label{app:complexity-analysis}

To simplify the exposition, we borrow from~\cite{BenediktCM23} the use of ``parameter tables'' to describe the growth of the parameters. 
These tables have the following shape.
\begin{center}
  \renewcommand\arraystretch{1.2}
  \setlength{\tabcolsep}{3pt}
  \begin{tabular}{|g|c|c|c|c|c|c|}
      \hline
      \rowcolor{light-gray}
      & $f_1(\cdot)$ 
      & $f_2(\cdot)$
      & $\cdots$
      & $f_n(\cdot)$ \\
      \hline
      \hline
      $\phi$
      & $a_1$ 
      & $a_2$
      & $\cdots$
      & $a_n$
      \\
      \hline
      \hline
      $\psi_1$
      & $f_{1,1}(a_1,\dots,a_n)$
      & $f_{1,2}(a_1,\dots,a_n)$
      & $\cdots$
      & $f_{1,n}(a_1,\dots,a_n)$
      \\ 
      \hline
      $\cdots$ & $\cdots$ & $\cdots$ & $\cdots$ & $\cdots$\\ 
      \hline
      $\psi_m$
      & $f_{m,1}(a_1,\dots,a_n)$
      & $f_{m,2}(a_1,\dots,a_n)$
      & $\cdots$
      & $f_{m,n}(a_1,\dots,a_n)$\\ 
  \hline
  \end{tabular}
\end{center}
In this table, $\phi$,$\psi_1$,\dots,$\psi_m$ are systems,  $f_1(\cdot)$,\dots,$f_n(\cdot)$ are parameter functions 
from systems to $\N$, $a_1,\dots,a_n \in \N \setminus \{0\}$, and all $f_{j,k}$ are functions from $\N^n$ to $\N$. 
The system $\phi$ should be seen as the ``input'', whereas $\psi_1,\dots,\psi_m$ should be seen as ``outputs''.
The table states that 
\begin{center}
  if $f_i(\phi) \leq a_i$ for all $i \in [1,n]$, 
  then $f_k(\psi_j) \leq f_{j,k}(a_1,\dots,a_n)$ for all $j \in [1,m]$ and $k \in [1,n]$.
\end{center}
However, we will have \textbf{two exceptions to this semantics}:
\begin{itemize}
  \item for the parameter $\fmod(\phi)$, the table must be read as if $\fmod(\phi)$ \emph{divides} $a$ (with $a$ value in the row of $\phi$ corresponding to $\fmod(\phi)$), then $\fmod(\psi_j)$ \emph{divides} $f_j(a_1,\dots,a_n)$ for all $j \in [1,m]$.
  To repeat, the parameter table always encodes a~$\leq$ relationship between input and output, except for $\fmod(\phi)$ where this relationship concerns divisibility. 
  An example of this is given in~\Cref{lemm:complexity-linearize}, where the $\fmod(\cdot)$ column should be read as if $\fmod(\phi) \divides d$, then $\fmod(\chi) \divides \totient(d)$ and $\fmod(\gamma) \divides d$. Here $\totient$ stands for Euler's totient function, see below.

  \item we will sometimes track the parameter $\card(\lst(\cdot,\cdot))$, which takes two arguments instead of one as in the parameters $f_1,\dots,f_n$ above. In this case we will indicate what is the second argument inside the cells of the column of $\card(\lst(\cdot,\cdot))$. 
  An example of this is given in~\Cref{lemm:complexity-elimmaxvar}, where 
  $\lst(\cdot,\cdot)$ should be evaluated with respect to the ordering $\theta$ when considering $\phi$, 
  and with respect to the ordering $\theta'$ 
  when considering $\psi$.
\end{itemize}
Unfilled cells in the tables correspond to quantities that are not relevant for ultimately showing~\Cref{prop:master-in-np} (see, e.g., the cells 
for~$\linnorm{\phi}$ and~$\linnorm{\gamma}$ 
in~\Cref{lemm:complexity-linearize}).

Our analysis requires the use of Euler's totient function, which we indicate with the capital phi symbol~$\totient$ to not cause confusion with homonymous linear-exponential systems.
Given a natural number $a \geq 1$ whose prime decomposition is $\prod_{i = 1}^n p_i^{k_i}$ (here $p_1,\dots,p_n$ are distinct primes and all $k_i$ are at least $1$), 
$\totient(a) \coloneqq \prod_{i=1}^n (p_i^{k_i-1}(p_i-1))$. Remark that $\totient(1) = 1$.

\subsection{Analysis of Algorithm~\ref{algo:linearize} (\SolvePrimitive)}

We will need the following folklore result (which we prove for completeness):

\begin{lemma}
  \label{lemma:folklore-mult-ord}
  Let $d$ be an odd number. Consider $x \in \N \setminus \{0\}$ satisfying $d \divides 2^x - 1$. 
  Let $\ell \in [1,d-1]$ be the multiplicative order of $2$ modulo $d$. Then $\ell$ divides $x$.
\end{lemma}

\begin{proof}
  Recall once more that $\ell$ is the smallest natural number in $[1,d-1]$ satisfying the divisibility constraint $d \divides 2^x - 1$. 
  \emph{Ad absurdum}, suppose that $\ell$ does not divide $x$. Then, there are $r \in [1,\ell-1]$ and $\lambda \in \N$ such that $x = \lambda \cdot \ell + r$. We have 
  \[ 
    d \divides 2^x - 1 
    \ \iff\ 
    d \divides 2^{\lambda \ell + r} - 1
    \ \iff\ 
    d \divides 2^{\lambda \ell} \cdot 2^r -1 
    \ \iff\ 
    d \divides 2^r - 1.
  \]
  Since $r \in [1,\ell-1]$, this contradicts the minimality of $\ell$. 
  Hence, $\ell$ divides $x$.
\end{proof}

\begin{lemma}
  \label{lemm:complexity-linearize}
  Consider a $(u,v)$-primitive linear-exponential system~$\phi$. On input $(u,v,\phi)$, \Cref{algo:linearize} (\SolvePrimitive) 
  returns a pair of linear-exponential systems $(\chi,\gamma)$ with bounds as shown below:
  \begin{center}
    \renewcommand\arraystretch{1.2}
    \setlength{\tabcolsep}{10pt}
    \begin{tabular}{|g|c|c|c|c|c|c|}
        \hline
        \rowcolor{light-gray}
        & $\card( \cdot )$ 
        & $\linnorm{ \cdot }$
        & $\fmod( \cdot )$ 
        & $\onenorm{ \cdot }$\\
        \hline
        \hline
        $\phi$
        & $s$ 
        &
        & $d$
        & $c$
        \\
        \hline
        \hline
        $\chi$
        & $2$
        & $1$
        & $\totient(d)$
        & $6 + 2 \cdot \log(\max(c,d))$
        \\  
        \hline
        $\gamma$
        & $s$
        & 
        & $d$
        & $2^6 \cdot \max(c,d)^3$\\ 
    \hline
    \end{tabular}
  \end{center}
  The algorithm runs in non-deterministic polynomial time in the bit size of~$\phi$.
\end{lemma}

\begin{proof}
  Let us start by showing the bounds in the table, column by column. 

  \begin{description}
    \item[Number of constraints.] Regarding~$\card(\cdot)$, the bound on $\card\chi$ is trivial and already given by the specification of the algorithm (proven correct in~\Cref{lemma:corr-linearize}).
    For $\card\gamma$, note that this system is defined either in line~\ref{line:linearize-chi-equality} or in line~\ref{line:linearize-chi-div}. In both cases, $\gamma$ is obtained by substitution into a subsystem of $\phi$ (in the first case it is exactly~$\phi$). 
    So, $\card\gamma \leq s$.

    \item[Linear norm.]
    We are only interested in~$\linnorm{\chi}$. Again following the specification of the algorithm, $\linnorm{\chi} = 1$.

    \item[Least common multiple of the divisibility constraints.]
    Regarding~$\fmod(\cdot)$, the case of $\gamma$ is trivial. Again this system is obtained by substituting terms $2^u$ form a subsystem of $\phi$, so $\fmod(\gamma) = \fmod(\phi)$. For $\chi$, we note that the only divisibility constraint this system can have is the divisibility $d' \divides u - r'$ from line~\ref{line:linearize-chi-case-2}. In that line, $d'$ is the multiplicative order of $2$ modulo $\frac{\fmod(\phi)}{2^n}$, where $n$ is the largest natural number dividing $\fmod(\phi)$.
    We have $\fmod(\chi) = d'$.
    By Euler's theorem, $(\frac{\fmod(\phi)}{2^n}) \divides 2^{\totient(\frac{\fmod(\phi)}{2^n})} - 1$, and therefore, by~\Cref{lemma:folklore-mult-ord},
    $d'$ divides $\totient(\frac{\fmod(\phi)}{2^n})$.
    Observe that for every $a,b \in \N \setminus \{0\}$ if $a \divides b$ then $\Phi(a) \divides \Phi(b)$.
    Then, $\totient(\frac{\fmod(\phi)}{2^n})$ divides $\totient(\fmod(\phi))$, which in turn divides $\totient(d)$. We conclude that $\fmod(\chi) \divides \totient(d)$, as required.
    
    \item[1-norm.]
    In the case of $\chi$, $\onenorm{\chi}$ is bounded by $C+1$, where $C$ is defined as in line~\ref{line:linearize-max-constant}. Then, 
    \begin{align*} 
      \onenorm{\chi} 
        &\leq \max(n, 3 + 2 \cdot \ceil{\log(\onenorm{\phi})}) + 1\\
        &\leq \max(\ceil{\log(\fmod(\phi))},  3 + 2 \cdot \ceil{\log(\onenorm{\phi})}) + 1\\
        &\leq 4 + 2 \cdot \ceil{\log(\max(\onenorm{\phi},\fmod(\phi)))}\\ 
        &\leq 6 + 2 \cdot \log(\max(c,d)).
    \end{align*}
    For $\onenorm{\gamma}$, we need to study the effects of the two substitutions $\sub{2^c}{2^u}$ and 
    $\sub{2^n \cdot r}{2^u}$ of lines~\ref{line:linearize-chi-equality} and~\ref{line:linearize-chi-div}. 
    Note that $2^n \cdot r \leq \fmod(\phi)$ and $2^c \leq 2^C$, and thus we have 
    \begin{align*} 
      \onenorm{\gamma} &\leq \onenorm{\phi} \cdot \max(\fmod(\phi),2^C)\\ 
      &\leq c \cdot \max(d, 2^{\max(n, 3 + 2 \cdot \ceil{\log(\onenorm{\phi})})})\\
      &\leq c \cdot \max(d,2^{6 + 2 \cdot \log(\max(c,d))}) 
      &\text{from the computation of $\onenorm{\chi}$}\\ 
      &\leq c \cdot \max(d, 2^6 \cdot \max(c,d)^2)
      \leq 2^6 \cdot \max(c,d)^3.
    \end{align*}
    This completes the proof of the parameter table. 
  \end{description}

  Let us now
  discuss the runtime of the procedure.
  Lines~\ref{line:linearize-decompose}--\ref{line:linearize-max-constant}
  clearly run in (deterministic) polynomial size in the bit
  size of $\phi$ (to construct the pair $(d,n)$ simply
  compute $\fmod(\gamma)$ and then iteratively divide it by
  $2$ until obtaining an odd number). To guess of $c$ in
  line~\ref{line:linearize-guess-c} it suffices to guess
  $\ceil{\log(C)} + 1$ many bits (the $+1$ is used to encode
  the case of $\star$); which can be done in
  non-deterministic polynomial time.
  Lines~\ref{line:linearize-gamma-equality},~\ref{line:linearize-chi-equality}
  and~\ref{line:assert-not-bottom} only require polynomial
  time in the bit size of $\phi$.
  To guess $r$ in line~\ref{line:linearize-guess-div} 
  it suffices guessing at most $\ceil{\fmod(\phi)}$ many bits. 
  We then arrive to lines~\ref{line:linearize-assert}--\ref{line:linearize-mult-ord}. These three lines correspond to discrete logarithm problems, which can be solved in non-deterministic polynomial time (in the bit size of $\phi$), see~\cite{Joux14}. For lines~\ref{line:linearize-assert} and~\ref{line:linearize-discrete-log}, one has to find $x \geq 0$ such that $d \divides 2^x - 2^n \cdot r$. 
  For line~\ref{line:linearize-mult-ord}, one has to find $x \in [1,d-1]$ \emph{minimal} such that $d \divides 2^{x} - 1$.
  Note that asking for the minimality here is not a problem and $x$ can be found in non-deterministic polynomial time: the non-deterministic algorithm simply guesses a number $\ell \in [1,d-1]$ and returns it together with its prime factorisation. 
  By~\Cref{lemma:folklore-mult-ord}, certifying in polynomial time that $\ell$ is the multiplicative order of $2$ modulo $d$ is trivial:
  \begin{enumerate}
    \item  check if $(2^\ell \bmod d)$ equals $1$ with a fast modular exponentiation algorithm; 
    \item check if the given prime factorisation is indeed the factorisation of $\ell$;
    \item check if dividing $\ell$ by any prime in its prime factorisation results in a number $\ell'$ such that $(2^\ell \bmod d)$ does not equal $1$.
  \end{enumerate}
  After line~\ref{line:linearize-mult-ord}, the algorithm runs in polynomial time in the bit size of $\phi$. 
  We conclude that, overall,~\Cref{algo:linearize} runs in non-deterministic polynomial time in the bit size of $\phi$.
\end{proof}

\subsection{Analysis of Algorithm~\ref{algo:elimmaxvar} (\ElimMaxVar)}

\begin{lemma}
  \label{lemm:complexity-elimmaxvar}
  Consider the execution of~\Cref{algo:elimmaxvar} (\ElimMaxVar)
  on input $(\theta,\phi,\sigma)$, 
  where  $\theta(\vec x)$ is an ordering of exponentiated variables, 
  $\phi(\vec x,\vec x', \vec z')$ is a quotient system induced by~$\theta$, with $n$ exponentiated variables~$\vec x$, $m$ quotient variables $\vec x'$, $m$ remainder variables $\vec z'$, 
  and $\sigma = \sub{x' \cdot 2^y + z'}{x}$ is a delayed substitution. 
  Let $\theta'$ be the ordering obtained from $\theta$ by removing its leading 
  term~$2^x$.
  The linear-exponential systems $\psi$ and $\gamma$ defined in the algorithm when its execution reaches line~\ref{line:elimmaxvar-gauss-2} satisfy the following bounds:
  \begin{center}
    \renewcommand\arraystretch{1.3}
    \setlength{\tabcolsep}{4.5pt}
    \begin{tabular}{|g|c|c|c|c|c|c|}
        \hline
        \rowcolor{light-gray}
        & $\card( \cdot )$ 
        & $\card\lst( \cdot , \cdot )$
        & $\linnorm{ \cdot }$
        & $\fmod( \cdot )$ 
        & $\onenorm{ \cdot }$\\
        \hline
        \hline
        $\phi$
        & $s$ 
        & $\theta \colon \ \ \ell$
        & $a$
        & $d$
        & $c$
        \\
        \hline
        \hline
        $\psi$
        & $s + 2 \cdot \ell  + 3$
        & $\theta' \colon \ \ \ell+2$
        & $a$
        & $\lcm(d,\totient(\alpha \cdot d))$
        & $\tocheck{16 \cdot (m+2)^2 \cdot (c+2)} + 4 \cdot \log(d)$
        \\  
        \hline
        $\gamma$
        & $s+m+2$
        & 
        & 
        & $\lcm(\alpha \cdot d, \totient(\alpha \cdot d))$
        & $\tocheck{(c + 3)^{14 \cdot (m+2)^2} \cdot d^3}$\\ 
    \hline
    \end{tabular}
  \end{center}
  where $\alpha \leq \tocheck{(a+2)^{(m+2)^2}}$.
  Moreover, the algorithm runs in non-deterministic polynomial time in the bit size of the input.
\end{lemma}

\begin{proof}
  We remark that the formula $\psi$ in the statement of the lemma also corresponds to the formula in output 
  of~\Cref{algo:elimmaxvar}.

  Because of the length of the proof, it is useful to split~\Cref{algo:elimmaxvar} into three steps.
  \begin{alphaenumerate}
    \item\label{comp:elimmax:s1} We consider lines~\ref{line:elimmaxvar-u}--\ref{line:elimmaxvar-psi-divisibility}, and write down the parameter table with respect to the systems~$\psi$ and~$\gamma$ when the execution reaches line~\ref{line:elimmaxvar-gauss-1}, before calling~\Cref{algo:gaussianqe}.
    \item\label{comp:elimmax:s2} We consider lines~\ref{line:elimmaxvar-gauss-1}--\ref{line:elimmaxvar-delayed-subs} to discuss how $\gamma$ and the auxiliary system $\chi$ evolve after the calls to~\Cref{algo:gaussianqe} and~\Cref{algo:linearize}.
    \item\label{comp:elimmax:s3} We analyse lines~\ref{line:elimmaxvar-x1-start}--\ref{line:elimmaxvar-gauss-2}, completing the proof.
  \end{alphaenumerate}
  Observe that this division into Steps~(\ref{comp:elimmax:s1}), (\ref{comp:elimmax:s2}), and~(\ref{comp:elimmax:s3}) 
  is slightly different from the one into Steps~(\ref{item:elimmaxvar-i}), (\ref{item:elimmaxvar-ii}), and~(\ref{item:elimmaxvar-iii}), 
  which we used to prove correctness of the algorithm. 
  Now the first step does not include line~\ref{line:elimmaxvar-gauss-1}. 
  Since we have already performed the complexity analysis of~\Cref{algo:gaussianqe} and~\Cref{algo:linearize}, it is here convenient to analyse the two in a single step.

  \subparagraph*{Step~(\ref{comp:elimmax:s1}).}
  Let $\psi$ and $\gamma$ be as defined in the algorithm when its execution reaches
  line~\ref{line:elimmaxvar-gauss-1}, before~\Cref{algo:gaussianqe} is called.
  They satisfy the following bounds:
  \begin{center}
    \renewcommand\arraystretch{1.2}
    \setlength{\tabcolsep}{10pt}
    \begin{tabular}{|g|c|c|c|c|c|c|}
        \hline
        \rowcolor{light-gray}
        & $\card( \cdot )$ 
        & $\card\lst( \cdot , \cdot )$
        & $\linnorm{ \cdot }$
        & $\fmod( \cdot )$ 
        & $\onenorm{ \cdot }$\\
        \hline
        \hline
        $\phi$
        & $s$ 
        & $\theta \colon \ \ \ell$
        & $a$
        & $d$
        & $c$
        \\
        \hline
        \hline
        $\psi$
        & $s + 2 \cdot \ell$
        & $\theta' \colon \ \ \ell$
        & $a$
        & $d$
        & $2 \cdot c + 1$
        \\  
        \hline
        $\gamma\sub{2^u}{u}$
        & $s$
        & 
        & $a\,^{\star}$
        & $d$
        & $c + 1$\\ 
    \hline
    \end{tabular}
    \setlength{\tabcolsep}{2pt}
    \begin{tabular}{rp{0.7\linewidth}}
    $\star : $& for $\gamma$, we are only interested in the coefficients of variables distinct from $u$.
    Hence the ``weird'' substitution of $u$ with $2^u$.
    \end{tabular}
  \end{center}
  As remarked below the table, in fact more than being interested in $\gamma$, we are interested in $\gamma\sub{2^u}{u}$. 
  The reason for this is simple: after calling~\Cref{algo:gaussianqe} on $\gamma$, the variable $u$ 
  (which is not eliminated by this algorithm) is replaced with $2^u$ in line~\ref{line:elimmaxvar-2-u}.
  Because $u$ is a placeholder for an exponentiated, not a linearly occurring variable, the coefficient of 
  $u$ in $\gamma$ should not be taken into account as part of $\linnorm{\gamma}$. 
  As done for~\Cref{lemm:complexity-linearize}, we show that the bounds in this table are correct column by column, starting form the leftmost one.

  \begin{description}
    \item[Number of constraints.] 
      The bound $\card{\gamma} \leq s$ is simple to establish. The system $\gamma$ is initially defined as $\top$
      in line~\ref{line:elimmaxvar-initialize}. 
      Afterwards, in every iteration of the \textbf{foreach} loop of line~\ref{line:elimmaxvar-foreach}, 
      $\gamma$ is conjoined with an additional constraint, either an inequality (line~\ref{line:elimmaxvar-gamma-inequality}) or a divisibility constraint (line~\ref{line:elimmaxvar-gamma-divisibility}).
      Therefore, when the execution reaches line~\ref{line:elimmaxvar-gauss-1}, 
      the number of constraints in $\gamma$ is bounded by the number of iterations of 
      the \textbf{foreach} loop, that is, $\card{\phi}$.

      Let us move to $\card{\psi}$. Again, $\psi$ is initially defined as $\top$ in line~\ref{line:elimmaxvar-initialize}.
      Similarly to $\gamma$, the updates in line~\ref{line:elimmaxvar-strict-add-to-psi} (for strict inequalities), 
      line~\ref{line:elimmaxvar-psi-equality} (for equalities), and line~\ref{line:elimmaxvar-psi-divisibility} (for divisibilities)
      add to $\psi$ at most $\card{\phi}$ many constraints. 
      Although it is easy to see, notice that strict inequalities become non-strict in line~\ref{line:elimmaxvar-strict-ineq}; 
      hence, in this case, the \textbf{if} statement of line~\ref{line:elimmaxvar-psi-equality} is clearly not true, 
      and the formula $\psi$ is only updated in line~\ref{line:elimmaxvar-strict-add-to-psi}.
      
      We then need to account for the update done in line~\ref{line:elimmaxvar-psi-inequality}.
      This update adds two inequalities 
      whenever the guard ``$\Delta(\rho)$ is undefined'' in the \textbf{if} statement of line~\ref{line:elimmaxvar-delta-undefined}
      is true. Here $\rho$ is the least significant part of an (in)equality of $\phi$.
      Each least significant part is considered only once, because of the update done to the map~$\delta$ in line~\ref{line:elimmaxvar-delta-update}.
      Therefore, during the \textbf{foreach} loop of line~\ref{line:elimmaxvar-foreach}, 
      at most $2 \cdot \card{\lst(\phi)}$
      many inequalities can be added to $\psi$ because of line~\ref{line:elimmaxvar-psi-inequality}. 
      We conclude that $\card(\psi) \leq s + 2 \cdot \ell$.

    \item[Least significant terms.] 
      We are only interested in bounding $\card{\lst(\psi,\theta')}$. Note that, with respect 
      to the code of the algorithm, the leading exponential term in $\theta'$ is $2^y$.
      Then, let us look once more at the (in)equalities added to $\psi$ in lines~\ref{line:elimmaxvar-psi-inequality},~\ref{line:elimmaxvar-strict-add-to-psi}, 
      and~\ref{line:elimmaxvar-psi-equality}.
      They have all the form $g \cdot 2^y \pm \rho \sim 0$, with $\sim \{<,\leq,=\}$, $g \in \Z$
      and $\rho$ being the least significant part of some term in $\phi$ (with respect to $\theta$).
      By definition, $\pm \rho$ is the least significant part of $g \cdot 2^y \pm \rho \sim 0$ with respect to $\theta'$.
      We conclude that~$\lst(\psi,\theta')$ is the set containing all such $\rho$ and $-\rho$.
      By definition of $\lst$, 
      both $\rho$ and $-\rho$ occur in $\lst(\phi,\theta)$. Therefore,
      $\card{\lst(\psi,\theta')} \leq \card{\lst(\phi,\theta)} \leq \ell$.
    
    \item[Linear norm.] 
      For the aforementioned reasons, we look at $\gamma\sub{2^u}{u}$ instead of $\gamma$. 
      The linear norm of $\gamma\sub{2^u}{u}$ is solely dictated by the update \ $\gamma\gets\gamma\land(a \cdot u + f(\vec x')+r\sim0)$ \ done in line~\ref{line:elimmaxvar-gamma-inequality}. 
      Here, $f(\vec x')$ is a linear term $a_1 \cdot x_1' + \dots + a_m \cdot x_m' + a_{m+1}$
      that occur in a quotient term $a \cdot 2^x + f(\vec x') \cdot 2^y + \rho$ of $\phi$. 
      In $\phi$, the linear term $f$ is accounted for in defining the linear norm.
      We conclude that $\linnorm{\gamma\sub{2^u}{u}} \leq \linnorm{\phi} \leq a$.

      One bounds $\linnorm{\psi}$ in a similar way. 
      In the updates done to $\psi$ in lines~\ref{line:elimmaxvar-psi-inequality},~\ref{line:elimmaxvar-strict-add-to-psi}, 
      and~\ref{line:elimmaxvar-psi-equality}, the linear norm of the added (in)equalities is the linear norm of~$\rho$. 
      Here, $\rho$ is a subterm of $\phi$, and therefore we have $\linnorm{\psi} \leq \linnorm{\phi} \leq a$.

    \item[Least common multiple of the divisibility constraints.] 
      Bounding this parameter is simple. Divisibility constraints~$d \divides \tau$ are added to $\gamma$ and $\psi$ in lines~\ref{line:elimmaxvar-gamma-divisibility} and~\ref{line:elimmaxvar-psi-divisibility}. 
      The divisor $d$ is also a divisor of some divisibility constraint of $\phi$. 
      So, $\fmod(\psi) \divides \fmod(\phi)$ and $\fmod(\gamma) \divides \fmod(\phi)$.
    
    \item[$1$-norm.] 
      Regarding $\onenorm{\gamma}$, again we need to look at the updates caused by line~\ref{line:elimmaxvar-gamma-inequality}. 
      There, notice that the 1-norm of the term $a \cdot u + f(\vec x') + \rho$ is bounded by $\onenorm{\phi}$, 
      and the constant $r$ is bounded in absolute value by $\onenorm{\rho}+1$. 
      We conclude that $\onenorm{\gamma} \leq c + 1$.

      For $\onenorm{\psi}$ the analysis is similar. 
      The inequalities $g \cdot 2^y \pm \rho \sim 0$ added in lines~\ref{line:elimmaxvar-psi-inequality},~\ref{line:elimmaxvar-strict-add-to-psi}, and~\ref{line:elimmaxvar-psi-equality}
      are such that $\onenorm{\rho} \leq \onenorm{\phi}$ and $\abs{g} \leq \onenorm{\phi}+1$. 
      Then, $\onenorm{\psi} \leq 2 \cdot c + 1$.
  \end{description}

  We have now established that $\psi$ and $\gamma$ satisfy the bounds of the table for line~\ref{line:elimmaxvar-gauss-1} of the algorithm. 
  We move to the second step of the proof.

  \subparagraph*{Step~(\ref{comp:elimmax:s2}).}
    This step only involves the formula~$\gamma$, which is first manipulated by~\Cref{algo:gaussianqe} in order to remove all quotient variables in $\vec x'$ that are different from the quotient variable $x'$ appearing in the delayed substitution, 
    and then passed to~\Cref{algo:linearize}
    in order to ``linearise'' the occurrences of $u$. 

    Note that before the execution of line~\ref{line:elimmaxvar-gauss-1}, 
    $\gamma$ has $m+1$ variables (i.e., $u$ plus the $m$ quotient variables~$\vec x'$).
    Note moreover that in line~\ref{line:elimmaxvar-gauss-1} 
    we conjoin $\gamma$ with the system $\vec x' \geq \vec 0$ featuring $m$ inequalities. 
    Then, starting form the bounds on $\gamma$ obtained in the previous part of the proof, by~\Cref{lemma:complexity-gaussianqe},
    line~\ref{line:elimmaxvar-gauss-1}
    yields the following new bounds for $\gamma$: 

    \begin{center}
      \renewcommand\arraystretch{1.3}
      \setlength{\tabcolsep}{10pt}
      \begin{tabular}{|g|c|c|c|c|c|c|}
          \hline
          \rowcolor{light-gray}
          & $\card( \cdot )$ 
          & $\linnorm{\cdot}$
          & $\fmod( \cdot )$ 
          & $\onenorm{ \cdot }$\\
          \hline
          \hline
          $\phi$
          & $s$ 
          & $a$
          & $d$
          & $c$
          \\
          \hline
          \hline
          $\gamma$
          & $s+m$
          &
          & $\alpha \cdot d$
          & $\tocheck{(c + 3)^{4 \cdot (m+2)^2} \cdot d}$\\ 
      \hline
      \end{tabular}
      \setlength{\tabcolsep}{2pt}
      \begin{tabular}{rp{0.7\linewidth}}
      \end{tabular}
    \end{center}
    where $\alpha \in [1,\tocheck{(\linnorm{\phi}+2)^{(m+2)^2}}] \subseteq [1,\tocheck{(a+2)^{(m+2)^2}}]$.

    Line~\ref{line:elimmaxvar-2-u} replaces $u$ with $2^u$ and does not change the bounds above. The procedure then calls 
    \Cref{algo:linearize} on $\gamma$ (line~\ref{line:elimmaxvar-solve-primitive}). 
    The output of this algorithm updates $\gamma$ and produces the auxiliary system $\chi$. By~\Cref{lemm:complexity-linearize}, 
    these systems enjoy the following bounds:
    \begin{center}
      \renewcommand\arraystretch{1.3}
      \setlength{\tabcolsep}{10pt}
      \begin{tabular}{|g|c|c|c|c|c|c|}
          \hline
          \rowcolor{light-gray}
          & $\card( \cdot )$ 
          & $\linnorm{ \cdot }$
          & $\fmod( \cdot )$ 
          & $\onenorm{ \cdot }$\\
          \hline
          \hline
          $\phi$
          & $s$ 
          &
          & $d$
          & $c$
          \\
          \hline
          \hline
          $\chi$
          & $2$
          & $1$
          & $\totient(\alpha \cdot d)$
          & $6+2 \cdot \log(\max(\tocheck{(c + 3)^{4 \cdot (m+2)^2} \cdot d},\alpha \cdot d))$
          \\  
          \hline
          $\gamma$
          & $s+m$
          & 
          & $\alpha \cdot d$
          & $2^6 \cdot \max(\tocheck{(c + 3)^{4 \cdot (m+2)^2} \cdot d},\alpha \cdot d)^3$\\ 
      \hline
      \end{tabular}
    \end{center}
    Note that, by definition, $\linnorm{\phi} \leq \onenorm{\phi} \leq c$. 
    After~\Cref{algo:linearize}, 
    the procedure updates $\chi$ 
    to $\chi\sub{x-y}{u}\sub{x'\cdot2^y+z'}{x}$. In $\chi$, the coefficient of $u$ in inequalities is always $\pm 1$. 
    Then, this update cause the 
    $1$-norm of $\chi$ to increase by 
    $2$. This completes Step~\eqref{comp:elimmax:s2}. 

    Before moving to Step~\eqref{comp:elimmax:s3}, let us clean up a bit the bounds on $\onenorm{\chi}$ and $\onenorm{\gamma}$. We have,
    \begin{align*} 
      \max(\tocheck{(c + 3)^{4 \cdot (m+2)^2} \cdot d},\alpha \cdot d)
      & \leq \max(\tocheck{(c + 3)^{4 \cdot (m+2)^2} \cdot d},\tocheck{(c+2)^{(m+2)^2}}\cdot d)\\
      & \leq \tocheck{(c + 3)^{4 \cdot (m+2)^2} \cdot d}\,,
    \end{align*}
    and so,
    \begin{align*} 
      \onenorm{\chi} 
      &\leq 6+2 \cdot \log(\tocheck{(c + 3)^{4 \cdot (m+2)^2} \cdot d})+2\\
      &= 8 + \tocheck{8 \cdot (m+2)^2 \cdot \log(c + 3)} + 2 \cdot \log(d)\\
      &= \tocheck{8 \cdot (1 + (m+2)^2 \cdot \log(c + 3))} + 2 \cdot \log(d)\,;
    \end{align*}
    \begin{align*}
      \onenorm{\gamma}
      & \leq 2^6 \cdot (\tocheck{(c + 3)^{4 \cdot (m+2)^2} \cdot d})^3\\
      & \leq \tocheck{(c + 3)^{12 \cdot (m+2)^2+6} \cdot d^3}\\
      & \leq \tocheck{(c + 3)^{14 \cdot (m+2)^2} \cdot d^3}\,.
    \end{align*} 
    
  \subparagraph*{Step~(\ref{comp:elimmax:s3}).} 
  We now analyse lines~\ref{line:elimmaxvar-x1-start}--\ref{line:elimmaxvar-gauss-2}. 
  Note that in these lines the definition of $\chi$ is used to update $\gamma$ and $\psi$. 
  Since $\psi$ was not updated during~Step~\eqref{comp:elimmax:s2}, 
  it still has the bounds from~Step~\eqref{comp:elimmax:s1}.
  We show that, when the execution of the procedure reaches 
  line~\ref{line:elimmaxvar-gauss-2}, 
  $\psi$ and $\gamma$ satisfy the bounds in the statement of the lemma, here reproposed:
  \begin{center}
    \renewcommand\arraystretch{1.3}
    \setlength{\tabcolsep}{4.5pt}
    \begin{tabular}{|g|c|c|c|c|c|c|}
        \hline
        \rowcolor{light-gray}
        & $\card( \cdot )$ 
        & $\card\lst( \cdot , \cdot )$
        & $\linnorm{ \cdot }$
        & $\fmod( \cdot )$ 
        & $\onenorm{ \cdot }$\\
        \hline
        \hline
        $\phi$
        & $s$ 
        & $\theta \colon \ \ \ell$
        & $a$
        & $d$
        & $c$
        \\
        \hline
        \hline
        $\psi$
        & $s + 2 \cdot \ell  + 3$
        & $\theta' \colon \ \ \ell+2$
        & $a$
        & $\lcm(d,\totient(\alpha \cdot d))$
        & $\tocheck{16 \cdot (m + 2)^2 \cdot (c + 2)} + 4 \cdot \log(d)$
        \\  
        \hline
        $\gamma$
        & $s+m+2$
        & 
        & 
        & $\lcm(\alpha \cdot d, \totient(\alpha \cdot d))$
        & $\tocheck{(c + 3)^{14 \cdot (m+2)^2} \cdot d^3}$\\ 
    \hline
    \end{tabular}
  \end{center}
  where $\alpha \leq \tocheck{(a+2)^{(m+2)^2}}$.

  Observe that the procedure executes either 
  lines~\ref{line:elimmaxvar-guess-c}--\ref{line:elimmaxvar-psi-equality-x1}
  or lines~\ref{line:elimmaxvar-guess-times}--\ref{line:elimmaxvar-psi-inequality-x1}, 
  depending on whether $\chi$ is an equality.

  \begin{description}
    \item[Number of constraints.] 
      If $\chi$ is an equality, then a single constraint is added to $\psi$. 
      If $\chi$ is not an equality, 
      the procedures adds $3$ constraints. 
      Together with the bounds in the table of Step~\eqref{comp:elimmax:s1}, we conclude that $\card\psi \leq s + 2 \cdot \ell + 3$ when the procedure reaches line~\ref{line:elimmaxvar-gauss-2}.

      In the case of $\gamma$, at most $2$ more constraints are added (this corresponds to the case of $\chi$ not being an equality). Then, $\card\gamma \leq s + m + 2$.

    \item[Least significant terms.]
      We are only interested in $\lst(\psi,\theta')$. Independently on whether $\chi$ is an equality, 
      the updates done to $\psi$ 
      only add the least significant terms 
      $\pm(-z'+y+c)$. 
      Therefore, $\card{\lst(\psi,\theta')} \leq \ell + 2$.
    \item[Linear norm.]
      Again, we are only interested in $\psi$. Independently on the shape of $\chi$, in the inequalities added to $\psi$ by the procedure all linearly occurring variables ($z'$ and $y$ in the pseudocode) have $\pm 1$ as a coefficient.
      Since we are assuming $a \geq 1$, 
      from the table of Step~\eqref{comp:elimmax:s1} 
      we conclude $\linnorm{\psi} \leq a$.

    \item[Least common multiple of the divisibility constraints.]
      When $\chi$ is an equality, 
      no divisibility constraints are added to $\psi$ and $\gamma$. 
      When $\chi$ is not an equality, 
      a single divisibility constraint is added to both $\psi$ and $\gamma$, 
      with divisor $\fmod(\chi)$.

      From the tables in Step~\eqref{comp:elimmax:s1} and Step~\eqref{comp:elimmax:s2}, 
      we conclude that $\fmod(\psi)$ divides $\lcm(d,\totient(\alpha \cdot d))$, 
      and $\fmod(\gamma)$ divides $\lcm(\alpha \cdot d, \totient(\alpha \cdot d))$.
    \item[$1$-norm.]   
      The $1$-norm of the (in)equalities added to $\gamma$ 
      is bounded by $\onenorm{\chi}$. 
      Therefore, 
      \begin{align*}
        \onenorm{\gamma} &\leq \max\big(\tocheck{8 \cdot (1 + (m+2)^2 \cdot\log(c + 3))} + 2 \cdot \log(d)\,,\, 
        \tocheck{(c + 3)^{14 \cdot (m+2)^2} \cdot d^3}\big)\\
        & \leq \tocheck{(c + 3)^{14 \cdot (m+2)^2} \cdot d^3}.
      \end{align*}
      The $1$-norm of the (in)equalities added to $\psi$ is bounded by $2 \cdot \onenorm{\chi}$. Hence, 
      \begin{align*}
        \onenorm{\psi} &\leq \max\big(\tocheck{16 \cdot (1 + (m+2)^2 \cdot \log(c + 3))} + 4 \cdot \log(d)\,,\,2 \cdot c + 1\big)\\
        & \leq \tocheck{16 \cdot (m+2)^2 \cdot (c + 2)} + 4 \cdot \log(d).
      \end{align*}

      This completes the proof of the parameter table in the statement of the lemma.
  \end{description}

  Let us now discuss the runtime of the procedure. First,
  recall that the bit size of a system
  $\phi'(x_1,\dots,x_k)$ belongs to $O(\card\phi \cdot k^2
  \cdot \log(\onenorm{\phi}) \cdot \log(\fmod(\phi)))$. This
  means that all the formula we have analysed above have a
  bit size polynomial in the bit size of $\phi$.
  It is then quite simple to see that~\Cref{algo:elimmaxvar} runs in non-deterministic polynomial time:
  \begin{itemize}
    \item the \textbf{foreach} loop in line~\ref{line:elimmaxvar-foreach} simply iterates over all constraints of $\phi$, and thus runs in polynomial time in the bit size of $\phi$. 
    \item For the two guesses in lines~\ref{line:elimmaxvar-guess-rho} and~\ref{line:elimmaxvar-guess-mod} requires guessing at most $\ceil{\log(\onenorm{\phi})}+1$ and $\ceil{\log(\fmod(\phi))}$ many bits, respectively. This number is bounded by the bit size of $\phi$.
    \item From the tables above, the inputs of~\Cref{algo:gaussianqe} and~\Cref{algo:linearize} in lines~\ref{line:elimmaxvar-gauss-1} and~\ref{line:elimmaxvar-solve-primitive} 
    have a polynomial bit size with respect to the bit size of $\phi$. 
    Therefore, these two algorithms run in non-deterministic polynomial time in the bit size of $\phi$, by~\Cref{thm:gaussianQE-in-np} and~\Cref{lemm:complexity-linearize}.
    \item The call to~\Cref{algo:gaussianqe} done in line~\ref{line:elimmaxvar-gauss-2} takes again as input a formula with bit size polynomial in the bit size of $\phi$. 
    So, once more, this step runs in non-deterministic polynomial time.
    \item Every other line in the code runs in polynomial time in the bit size of $\phi$.
    \qedhere
  \end{itemize}
\end{proof}

\subsection{Analysis of Algorithm~\ref{algo:master} (\Master)}

We first compute the bounds for a single iteration of the \textbf{while} loop of line~\ref{algo:master:outer-loop}. 

\LemmaBoundsOneLoop*

\begin{proof}
  We prove this result by first analysing lines~\ref{algo:master:extrX}--\ref{algo:master:end-inner-loop} and then appealing to~\Cref{lemm:complexity-elimmaxvar}.

  Observe that, from the proof 
  of~\Cref{prop:master-correct}, the variables in $\phi_i$ are from either $\{x_0,\dots,x_{n-i}\}$, with $x_0$ being the fresh variable introduced in 
  line~\ref{algo:master:addxzero}, 
  or from a vector $\vec z_i$ of at most $i$ variables (see~\Cref{inv:it0-apx}  and~\Cref{inv:it3half-apx} of the invariant in the proof of~\Cref{prop:master-correct}).
  This already implies that $\phi_{i}$ 
  has at most $n+1$ variables, for every $i \in [0,n]$.

  Let $i \in [0,n-1]$.
  Starting from the system~$\phi_i$, we consider the execution of the 
  \textbf{while} loop of line~\ref{algo:master:outer-loop},
  until reaching the call to~\Cref{algo:elimmaxvar} in line~\ref{algo:master:call-elimmaxvar} (and before executing such a call). 
  Let $\phi_i'$ be the quotient system induced by~$\theta_i$ 
  passed in input to~\Cref{algo:elimmaxvar}
  at that point of the execution of the program. We show the following parameter table:

  \begin{center}
    \renewcommand\arraystretch{1.2}
    \setlength{\tabcolsep}{8pt}
    \begin{tabular}{|g|c|c|c|c|c|c|}
        \hline
        \rowcolor{light-gray}
        & $\card( \cdot )$ 
        & $\card\lst( \cdot , \cdot )$
        & $\linnorm{ \cdot }$
        & $\fmod( \cdot )$ 
        & $\onenorm{ \cdot }$\\
        \hline
        \hline
        $\phi_i$
        & $s$ 
        & $\theta \colon \ \ \ell$
        & $a$
        & $d$
        & $c$
        \\
        \hline
        \hline
        $\phi_i'$
        & $s + 2 \cdot (i+1)$
        & $\theta \colon \ \ \ell+ 2 \cdot (i+1)$
        & $3 \cdot a$
        & $d$
        & $2 \cdot (c+1)$
        \\ 
    \hline
    \end{tabular}
  \end{center}

  \begin{description}
    \item[Number of constraints.] 
      Constraints are only added in 
      line~\ref{algo:master:imposezlessy}, 
      as part of the \textbf{foreach} loop in line~\ref{algo:master:inner-loop}.
      Each iteration of this loop adds $2$ constraints, and the number of iterations depends on the cardinality of $\vec z$. 
      As explained in the proof 
      of~\Cref{prop:master-correct}, 
      $\vec z$ has at most $i+1$ variables,  
      and therefore $\card\phi_i' \leq \card\phi_i + 2 \cdot (i+1)$.
    \item[Least significant terms.] 
      Note that while $\psi_i$ is a linear-system, $\psi_i'$ is a quotient system induced by $\theta$, so in the latter any least significant part of a term is of the form ${b \cdot y + \rho(\vec x \setminus \{x,y\}, \vec z)}$, where $\vec x \setminus \{x,y\}$ is the vector obtained from $\vec x$ by removing $x$ and $y$.
      It is simple to see that the substitutions performed in lines~\ref{algo:master:large-mod-sub}, \ref{algo:master:elimmod} and~\ref{algo:master:linear-substitution} do not change the cardinality of $\lst$. 
      Indeed, if two inequalities had the same least significant part before one of these substitutions, they will have the same least significant part also after the substitution.
      Least significant parts are however introduced in line~\ref{algo:master:imposezlessy}. Indeed, 
      the least significant part of $0 \leq z'$ is $-z'$ and the least significant part of $z' < 2^y$ is $z'$. 
      Since the \textbf{foreach} loop in line~\ref{algo:master:inner-loop} executes $i+1$ times, 
      we conclude that $\card{\lst(\phi_i',\theta_i)} \leq \card{\lst(\phi_i,\theta_i)} + 2 \cdot (i+1)$.
    \item[Linear norm.] 
      We only consider the lines of the loop that might influence the linear norm. 
      Line~\ref{algo:master:large-mod-sub} 
      performs on $\phi_i$ the substitution $\sub{w}{(w \bmod 2^x) : w \text{ is a variable}}$. 
      This can (at most) double the linear norm, as terms of the form $a_1 \cdot w + a_2 \cdot (w \bmod 2^x)$ will 
      be rewritten into $(a_1 + a_2)\cdot w$. We observe however that the maximum absolute value of a coefficient of some modulo operator $(w' \bmod 2^{w''})$ is still bounded by $\linnorm{\phi_i}$ after this update. The terms in the inequalities of line~\ref{algo:master:imposezlessy} have a linear norm of $1$. Since we are assuming $a \geq 1$, this does not influence our bound for the linear norm of $\phi_i'$. 
      Line~\ref{algo:master:elimmod} does not influence the linear norm, as at this stage point in the execution of the procedure 
      $z'$ only occur in the inequalities added to the system in line~\ref{algo:master:imposezlessy}.
      Line~\ref{algo:master:simpmod} does not increase the linear norm, since before the substitution performed 
      in this line the variable $z'$ does not occur inside the modulo operator. 
      Line~\ref{algo:master:linear-substitution} does increase the linear norm: the substitution $\sub{(x'\cdot 2^y + z')}{z}$ 
      causes $\linnorm{\phi_i'}$ to be bounded by $3 \cdot a$ (after all iterations of the loop of line~\ref{algo:master:inner-loop}). 
      This is because, before performing the substitution, 
      the coefficients of $z$ in inequalities of $\phi_i'$ are bounded by $2 \cdot a$ (this stems from the update performed in line~\ref{algo:master:large-mod-sub}), 
      and the coefficient of $z'$ is bounded by~$a$ (this stems from the substitution in line~\ref{algo:master:elimmod}). 
    \item[Least common multiple of the divisibility constraints.] No line increases $\fmod(\cdot)$.
    \item[$1$-norm.] Recall that we are assuming $c \geq 1$. Observe that lines~\ref{algo:master:large-mod-sub}, \ref{algo:master:elimmod} and~\ref{algo:master:simpmod} 
    do not modify the $1$-norm. The inequalities added in line~\ref{algo:master:imposezlessy} have a $1$-norm of $2$, and are thus bounded by $c+1$. 
    The only other line that affects the $1$-norm is therefore line~\ref{algo:master:linear-substitution}, which doubles it.
  \end{description}
  We have now established that $\phi_i'$ satisfies the bounds on the parameter table. 
  Observe also that $\phi_i'$ has at most $(i+1)$ quotient variables and $(i+1)$ remainder variables (since~$\vec z$ has at most $(i+1)$ many variables). 
  Then, the bounds on $\phi_{i+1}$ provided in the statement of the lemma follow directly from~\Cref{lemm:complexity-elimmaxvar}.

  We conclude the proof of the lemma by discussing the runtime of the body of the \textbf{while} loop.
  Observe that lines~\ref{algo:master:extrX}--\ref{algo:master:end-inner-loop} run in deterministic polynomial time. Moreover, from the parameter table above, $\phi_i'$ is of bit size polynomial in the bit size of $\phi_i$. Therefore, when accounting for the execution of~\Cref{algo:elimmaxvar}, we conclude
  by~\Cref{lemm:complexity-elimmaxvar} 
  that the body of the \textbf{while} loop runs in non-deterministic polynomial time in the bit size of $\phi_i$.
\end{proof}

Before iterating the bounds of~\Cref{lemma:bounds-one-loop}, thus completing the analysis of the complexity of~\Cref{algo:master}, we prove~\Cref{lemma:totient-bound}.
The proof of this lemma requires a few auxiliary definitions that we now introduce. 

Let $a \in \N \setminus \{0\}$ and let $\prod_{i=1}^n p_i^{k_i}$ be its prime factorisation, 
where $p_1,\dots,p_n$ are distinct primes.
We recall that the \emph{radical} of $a$ is defined as $\rad(a) \coloneqq \prod_{i=1}^n p_i$. 
We introduce the notion of \emph{shifted radical} of $a$, defined as $\srad(a) \coloneqq \prod_{i=1}^n(p_i-1)$. Note that the map $x \mapsto \srad(x)$ over $\N$ corresponds to the sequence \href{https://oeis.org/A173557}{\underline{A173557}}
in Sloane's \textit{On-Line Encyclopedia of Integer Sequences} (OEIS).

By definition of Euler's totient function, we have $\totient(a) = \frac{a}{\rad(a)}\srad(a)$. Let us state a few properties of the shifted radical that follow directly from its definition. For every $a,b \in \N \setminus \{0\}$:

\begin{enumerate}[A.]
  \item\label{srad:p1} $\srad(a) \leq a$; 
  \item\label{srad:p2} if $a \mid b$ then $\srad(a) \mid \srad(b)$;
  \item\label{srad:p3} $\srad(a^k \cdot b) = \srad(a \cdot b)$;
  \item\label{srad:p4} $\srad(a \cdot b) \mid \srad(a) \cdot \srad(b)$. More interestingly, note that $\srad(a \cdot b) = \srad(\lcm(a,b)) = \frac{\srad(a) \cdot \srad(b)}{\srad(\gcd(a,b))}$.
\end{enumerate}
We also need the notion of the \emph{iterated shifted radical} function~$\israd(x,k)$:
\begin{align*}
  \israd(x,0) &\coloneqq 1\\
  \israd(x,k+1) &\coloneqq \srad(x \cdot \israd(x,k)).
\end{align*}
Regarding this function, we need the following two properties (\Cref{israd:p1,israd:p2}).

\begin{lemma}
  \label{israd:p1}
  For all $n \in \N \setminus \{0\}$ and $k \in \N$, 
  $\israd(n,k)$ divides $\israd(n,k+1)$. 
\end{lemma}

\begin{proof}
  By induction on $k$.
  \begin{description}
    \item[base case $k = 0$:] $\israd(n,0) = 1$ divides $\srad(n) = \israd(n,1)$.
    \item[induction step $k \geq 1$:] 
      \begin{align*}
      \israd(n,k) &= \srad(n \cdot \israd(n, k-1))\\ 
      &\ \mid\ \srad(n \cdot \israd(n,k)) 
        &\text{by I.H.~and $a \mid b$ implies $\srad(a) \mid \srad(b)$}\\
      & = \israd(n,k+1).
      &&\qedhere
      \end{align*}
  \end{description}
\end{proof}

\begin{lemma}
  \label{israd:p2}
  For all $n \in \N \setminus \{0\}$ and $k \in \N$, 
  $\israd(n,k) \leq n^k$. 
\end{lemma}

\begin{proof}
  By induction on $k$. 
  \begin{description}
    \item[base case $k = 0$:] $\israd(n,0) = 1 = n^0$. 
    \item[induction step $k \geq 1$:] 
    \begin{align*}
      \israd(n,k) &= \srad(n \cdot \israd(n, k-1))\\ 
      &\leq \srad(n) \cdot \srad(\israd(n,k-1)) 
        &\text{by $\srad(a \cdot b) \leq \srad(a) \cdot \srad(b)$}\\
      &\leq n \cdot \israd(n, k-1) \leq n^k
        &\text{by $\srad(a) \leq a$ and I.H.}
      &\qedhere
    \end{align*}
  \end{description}
\end{proof}

The proof of~\Cref{lemma:totient-bound} uses an auxiliary recurrence defined in the next lemma.
\begin{lemma}
  \label{israd:p3}
  Let $\alpha\geq 1$ be in $\N$. Let the integer sequence $d_1,d_2,...$ be defined by the recurrence 
  $d_1\coloneqq \alpha \cdot \srad(\alpha)$ and $d_{i+1}\coloneqq \alpha \cdot d_i \cdot \srad(d_i)$ for $i \geq 1$.
  For all $i\in \N \setminus \{0\}$, $d_i=\alpha^i\cdot\prod_{k=1}^i\israd(\alpha,k)$.   
\end{lemma}
\begin{proof}
  By induction on $i \geq 1$. 
  \begin{description}
    \item[base case $i = 1$.] $d_1 = \alpha\cdot\srad(\alpha) = \alpha^1\cdot\israd(\alpha , 1)$. 
    \item[induction step.] 
      Assume that $d_{i-1} = \alpha^{i-1}\cdot\prod_{k=1}^{i-1}\israd(\alpha, k)$ for some $i > 1$. 
      Then we have:
      \begin{align*} 
        d_{i} 
          & = \alpha\cdot d_{i-1} \cdot \srad(d_{i-1}) = \alpha^i \cdot \prod_{k=1}^{i-1}\israd(\alpha,k)\cdot\srad(\alpha^{i-1}\cdot\prod_{k=1}^{i-1}\israd(\alpha,k))
          & \text{by I.H.}\\
          & = \alpha^i \cdot \prod_{k=1}^{i-1}\israd(\alpha,k)\cdot \srad(\alpha\cdot\prod_{k=1}^{i-1}\israd(\alpha,k))
          & \text{by $\srad(a^k\cdot b)=\srad(a\cdot b)$}\\
          & = \alpha^i \cdot \prod_{k=1}^{i-1}\israd(\alpha,k)\cdot \srad(\alpha\cdot\israd(\alpha,i-1))
          & \text{by \Cref{israd:p1}} \\ 
          & = \alpha^i \cdot \prod_{k=1}^{i-1}\israd(\alpha,k)\cdot \israd(\alpha,i)
          & \text{by def. of $\israd(x,n)$} \\
          & = \alpha^i \cdot \prod_{k=1}^{i}\israd(\alpha,k).
          &\qedhere
      \end{align*}
  \end{description}
\end{proof}

We are now ready to prove~\Cref{lemma:totient-bound}.

\LemmaTotientBound*

\begin{proof}
  We prove that for every $i \in \N \setminus \{0\}$, $b_i$ divides $d_i$ (defined in~\Cref{israd:p3}).
  
  The proof is by induction on $i \geq 1$.

  \begin{description}
    \item[base case $i = 1$.] We have $b_1 = \lcm(b_0,\totient(\alpha \cdot b_0)) = \lcm(1,\totient(\alpha \cdot 1)) = \totient(\alpha) 
    = \frac{\alpha}{\rad(\alpha)} \cdot \srad(\alpha)$. Therefore, $b_1$ divides $\alpha \cdot \srad(\alpha)=d_1$, proving the base case. 
    \item[induction step.] 
      Assume that $b_i$ divides $d_i$. 
      A well-known property of Euler's totient function applied to the product $\alpha \cdot d_i$
      gives us the equation
      \begin{equation}
        \label{eq:totiemt:lcm-gcd}
        \totient(\alpha \cdot d_i) = \gcd(\alpha, d_i) \cdot \totient(\lcm(\alpha, d_i)) = \alpha \cdot \totient(d_i),
      \end{equation}
      where the second equality follows from the fact that $\alpha$ is a divisor of $d_i$.
      Now we show that $b_{i+1}$ divides $d_{i+1}$. 
    \end{description}
      \begin{align*}
        b_{i+1}
        &=
        \lcm(b_i,\, \totient(\alpha \cdot b_i))\\
        &\ \mid \ \lcm(d_i,\, \totient(\alpha \cdot d_i))
        &\text{by I.H.~and $a \mid b$ implies $\totient(a) \mid \totient(b)$}\\
        & = \lcm(d_i,\, \alpha \cdot \totient(d_i))
        &\text{by \Cref{eq:totiemt:lcm-gcd}}\\
        &\ \mid \ \alpha \cdot\lcm(d_i,\, \totient(d_i))\\
        & = \alpha \cdot \lcm(d_i,\, \frac{d_i}{\rad(d_i)}\cdot \srad(d_i))
        &\text{def.~of~$\totient$}\\
        &\ \mid \ \alpha \cdot d_i\cdot \srad(d_i) = d_{i+1}.
      \end{align*}

  Therefore, we conclude that for every $i \geq 1$, $b_i$ is bounded as: 
  \begin{align*}
    b_i \leq d_i &= \alpha^i\cdot\prod_{k=1}^i\israd(\alpha,k)
    &\text{by~\Cref{israd:p3}}\\
    &\leq \alpha^i\cdot \israd(a,i)^i\leq \alpha^i \cdot \alpha^{i^2} \leq \alpha^{2 \cdot i^2}.
    &\text{by~\Cref{israd:p1,israd:p2}}
    &\qedhere
  \end{align*}
\end{proof}

We are now ready to complete the complexity analysis of~\Cref{algo:master}.

\PropMasterInNp* 

\begin{proof}
  Consider the execution of~\Cref{algo:master}
  on an input~$\phi(x_1,\dots,x_n)$, with $n \geq 1$. For $i \in [0,n]$, let $(\phi_i,\theta_i)$ be the pair of system and ordering
  obtained after the $i$th iteration of the \textbf{while} loop of~line~\ref{algo:master:outer-loop}, 
  where $\phi_0 = \phi$ and $\theta_0$ is the ordering guessed in line~\ref{algo:master:guess-order}. 

  By virtue of~\Cref{lemma:bounds-one-loop}, 
  in order to prove the proposition it suffices to show that, for every $i \in [0,n]$, the formula $\phi_i$ is of bit size polynomial in the bit size of the input formula $\phi$. 
  Indeed, recall that the \textbf{while} loop of~line~\ref{algo:master:outer-loop}
  is executed $n$ times, 
  and during its $i$th iteration, the loop body runs in non-deterministic polynomial time in the bit size of $\phi_i$.

  To prove that the bit size of $\phi_i$ 
  is polynomial in the bit size of $\phi$, it suffices to iterate $i$ times the bounds stated in~\Cref{lemma:bounds-one-loop}. That is, we establish the following parameter table:
  \begin{center}
    \renewcommand\arraystretch{1.3}
    \setlength{\tabcolsep}{3pt}
    \begin{tabular}{|g|c|c|c|c|c|c|}
        \hline
        \rowcolor{light-gray}
        & $\card\lst( \cdot , \cdot )$
        & $\card( \cdot )$ 
        & $\linnorm{ \cdot }$
        & $\fmod( \cdot )$ 
        & $\onenorm{ \cdot }$\\
        \hline
        \hline
        $\phi$
        & $\theta \colon \ \ell$
        & $s$ 
        & $a$
        & $1$
        & $c$
        \\
        \hline
        \hline
        $\phi_i$
        & $\theta_i \colon \ \ell + 5 \cdot i^2$
        & $s + 16 \cdot i \cdot (i+2)^2 + 2 \cdot i \cdot \ell$
        & $3^i a$
        & $\beta_i$
        & $\tocheck{(i+3)^{8 \cdot i} (\log(a+1) + c)}$
        \\ 
    \hline
    \end{tabular}
  \end{center}
  where $\beta_i \in [1,\tocheck{(3^i \cdot a + 2)^{18 \cdot i^5}}]$.

  Observe that the parameter table
  states that $\card{\phi_i}$, $\log(\onenorm{\phi_i})$ and $\log(\fmod(\phi_i))$ are bounded polynomially in the bit size of $\phi$.
  This does indeed imply that the bit size of $\phi_i$ is polynomial in the bit size of $\phi$, since the former is in 
  $O(\card\phi_i \cdot n^2 \cdot \log(\onenorm{\phi_i}) \cdot \log(\fmod(\phi_i)))$.

  To show that the parameter table above is correct,
  we first prove the $\linnorm{\cdot}$ column
  by induction on $i$. 
  Then, we switch to $\fmod(\cdot)$, which is proved by using~\Cref{lemma:totient-bound}. 
  Lastly, we tackle the cases of the remaining parameters by induction on $i$.

  \begin{description}
    \item[Linear norm.] 
      The base case is simple: $\linnorm{\phi_0} = \linnorm{\phi} \leq a$. 
      In the induction step, assume $\linnorm{\phi_i} \leq 3^i a$ by induction hypothesis ($i \geq 0$).
      By~\Cref{lemma:bounds-one-loop}, 
      $\linnorm{\phi_{i+1}} \leq 3^{i+1} a$.  
    \item[Least common multiple of the divisibility constraints.] 
      First, let us bound the value of $\alpha_i$ 
      in the statement of~\Cref{lemma:bounds-one-loop}. In that lemma, 
      $\alpha_i$ is stated to be in $[1,\tocheck{(3 \cdot a' + 2)^{(i+3)^2}}]$ where $a'$ is an upper bound to $\linnorm{\phi_i}$. 
      Hence, from the previous step of the proof, 
      we have $\alpha_i \in [1,\tocheck{(3^{i+1} a + 2)^{(i+3)^2}}]$. Now, following again the statement of~\Cref{lemma:bounds-one-loop}, we conclude that there is a sequence of integers $b_0,b_1,\dots$ such that 
      \begin{align*} 
        \fmod(\phi_0) &\divides b_0 &b_0 &= 1\\
        \fmod(\phi_{i+1}) &\divides b_{i+1} &b_{i+1} &= \lcm(b_i,\totient(\alpha_i \cdot b_i)).
      \end{align*}
      We almost have the sequence in~\Cref{lemma:totient-bound}, the only difference being that the value $\alpha_i$ changes for each $b_i$.
      Let $\overline{\alpha}_0 \coloneqq 1$ and for every positive 
      $j \in \N$, let $\overline{\alpha}_j \coloneqq \lcm(\alpha_0,\dots,\alpha_{j-1})$. For every $i \in \N$, we consider the sequence of $i+1$ terms $c_0,\dots,c_{i}$ given by 
      \[
        c_0 \coloneqq 1, \qquad c_{j+1} \coloneqq \lcm(c_j,\totient(\overline{\alpha}_{i}\cdot c_j)) \ \ \text{for } j \in [0,i-1].
      \]
      By~\Cref{lemma:totient-bound}, we conclude that $c_i \leq (\overline{\alpha}_{i})^{2 \cdot i^2}$. 

      Let us show that, for every $j \in [0,i]$, 
      $b_j \divides c_j$.
      This is done with a simple induction hypothesis. The base case for $b_0$ and $c_0$ is trivial. For the induction step, pick $j \in [0,i-1]$. We have:
      \begin{align*}
        b_{j+1} &= \lcm(b_j,\totient(\alpha_j \cdot b_j))\\ 
        &\ \divides\ \lcm(c_j,\totient(\alpha_j \cdot c_j))
        &\text{by~I.H.~and from ``$a \divides b$ implies $\totient(a) \divides \totient(b)$''}\\
        &\ \divides \lcm(c_j, \totient(\overline{\alpha}_{i} \cdot c_j))
        &\text{$\alpha_{j}$ divides $\overline{\alpha}_{i}$, and again ``$a \divides b$ implies $\totient(a) \divides \totient(b)$''}\\ 
        & = c_{j+1}.
      \end{align*}
      Then, from the definition of $\overline{\alpha}_i$ and the bound $c_i \leq (\overline{\alpha}_{i})^{2 \cdot i^2}$, 
      we conclude that $\fmod(\phi_i)$ divides some $\beta_i \in [1,\tocheck{(3^i \cdot a + 2)^{18 \cdot i^5}}]$.
  \end{description}

  We now show by induction on $i$ the bounds on the remaining three parameters in the table.
  The base case $i = 0$ is trivial, as $\phi_0 = \phi$. For the induction step, assume that the bounds in the table are correct for $\phi_i$. 
  We show that then $\phi_{i+1}$ follows its respective table.
  \begin{description}
    \item[Least significant terms.] By induction hypothesis, $\card{\lst(\phi_i,\theta_i)} \leq \ell + 5 \cdot i^2$. 
    We apply the bounds in~\Cref{lemma:bounds-one-loop}, and obtain  
      \begin{align*} 
        \card{\lst(\phi_{i+1},\theta_{i+1})} 
        & \leq (\ell + 5 \cdot i^2) + 2\cdot (i+2)\\ 
        & \leq \ell + 5 \cdot i^2 + 2\cdot i + 4\\ 
        & \leq \ell + 5 \cdot (i^2 + 2\cdot i + 1)\\
        & \leq \ell + 5 \cdot (i+1)^2.
      \end{align*}
    \item[Number of constraints.] 
        By induction hypothesis $\card\phi_i \leq s + 16 \cdot i \cdot (i+2)^2 + 2 \cdot i \cdot \ell$ and, as in the previous case, $\card{\lst(\phi_i,\theta_i)} \leq \ell + 5 \cdot i^2$.
        We apply the bounds in~\Cref{lemma:bounds-one-loop}, and obtain  
        \begin{align*} 
          \card\phi_{i+1}
          & \leq (s + 16 \cdot i \cdot (i+2)^2 + 2 \cdot i \cdot \ell) + 6(i+2) + 2 \cdot (\ell + 5 \cdot i^2)\\
          &    = s + (16 \cdot i \cdot (i+2)^2 + 6(i+2) + 10 \cdot i^2) + 2 \cdot (i+1) \cdot \ell\\ 
          & \leq s + (16 \cdot i \cdot (i+2)^2 + 16(i+2)^2) + 2 \cdot (i+1) \cdot \ell\\
          & \leq s + 16 \cdot (i+1) \cdot (i+3)^2 + 2 \cdot (i+1) \cdot \ell.  
        \end{align*}
    \item[$1$-norm.] 
    By applying the induction hypothesis, we have $\onenorm{\phi_i} \leq \tocheck{(i+3)^{8 \cdot i} (\log(a+1) + c)}$ and ${\fmod(\phi_i) \leq \tocheck{(3^i \cdot a + 2)^{18 \cdot i^5}}}$. Then, following~\Cref{lemma:bounds-one-loop}, 
    \begin{align*}
      \onenorm{\phi_{i+1}}
      \leq{}& 
      \tocheck{2^5(i+3)^2\big((i+3)^{8 \cdot i} (\log(a+1) + c)+2\big)} + 4 \cdot \log(\tocheck{(3^i \cdot a + 2)^{18 \cdot i^5}})\\ 
      \leq{}& \tocheck{
        2^5(i+3)^{8i+2} (\log(a+1) + c)
      } \, + \, \tocheck{2^6(i+3)^2} \, + \, \tocheck{4 \cdot 18 \cdot 3 \cdot i^6 \log(a+1)}\\
      &\hspace{4.4cm}\text{note: $\log(3a+2) \leq 3\log(a+1)$ for $a \geq 1$}\\
      \leq{}&  
        \tocheck{
        3 \cdot 2^5(i+3)^{8i+2} (\log(a+1) + c)
        }\\
      &\hspace{3.8cm} \text{last two summands are smaller than the first}\\
      \leq{}& \tocheck{
        (i+3)^{8(i+1)} (\log(a+1) + c)
      }.
      &&\qedhere
    \end{align*} 
  \end{description}
\end{proof}

\end{document}